\def\draft{1}

\documentclass[11pt]{article}
\usepackage[utf8]{inputenc}
\usepackage{fullpage}
\usepackage{amsthm}
\usepackage{mathtools,amsmath,amssymb}
\usepackage{bm}
\usepackage{enumitem}
\usepackage{mathtools}
\usepackage[colorlinks=true, allcolors=blue]{hyperref}
\usepackage{complexity}
\usepackage{dsfont}
\usepackage{float}

\newcommand{\One}{\mathds{1}}
\newcommand{\GOOD }{\textsf{GOOD}}
\renewcommand{\epsilon}{\varepsilon}
\newcommand{\eps}{\varepsilon}

\usepackage[breakable]{tcolorbox}
\newtcolorbox{reduction}[2][]
{
  breakable,
  colframe = gray!50,
  colback  = gray!10,
  coltitle = gray!10!black,
  before skip = 10pt,
  after skip = 10pt,
  title    = \textbf{#2},
  #1,
}

\newtcolorbox{examplebox}[2][]
{
  breakable,
  colframe = gray!50,
  colback  = gray!10,
  coltitle = gray!10!black,
  before skip = 10pt,
  after skip = 10pt,
  title    = \textbf{#2},
  #1,
}

\usepackage{algorithmicx}
\usepackage{algorithm}
\usepackage[noend]{algpseudocode}
\newcounter{algsubstate}
\renewcommand{\thealgsubstate}{\alph{algsubstate}}

\algnewcommand\algorithmicinput{\textbf{Input:}}
\algnewcommand\Input{\item[\algorithmicinput]}

\algnewcommand\algorithmicoutput{\textbf{Output:}}
\algnewcommand\Output{\item[\algorithmicoutput]}

\algnewcommand\algorithmicgoal{\textbf{Goal:}}
\algnewcommand\Goal{\item[\algorithmicgoal]}

\newcommand{\mnote}[1]{\ifnum\draft=1 {\color{red} \textbf{Madhu's note:} #1}\fi}
\newcommand{\cnote}[1]{\ifnum\draft=1 {\color{brown} \textbf{Chi-Ning's note:} #1}\fi}
\newcommand{\snote}[1]{\ifnum\draft=1 {\color{orange} \textbf{Sasha's note:} #1}\fi}
\newcommand{\vnote}[1]{\ifnum\draft=1 {\color{green!40!black} \textbf{Santhoshini's note:} #1}\fi}
\newcommand{\todo}[1]{\ifnum\draft=1{\color{violet} \textbf{TODO:} #1} \fi}

\usepackage{amsthm}
\usepackage{thmtools,thm-restate}

\numberwithin{equation}{section}
\declaretheoremstyle[bodyfont=\it,qed=\qedsymbol]{noproofstyle}

\declaretheorem[name=Observation,numbered=no]{observation*}

\declaretheorem[numberlike=equation]{theorem}

\declaretheorem[name=Theorem,numbered=no]{theorem*}

\declaretheorem[numberlike=equation]{lemma}
\declaretheorem[name=Lemma,numbered=no]{lemma*}

\declaretheorem[numberlike=equation]{corollary}
\declaretheorem[name=Corollary,numbered=no]{corollary*}

\declaretheorem[numberlike=equation]{proposition}
\declaretheorem[name=Proposition,numbered=no]{proposition*}

\declaretheorem[numberlike=equation]{claim}
\declaretheorem[name=Claim,numbered=no]{claim*}

\declaretheorem[name=Conjecture,numbered=no]{conjecture*}

\declaretheorem[name=Question,numbered=no]{question*}

\declaretheoremstyle[bodyfont=\it]{defstyle} 

\declaretheorem[numberlike=equation,style=defstyle]{definition}
\declaretheorem[unnumbered,name=Definition,style=defstyle]{definition*}

\declaretheorem[unnumbered,name=Example,style=defstyle]{example*}

\declaretheorem[unnumbered,name=Notation=defstyle]{notation*}

\declaretheorem[unnumbered,name=Construction,style=defstyle]{construction*}

\declaretheoremstyle[]{rmkstyle}

\declaretheorem[numberlike=equation]{remark}
\declaretheorem[name=Remark,numbered=no]{remark*}

\DeclarePairedDelimiter\ceil{\lceil}{\rceil}

\newcommand{\opt}{\mathop{\mathrm{opt}}}
\newcommand{\gbmaxf}{{(\gamma,\beta)\textrm{-}\maxf}}

\newcommand{\Exp}{\mathop{\mathbb{E}}}

\newcommand{\cA}{\mathcal{A}}
\newcommand{\cB}{\mathcal{B}}
\newcommand{\cD}{\mathcal{D}}
\newcommand{\cF}{\mathcal{F}}
\newcommand{\cN}{\mathcal{N}}
\newcommand{\cY}{\mathcal{Y}}
\newcommand{\N}{\mathbb{N}}
\newcommand{\F}{\mathbb{F}}
\newcommand{\Z}{\mathbb{Z}}
\renewcommand{\R}{\mathbb{R}}
\newcommand{\bern}{\mathsf{Bern}}
\newcommand{\unif}{\mathsf{Unif}}
\newcommand{\maxf}{\textsf{Max-CSP}(f)}
\newcommand{\bias}{\textsf{bias}}
\newcommand{\val}{\textsf{val}}

\newcommand{\ALG}{\mathbf{ALG}}
\newcommand{\supp}{\textsf{supp}}
\newcommand{\yes}{\textbf{YES}}
\newcommand{\no}{\textbf{NO}}

\newcommand{\maxtwoand}{\textsf{Max-2AND}}
\newcommand{\RMD}{\textsf{RMD}}
\newcommand{\sRMD}{\textsf{streaming-RMD}}
\newcommand{\psRMD}{\textsf{padded-streaming-RMD}}
\newcommand{\pstrm}{\text{pad-stream}}

\newcommand{\veca}{\mathbf{a}}
\newcommand{\vecb}{\mathbf{b}}
\newcommand{\vecc}{\mathbf{c}}
\newcommand{\vece}{\mathbf{e}}
\newcommand{\vecj}{\mathbf{j}}
\newcommand{\vecs}{\mathbf{s}}
\newcommand{\vecu}{\mathbf{u}}
\newcommand{\vecv}{\mathbf{v}}
\newcommand{\vecw}{\mathbf{w}}
\newcommand{\vecx}{\mathbf{x}}
\newcommand{\vecy}{\mathbf{y}}
\newcommand{\vecz}{\mathbf{z}}

\newcommand{\vecsigma}{\boldsymbol{\sigma}}
\newcommand{\veclambda}{\boldsymbol{\lambda}}
\newcommand{\vecmu}{\boldsymbol{\mu}}
\newcommand{\vecnu}{\boldsymbol{\nu}}

\newcommand{\COMP}{\textsf{COMP}}
\newcommand{\COMB}{\textsf{COMB}}
\newcommand{\wtS}{\widetilde{S}}

\title{Approximability of all Boolean CSPs with linear sketches\footnote{This paper replaces the paper \cite{CGSV20} by the authors. The previous version had errors and is now withdrawn.}}

\author{
Chi-Ning Chou\thanks{School of Engineering and Applied Sciences, Harvard University, Cambridge, Massachusetts, USA. Supported by NSF awards CCF 1565264 and CNS 1618026. Email: \texttt{chiningchou@g.harvard.edu}.}
\and Alexander Golovnev\thanks{Department of Computer Science, Georgetown University. Email: \texttt{alexgolovnev@gmail.com}.}
\and Madhu Sudan\thanks{School of Engineering and Applied Sciences, Harvard University, Cambridge, Massachusetts, USA. Supported in part by a Simons 
Investigator Award and NSF Award CCF 1715187. Email: \texttt{madhu@cs.harvard.edu}.}
\and Santhoshini Velusamy\thanks{School of Engineering and Applied Sciences, Harvard University, Cambridge, Massachusetts, USA. Supported in part by a Simons 
Investigator Award and NSF Award CCF 1715187. Email: \texttt{svelusamy@g.harvard.edu}.}
}

\begin{document}
\date{}
\sloppy
\maketitle
\begin{abstract}
A Boolean constraint satisfaction problem (CSP), $\maxf$, is a maximization problem specified by a constraint $f:\{-1,1\}^k\to\{0,1\}$. An instance of the problem consists of $m$ constraint applications on $n$ Boolean variables, where each constraint application applies the constraint to $k$ literals chosen from the $n$ variables and their negations. The goal is to compute the maximum number of constraints that can be satisfied by a Boolean assignment to the $n$~variables. In the $(\gamma,\beta)$-approximation version of the problem for parameters $\gamma \geq \beta \in [0,1]$, the goal is to distinguish instances where at least $\gamma$ fraction of the constraints can be satisfied from instances where at most $\beta$ fraction of the constraints can be satisfied. 

In this work we consider the approximability of $\maxf$ in the context of sketching algorithms and completely characterize the approximability of all Boolean CSPs. Specifically, given $f$, $\gamma$ and $\beta$ we show that either (1) the $(\gamma,\beta)$-approximation version of $\maxf$ has a linear sketching algorithm using $O(\log n)$ space, or (2) for every $\epsilon > 0$ the $(\gamma-\epsilon,\beta+\epsilon)$-approximation version of $\maxf$ requires $\Omega(\sqrt{n})$ space for any sketching algorithm. We also prove lower bounds against streaming algorithms for several CSPs. In particular, we recover the streaming dichotomy of \cite{CGV20} for $k=2$ and show streaming approximation resistance of all CSPs for which $f^{-1}(1)$ supports a distribution with uniform marginals. 
%Previously such a separation was known only for $k=2$. We stress that for $k=2$, there are only finitely many distinct problems to consider.

Our positive results show wider applicability of bias-based algorithms used previously by \cite{GVV17} and
\cite{CGV20} by giving a systematic way to discover biases. Our negative results combine the Fourier analytic methods of \cite{KKS}, which we extend to a wider class of CSPs, with a rich collection of reductions among communication complexity problems that lie at the heart of the negative results.  

%\mnote{Abstract modified.}

\end{abstract}
%\todo{Here is what I think needs to be done in the paper: Main thing is Section 5. Here after introducing RMD and stating the bound we will prove in sections 6 and 7, we should have a subsection (5.2) on simultaneous and padded-streaming-RMD. We should then state and prove the current Lemma 5.13 (csp-value). Then section (5.3) can be whatever is needed for the dynamic lower bound and section (5.4) can be the insertion only lower bound. (All this is just shfiting information around.)       -Madhu.}

\newpage 
\tableofcontents

\newpage

\section{Introduction}

In this paper we give a complete characterization of the approximability of Boolean constraint satisfaction problems (CSPs) by sketching algorithms. We describe the exact class of problems below, and give a brief history of previous work before giving our results.

\subsection{Boolean CSPs}

%\mnote{Added weights.}

In this paper we use $\N$ to denote the set of natural numbers $\{1,2,3,\ldots\}$. For $n \in \N$ we use $[n]$ to denote the set $\{1,2,\ldots,n\}$. We refer to a variable taking values in $\{-1,1\}$ as a Boolean variable. Given a Boolean variable $X$, we refer to $X$ and $-X$ as the literals associated with $X$. For vectors $\veca,\vecb \in \R^n$ we use $\veca \odot \vecb$ to denote their coordinate-wise product. I.e., if $\veca = (a_1,\ldots,a_n)$ and $\vecb = (b_1,\ldots,b_n)$ then 
$\veca \odot \vecb = (a_1b_1,\ldots,a_nb_n)$.

In this paper, a Boolean CSP is a maximization problem, $\maxf$, specified by a single constraint function $f:\{-1,1\}^k \to \{0,1\}$ for some positive integer $k$. Given $n$ Boolean variables $x_1,\ldots,x_n$, an application of the constraint function $f$ to these variables, which we term simply a {\em constraint}, is given by two $k$-tuples $\vecj = (j_1,\ldots,j_k) \in [n]^k$ and $\vecb = (b_1,\ldots,b_k)\in \{-1,1\}^k$ where the $j_i$'s are distinct, and represents the application of the constraint function $f$ to the literals $b_1x_{j_1},\ldots,b_kx_{j_k}$. 
Specifically an assignment $\vecsigma = (\sigma_1,\ldots,\sigma_n) \in \{-1,1\}^n$ satisfies a constraint given by $(\vecj,\vecb)$ if $f(b_1\sigma_{j_1},\ldots,b_k\sigma_{j_k})=1$. For a constraint $C = (\vecj,\vecb)$ and assignment $\vecsigma$ we use $\vecsigma|_\vecj$ as shorthand for $(\sigma_{j_1},\ldots,\sigma_{j_k})$ and $C(\vecsigma)$ as shorthand for $f(\vecb \odot \vecsigma|_\vecj) = f(b_1\sigma_{j_1},\ldots,b_k\sigma_{j_k})$.
An instance~$\Psi$ of weighted $\maxf$ consists of $m$ constraints $C_1,\ldots,C_m$  applied to $n$ variables $x_1,\ldots,x_n$, along with $m$ non-negative weights $w_1,\ldots,w_m$. The value of an assignment $\vecsigma \in \{-1,1\}^n$ on an instance $\Psi = ((C_1,w_1),\ldots,(C_m,w_m))$, denoted $\val_\Psi(\vecsigma)$, is the fraction of weight of constraints satisfied by $\vecsigma$, i.e., $\val_\Psi(\vecsigma) = \frac{\sum_{i \in [m]} w_i C_i(\vecsigma)}{\sum_{i \in [m]} w_i}$.  The goal of the {\em exact} problem is to compute the maximum, over all assignments, of the value of the assignment on the input instance, i.e., to compute, given $\Psi$, the quantity $\val_\Psi = \max_{\vecsigma \in \{-1,1\}^n}\{\val_\Psi(\vecsigma)\}$.
\footnote{We note that the literature on CSPs has several generalizations: one may allow an entire set of constraint functions, not just a single one. One may restrict the constraint applications to be applicable only to variables and not literals. And finally one can of course consider non Boolean CSPs. We do not do any of those in this paper, though extending our techniques to classes of functions seems immediately feasible. See more discussion in \autoref{ssec:future}.}

In this work we consider the approximation version of $\maxf$, which we study in terms of the ``gapped promise problems''. Specifically given $0 \leq \beta < \gamma \leq 1$, the $(\gamma,\beta)$-approximation version of $\maxf$, abbreviated $(\gamma,\beta)\textrm{-}\maxf$, is the task of distinguishing between instances from $\Gamma = \{\Psi | \opt(\Psi) \geq \gamma\}$ and instances from $B = \{\Psi | \opt(\Psi) \leq \beta\}$. It is well-known that this distinguishability problem is a refinement of the usual study of approximation which usually studies the ratio of $\gamma/\beta$ for tractable versions of $\gbmaxf$. See \autoref{prop:approx-equivalence} for a formal statement in the context of streaming approximability of $\maxf$ problems.

\subsection{Streaming algorithms}

We study the complexity of $(\gamma,\beta)$-$\maxf$ in the setting of randomized streaming algorithms. Here, an instance $\Psi = (C_1,\ldots,C_m)$ is presented as a stream $\sigma_1,\sigma_2,\ldots,\sigma_m$ with $\sigma_i = (\vecj_i,\vecb_i)$ representing the $i$th constraint.
We study the space required to solve the $(\gamma,\beta)$-approximation version of $\maxf$. Specifically we consider algorithms that are allowed to use internal randomness and $s$~bits of space. The algorithms output a single bit at the end. They are said to solve the $(\gamma,\beta)$-approximation problem correctly if they output the correct answer with probability at least $2/3$ (i.e., they err with probability at most $1/3$).

The main focus of this work is sketching algorithms, a special class of streaming algorithms, where the algorithm's output is determined by a small sketch it produces of the input stream, and the sketch itself has the property that the sketch of the concatenation of two streams can be computed from the sketches of the two component streams. (See \autoref{def:sketching alg} for a formal definition.) We define the space of the sketching algorithm to be the length of the sketch.

\iffalse{}We consider algorithms for the insertion-only as well as dynamic (allowing insertions and deletions) streaming settings. In both settings, the algorithm is assumed to know $n$ initially. In the insertion-only setting, a sequence of constraints $C_1,\ldots,C_m$ arrive one at a time (each with unit weight), leading to the input $\Psi = (C_1,\ldots,C_m)$. In the  dynamic setting we follow the ``strict turnstile'' setting: So constraints may be inserted with potential repetitions and deleted, one at a time, till the stream ends to produce a weighted instance $\Psi = ((C_1,w_1), \ldots, (C_m,w_m))$. The ``strict turnstile'' setting requires that at all times the weight of every constraint is non-negative (and so constraints are deleted only after they are inserted). Our algorithms are allowed to use internal randomness and $s$~bits of space. The algorithms output a single bit at the end. They are said to solve the $(\gamma,\beta)$-approximation problem correctly if they output the correct answer with probability at least $2/3$ (i.e., they err with probability at most $1/3$). 
\fi

Our main dividing line is between algorithms that work with space $O(\poly\log n)$, versus algorithms that require space at least $n^\epsilon$ for some $\epsilon > 0$. In informal usage we refer to a streaming problem as ``easy'' if it can be solved with polylogarithmic space (the former setting) and ``hard'' if it requires polynomial space (the latter setting). We note that all the positive results (algorithms) given in this paper are linear sketching algorithms which are more restrictive than general sketching algorithms.
We also note that many of our lower bounds work against general streaming algorithms and we elaborate on this in \autoref{ssec:results}.

\subsection{Past work}

To the best of our knowledge, streaming algorithms for Boolean CSPs have not been investigated extensively. Here we cover the few results we are aware of. %We note that all previous works focused explicitly only on the insertion-only setting. In particular all negative results hold in the insertion-only setting, and while algorithmic results are also only stated for this setting, they actually  hold in the more general fully dynamic setting. 
%\mnote{Santhoshini - is the previous statement correct?}\vnote{Yes, looks correct to me.} 
On the positive side, it may be surprising that there exists any non-trivial algorithm at all. Here, and later, we describe algorithms solving the $(1,\rho(f)-\epsilon)$-approximation problem for $\eps>0$  as ``trivial'', where $\rho(f) = 2^{-k} \sum_{\veca \in \{-1,1\}^k} f(\veca)$
is the fraction of clauses satisfied by a random assignment. Note that the algorithm that always outputs $1$ correctly solves the $(1,\rho(f)-\eps)$-approximation version of the $\maxf$ problem.

It turns out that there do exist some non-trivial approximation algorithms for Boolean CSPs.  This was established by the work of Guruswami, Velingker, and Velusamy~\cite{GVV17} who, in our notation, gave an algorithm for the $(\gamma,2\gamma/5-\epsilon)$-approximation version of $\textsf{Max-2AND}$, for every $\gamma \in [0,1]$ ($\maxtwoand$ is the $\maxf$ problem corresponding to $f(a,b) = 1$ if $a = b = 1$ and $0$ otherwise). A central ingredient in their algorithm is the ability of streaming algorithms to approximate the $\ell_1$ norm of a vector in the turnstile setting, which allows them to estimate the ``bias'' of $n$ variables (how often they occur positively in constraints, as opposed to negatively). Subsequently, the work of Chou, Golovnev, and Velusamy~\cite{CGV20} further established the utility of such algorithms, which we refer to as bias-based algorithms, by giving optimal algorithms for all Boolean CSPs on $2$ variables. In particular they give a better (optimal!) analysis of bias-based algorithms for $\textsf{Max-2AND}$, and show that $\textsf{Max-2SAT}$ also has an optimal algorithm based on bias. We note that $\textsf{Max-2SAT}$ is again not covered by the results of the current paper since it involves two functions corresponding to clauses of length 1, and clauses of length 2.

On the negative side, the problem that has been explored the most is $\textsf{Max-CUT}$, or in our language $\textsf{Max-2XOR}$, which corresponds to $f(x,y) = x \oplus y = (1-xy)/2$.\footnote{Strictly speaking this work does not include $\textsf{Max-CUT}$, which does not allow constraints to be placed on arbitrary literals. $\textsf{Max-2XOR}$ is however very closely related and in particular is harder than $\textsf{Max-CUT}$.}
Kapralov, Khanna, and Sudan~\cite{KKS} showed that $\textsf{Max-2XOR}$ does not have a $(1,1/2 + \epsilon)$-approximation algorithm using $o(\sqrt{n})$-space, for any $\epsilon > 0$.  This was subsequently improved upon by Kapralov, Khanna, Sudan, and Velingker~\cite{KKSV17}, and Kapralov and Krachun~\cite{KK19}. The final paper~\cite{KK19} completely resolves $\textsf{Max-CUT}$ and $\textsf{Max-2XOR}$ showing that $(1,1/2 + \epsilon)$-approximation for these problems requires $\Omega(n)$ space. Turning to other problems, the work by \cite{GVV17} notices that the $(1,1/2+\epsilon)$-inapproximability of $\textsf{Max-2XOR}$ immediately yields $(1,1/2+\epsilon)$-inapproximability of $\maxtwoand$ as well. In \cite{CGV20} more sophisticated reductions are used to improve the inapproximability result for $\maxtwoand$ to a $(\gamma,4\gamma/9 + \epsilon)$-inapproximability for some positive $\gamma$, which turns out to be the optimal ratio by their algorithm and analysis. As noted earlier their work gives optimal algorithms for all functions $f:\{-1,1\}^2 \to \{0,1\}$.

\subsection{Our results}\label{ssec:results}

Our main theorem is a decidable dichotomy theorem for $(\gamma,\beta)$-$\maxf$ with sketching algorithms.

%\cnote{Update the theorem below to dynamic setting.}
\begin{theorem}\label{thm:main-intro} For every $k\in \N$, for every function $f:\{-1,1\}^k \to \{0,1\}$, and for every $0 \leq \beta < \gamma \leq 1$, at least one of the following always holds:
\begin{enumerate}
    \item $(\gamma,\beta)$-$\maxf$ has a $O(\log n)$-space linear sketching algorithm.
    \item For every $\epsilon > 0$, any sketching algorithm that solves $(\gamma-\epsilon,\beta+\epsilon)$-$\maxf$ requires $\Omega(\sqrt{n})$ space. If $\gamma = 1$, then any sketching algorithm that solves $(1,\beta+\epsilon)$-$\maxf$ requires $\Omega(\sqrt{n})$ space. 
\end{enumerate}
Furthermore, there is an algorithm using space $\poly(2^k,\ell)$ that decides which of the two conditions holds, given the truth-table of $f$, and $\gamma$ and $\beta$ as $\ell$-bit rationals\footnote{$\alpha\in\R$ is said to be an $\ell$-bit rational if there exist integers $-2^\ell < p,q < 2^{\ell}$ such that $\alpha = p/q$.}.
\end{theorem}

In analogy with the terminology used in the study of CSP approximation in polynomial time, we define a problem to be ``approximation-resistant'' if it is hard to beat the random assignment with $n^{o(1)}$-space. 
Recall $\rho(f)$ denotes the fraction of assignments that satisfy a function~$f$. We say that $\maxf$ is {\em approximation-resistant} if, for every $\epsilon > 0$ there exists $\delta > 0$ such that $(1,\rho+\epsilon)$-$\maxf$ requires $\Omega(n^\delta)$ space. %We use this terminology, with appropriate qualifiers, in both the dynamic and the insertion-only settings.  (We suppress the qualifier ``streaming-'' for most of the paper.) 
\\
We get the following dichotomy for approximation-resistance to sketching algorithms.

%\cnote{Update the theorem below to dynamic setting.}
\begin{corollary}\label{cor:approx-res}
For every $k\in \N$, for every function $f:\{-1,1\}^k \to \{0,1\}$, if $\maxf$ is approximation-resistant to sketching algorithms, then for every $\epsilon > 0$, any sketching algorithm that solves $(1,\rho(f)+\epsilon)$-approximation version of $\maxf$ requires $\Omega(\sqrt{n})$ space. If $\maxf$ is not approximation-resistant, then there exists $\epsilon > 0$ such that $(1-\epsilon,\rho(f)+\epsilon)$-$\maxf$ can be solved by a linear sketching algorithm in logarithmic space . Furthermore, given the truth-table of the function $f$, there is an algorithm running in space $\poly(2^k)$ that decides whether or not $\maxf$ is approximation-resistant to sketching algorithms.
\end{corollary}

The results above (and in particular the negative results) apply only to sketching algorithms for streaming CSPs. For general streaming algorithm, we get some partial results. To describe our next result, we define the notion of a function supporting a one-wise independent distribution.

We say that a function $f$ {\em supports one-wise independence} if there exists a distribution $D$ supported on the satisfying assignments to $f$, i.e., on $f^{-1}(1) \subseteq \{-1,1\}^k$ such that its marginals are all uniform, i.e., for every $j \in [k]$, we have $\Exp_{\veca \sim D} [a_j] = 0$.

\begin{theorem}\label{thm:one-wise}
If $f:\{-1,1\}^k \to \{0,1\}$ supports one-wise independence then $\maxf$ is approximation resistant in the streaming setting.
\end{theorem}

We also give a (very) partial converse, showing that symmetric functions are approximation resistant if and only if they support one-wise independence  (see \autoref{prop:symmetric}).

While we do believe that there are other approximation-resistant problems in the streaming setting, we do not know of one even approximation-resistant to sketching algorithms (and in particular do not give one in this paper). We discuss this more in the next section.

We also give theorems capturing hardness in the streaming setting beyond the one-wise independent case. Stating the full theorem requires more notions (see \autoref{ssec:lb-detail-insert}), but as a consequence we get the following extension of the work of \cite{CGV20} who study the setting of $k=2$.

\begin{theorem}\label{thm:main-intro-k=2} For every function $f:\{-1,1\}^2 \to \{0,1\}$, and for every $0 \leq \beta < \gamma \leq 1$, at least one of the following always holds:
\begin{enumerate}
    \item $(\gamma,\beta)$-$\maxf$ has a $O(\log n)$-space linear sketching algorithm.
    \item For every $\epsilon > 0$, every streaming algorithm that solves $(\gamma-\epsilon,\beta+\epsilon)$-$\maxf$ requires $\Omega(\sqrt{n})$ space. If $\gamma = 1$, then $(1,\beta+\epsilon)$-$\maxf$ requires $\Omega(\sqrt{n})$ space. 
\end{enumerate}
Furthermore, there is an algorithm using space $\poly(\ell)$ that decides which of the two conditions holds  given the truth-table of $f$, and $\gamma$ and $\beta$ as $\ell$-bit rationals.
\end{theorem}

This reproduces the results of \cite{CGV20} while giving a more refined picture of the approximability by considering all $\beta < \gamma$. 
In~\autoref{sec:examples}, we show how to apply our theorem above to get a full characterization of the approximation profile of the $\maxtwoand$ problem (i.e., the $\maxf$ problem for $f(x,y) = 1$ if $x = y = 1$ and $0$ otherwise). 

%\mnote{Fix this.}
\paragraph{This version:} This version of the paper replaces a previous version~\cite{CGSV20}. The previous version, now withdrawn, claimed \autoref{thm:main-intro} in the streaming setting, but that version had an error and the status of Theorem 1.1 in \cite{CGSV20} is currently open.

\subsection{Contrast with dichotomies in the polynomial time setting}

The literature on dichotomies of $\maxf$ problems is vast. One broad family of results here \cite{Schaefer,Bulatov,Zhuk} 
considers the exact satisfiability problems (corresponding to distinguishing between instances from $\{\Psi | \opt(\Psi) = 1\}$ and instances from $\{\Psi | \opt(\Psi) < 1\}$.
Another family of results~\cite{raghavendra2008optimal,AustrinMossel,KhotTW} considers the approximation versions of $\maxf$ and gets ``near dichotomies'' along the lines of this paper --- i.e., they either show that the $(\gamma,\beta)$-approximation is easy (in polynomial time), or for every $\epsilon > 0$ the  $(\gamma - \epsilon,\beta+\epsilon)$-approximation version is hard (in some appropriate sense). Our work resembles the latter series of works both in terms of the nature of results obtained, the kinds of characterizations used to describe the ``easy'' and ``hard'' classes, and also in the proof approaches (though of course the streaming setting is much easier to analyze, allowing for much simpler proofs overall). We summarize their results giving comparisons to our theorem and then describe a principal contrast. 

In a seminal work, Raghavendra~\cite{raghavendra2008optimal} gave a characterization of the polynomial time approximability of the $\maxf$ problems based on the unique games conjecture~\cite{Kho02}. Our~\autoref{thm:main-intro} is analogous to his theorem, though restricted to a single function, with Boolean variables, with ability to complement variables. 
A characterization of approximation resistant functions is given by Khot, Tulsiani and Worah~\cite{KhotTW}. Our~\autoref{cor:approx-res} is analogous to this. Austrin and Mossel~\cite{AustrinMossel} show that all functions supporting a pairwise independent distribution are approximation-resistant. Our~\autoref{thm:one-wise} is analogous to this theorem. 

While our results run in parallel to the work on polynomial time approximability our characterizations are not immediately comparable. Indeed there are some significant differences which we highlight below. Of course there is the obvious difference that our negative results are unconditional (and not predicated on a complexity theoretic assumption like the unique games conjecture or \textsf{P}$\ne$\textsf{NP}). But more significantly our characterization is a bit more ``explicit'' than those of \cite{raghavendra2008optimal} and \cite{KhotTW}. 
In particular the former only shows decidability of the problem which take $\epsilon$ as an input (in addition to $\gamma,\beta$ and $f$) and distinguishes  $(\gamma,\beta)$-approximable problems from $(\gamma-\epsilon,\beta+\epsilon)$-inapproximable problems. The running time of their decision procedure grows with $1/\epsilon$. In contrast our distinguishability separates $(\gamma,\beta)$-approximability from ``$\forall \epsilon > 0$, $(\gamma-\epsilon,\beta+\epsilon)$-inapproximability'' --- so our algorithm does not require $\epsilon$ as an input - it merely takes $\gamma,\beta$ and $f$ as input.
Indeed this difference is key to the understanding of approximation resistance. Due to the stronger form of our main theorem (\autoref{thm:main-intro}), our characterization of approximation-resistance to sketching algorithms is explicit (decidable in \textsf{PSPACE}), whereas a decidable characterization of approximation-resistance in the polynomial time setting seems to be still open. 

Our characterizations also seem to differ from the previous versions in terms of the features being exploited to distinguish the two classes. This leads to some strange gaps in our knowledge. For instance, it would be natural to suspect that (conditional) inapproximability in the polynomial time setting should also lead to (unconditional) inapproximability in the streaming setting. But we don't have a formal theorem proving this.\footnote{Of course, if this were false, it would be a breakthrough result giving a polynomial time (even log space) algorithm for the unique games!} One (unfulfilling) consequence of this gap in knowledge is that we do not yet have an approximation-resistant problem, even to sketching algorithms, that is not covered by~\autoref{thm:one-wise}. In the polynomial time setting,  Potechin~\cite{Potechin} gives a balanced linear threshold function that is approximation-resistant. Balanced linear threshold functions do not support one-wise independence and so that function would be a good candidate for a streaming-approximation-resistant function that is not covered by~\autoref{thm:one-wise}.

\subsection{Overview of our analysis}

At the heart of our characterization is a family of linear sketching algorithms for $\maxf$. We will describe this family soon, but the main idea of our proof is that if no algorithm in this family solves $(\gamma,\beta)$-$\maxf$, then we can extract a single pair of instances, roughly a $\gamma$-satisfiable ``yes'' instance and an at most $\beta$-satisfiable ``no'' instance, that certify this inability. We then show how this pair of instances can be exploited as gadgets in a negative result. Up to this part our approach resembles that in \cite{raghavendra2008optimal} (though of course all the steps are quite different). The main difference is that we are able to use the structure of the algorithm and the lower bound construction to show that we can afford to consider only instances on $k$ variables. (This step involves a non-trivial choice of definitions that we elaborate on shortly.) 
This bound on the number of variables allows us to get a very ``decidable'' separation between approximable and inapproximable problems. Specifically we show that distinction between approximable setting and the inapproximable one can be expressed by a quantified formula over the reals with a constant number of quantifiers over $2^k$ variables and equations --- a problem that is known to be solvable in \textsf{PSPACE}. We give more details below.

\paragraph{Bias-based algorithms.} 
%\cnote{Add a remark somewhere to emphasize we can handle dynamic input stream.}
For every $\veclambda = (\lambda_1,\ldots,\lambda_k) \in \R^k$ we define the $\veclambda$-bias measure of an instance $\Psi$ of $\maxf$ as follows. Let $p_{ij}$ denote the number of occurrences of the literal $x_i$ as the $j$th variable in a constraint, and let $n_{ij}$ denote the same quantity for the literal $-x_i$. Let $\bias_{i,j} = \frac{1}m (p_{ij} - n_{ij})$. We define the $\veclambda$-bias of the $i$th variable to be a weighted sum of $\bias_{i,j}$ as follows: $\bias_{\veclambda} (\Psi)_i = \sum_{j=1}^k \lambda_j \bias_{i,j}$. Let the bias vector of the instance $\Psi$ be $\bias_{\veclambda} (\Psi) = (\bias_{\veclambda} (\Psi)_1,\ldots,\bias_{\veclambda} (\Psi)_n)$.   It turns out that the ability to estimate the $\ell_1$ norm of a vector in the ``turnstile setting'' implies that for any given $\veclambda$ vector, we can estimate the $\ell_1$ norm of $\bias_{\veclambda}(\Psi)$ (to within a multiplicative factor of $(1\pm \epsilon)$ for arbitrarily small $\epsilon> 0$) dynamically. We refer to an algorithm that aims to solve the $(\gamma,\beta)$-$\maxf$ using only an estimate of the $\ell_1$ norm of $\bias_\lambda(\Psi)$ (for some $\veclambda$ based on $f,\gamma,\beta$) as a ``bias-based algorithm''. A priori it is not clear how to choose a $\veclambda$ vector for a given problem. The crux of our analysis is to identify two (bounded, closed) convex sets
$K_\gamma^Y,K_\beta^N \subseteq \R^k$ such that if the two sets are disjoint then the hyperplane separating them gives us the desired $\veclambda$.

We now give some insight into the sets $K_\gamma^Y$ and $K_\beta^N$. Roughly these sets capture properties of instances of $\maxf$ on \emph{$k$ variables}, say $x_1,\ldots,x_k$. The instances we consider are special in that $x_i$ always appears as the $i$th variable in every constraint: the only variability being in whether it appears positively or negatively. The set $K_\gamma^Y$ consists of the bias vectors $\bias_{\veclambda} (\Psi)$ of all instances $\Psi$ that have $\val_\Psi(1^k) \geq \gamma$, i.e., the assignment of all $1$'s satisfied $\gamma$ fraction of the constraints of $\Psi$. The set $K_\beta^N$ is similarly supposed to capture the biases $\bias_{\veclambda} (\Psi)$ of instances $\Psi$ for which the value is at most $\beta$. Determining exactly which assignments achieve this bounded value turns out to be subtle and we defer describing it here. But given our choice, our analysis roughly works as follows: Given an instance $\Psi$ on $n$ variables, we create a distribution $\cD(\Psi)\in\Delta(\{-1,1\}^k)$ and its projection $\vecmu$ onto $\R^k$  such that if $\Psi$ is a YES instance, then $\vecmu$ ends up being in $K_\gamma^Y$, while if $\Psi$ is a NO instance, $\vecmu\in K_\beta^N$. Most crucially, the $\ell_1$ norm of $\bias_{\veclambda}(\Psi)$ exactly corresponds to the distance from $\vecmu$ to the hyperplane separating $K_\gamma^Y$ and $K_\beta^N$, which allows us to distinguish the YES and NO cases. Details of the definition of sets can be found in~\autoref{sec:results} and the analysis of the algorithm can be found in~\autoref{sec:algorithm}.

\paragraph{Lower bounds via a new set of communication problems.} 
%\cnote{Update this part to our new analysis.}
%%\vnote{Just adding a note here that the LNW reduction is non-uniform but it does not affect the lower bounds. See for example - paragraph 2 sec 1.1 in AHLW, LNW second page last paragraph and the paragraph before "Our techniques" in page 3. So in the reduction from the simultaneous communication game to the linear sketching algorithm, they do crucially use the fact that each player has unbounded space and time for computation and only the length of the messages passed is counted.}
%mnote{Didn't find a good place to state this. Also in CC world non-uniform reductions are common (e.g., when we say wlog we assume protocol is deterministic ...). or using strategies that involve sampling from some conditional distribution. So it is not crucial to state the additional stuff from AHLW.}
Hardness results in streaming are usually obtained by appealing to lower bounds for communication complexity problems. In our case, both our lower bounds for sketching algorithms and general streaming algorithms are derived from lower bounds on the one-way communication complexity of a class of 2-player problem we call the ``Randomized Mask Detection'' (RMD) problems. (See \autoref{def:rmd}.) We first describe this problem and our results about this problem before returning to the streaming lower bounds.

An RMD problem is specified by two distributions $\cD_Y$ and $\cD_N$ supported on $\{-1,1\}^k$. In this problem Alice gets a vector $\vecx^* \in \{-1,1\}^n$ chosen uniformly at random which we view as a $2$-coloring of the vertex set $[n]$, and Bob gets a random $k$-uniform hypermatching $M$ with $\alpha n$ hyperedges on $[n]$, along with a vector $\vecz \in \{-1,1\}^{k\alpha n}$ whose distribution depends on whether we are in the YES case or NO case (here $\alpha$ is some small but positive constant). Specifically, $\vecz$ specifies the values of $\vecx^*$ on the vertices touched by $M$, but this information is hidden partially by picking for each edge (independently) a masking vector $\vecb \in \{-1,1\}^k$ and letting $\vecz$ for this edge be the information for $\vecx^*$ masked by (xor'ing with) $\vecb$. See \autoref{sec:BHBHM boosting} for a mathematically precise statement. The key difference between the YES instance and the NO instance is the distribution of $\vecb$: In the YES case, for every edge, the masking vector $\vecb$ is chosen independently according to some distribution $\cD_Y$ supported on $\{-1,1\}^k$ whose marginals are in $K_\gamma^N$; and in the NO case, they come independently from the distribution $\cD_N$ whose marginals are in $K_\beta^Y$. In the settings of interest to us $K_\gamma^Y$ and $K_\beta^N$ intersect and we ignore $K_\gamma^N$ and $K_\beta^Y$, and just consider two arbitrary distributions $\cD_Y$ and $\cD_N$ with matching marginals. The technical meat of our negative result is proving that for an arbitrary pair of distributions $\cD_Y$ and $\cD_N$ with matching marginals, any one-way communication protocol with $o(\sqrt{n})$ communication from Alice to Bob has $o(1)$-advantage in distinguishing the YES and NO cases. See \autoref{thm:communication lb matching moments}. 

The proof of \autoref{thm:communication lb matching moments} starts with the work of Kapralov, Khanna, and Sudan~\cite{KKS} which roughly  shows that $(\cD_Y,\cD_N)$-RMD is hard on the special case where $\cD_Y$ is uniform on $\{(1,1), (-1,-1)\}$ and $\cD_N$ is uniform on $\{-1,1\}^2$. Strictly speaking their formalism is slightly different\footnote{In order to handle the general \textsf{Max-CSP} problem, in RMD we extend the previous framework with a more detailed encoding of the hypermatching $M$, and also allow for a general masking vector $\vecb$. Due to these extensions, we cannot immediately conclude hardness of RMD from previous results, and we prove it from scratch.} --- and one in which we are not able to express all our problems, but their proof for this case certainly applies to our formalism.
The proof of \cite{KKS} is Fourier analytic, based on prior work of Gavinsky, Kempe, Kerenidis, Raz, and de Wolf~\cite{GKKRW}.
The first step of our analysis extends this Fourier analytic approach to the case of distributions over $\{-1,1\}^k$ for all values of $k$, and to all distributions $\cD_Y$ and $\cD_N$ that have {\em uniform marginals}. This is reported in \autoref{sec:BHBHM 1 wise}.

The Fourier analytic proof does not seem to extend to the case where $\cD_Y$ and $\cD_N$ have arbitrary but matching marginals (at least we were unable to do so). To get the full case, we turn to reductions. Specifically we show that while we cannot directly prove the indistinguishability of general $\cD_Y$ and $\cD_N$ with matching marginals, we can use the indistinguishability for uniform marginals as a tool (via reductions) to show indistinguishability of some restricted pairs of distributions $(\cD,\cD')$. 
The key to the final result is that for any pair of distributions $\cD_Y$ and $\cD_N$ with matching marginals, there is a path from one to the other of finite length (our upper bound is $\poly(k!)$) such that every adjacent pair of distributions on the path is indistinguishable by our aforementioned reductions for restricted pairs. We remark that while $\cD_Y$ and $\cD_N$ are typically chosen to have interesting properties with respect to their value on various assignments, the intermediate distributions may not have any interesting properties for the underlying optimization problem! But the generality of the framework turns out to be a strength in that we can refer to these problems anyway and use their indistinguishability features. The path from $\cD_Y$ to $\cD_N$ allows us to use triangle-inequality for indistinguishability to get the final result on indistinguishability of RMD on distributions with matching marginals. Details of this part can be found in \autoref{sec:BHBHM general}.

\paragraph{The actual lower bounds.} 
Returning to the streaming problems, the rough idea is to use the two player lower bounds to derive lower bounds for a streaming version of the RMD, and then to reduce this problem to our target $\maxf$ problem. An instance of the streaming RMD problem with distributions $\cD_Y,\cD_N$ generates an Alice input $\vecx^*$ as in the RMD problem, and $T$ sets of Bob inputs $(M_1,\vecz_1),\ldots,(M_T,\vecz_T)$ independently conditioned on $\vecx^*$. It then creates a stream concatenating the $T$ Bob inputs and the streaming challenge is to determine if the underlying mask distribution is $\cD_Y$ or $\cD_N$. Note that in the streaming setting, there is no player corresponding to Alice, making  the streaming RMD problem potentially harder to solve than the 2-player problem. Indeed our initial hope (and claim) was that the streaming RMD problem reduces to the 2-player RMD problem, but this hope turns out to be false. We are however able to establish such a reduction when $\cD_Y$ and $\cD_N$ have uniform marginals as claimed in \autoref{thm:one-wise}. In fact we get a slight generalization which allows the two distribution $\cD_Y$ and $\cD_N$ to be derived from two distributions $\cD'_Y$ and $\cD'_N$ with uniform marginals, by padding with a common distribution $\cD_0$ (see \autoref{thm:main-negative}).

While the hope of reducing the 2-player RMD problem to the streaming problem fails in generality, it turns out that we can get a reduction to a $T$-player simultaneous communication version of RMD. (In this simultaneous communication version, the $t$th player gets $(M_t,\vecz_t)$ as input and needs to send a message to a referee who collects the messages from the $T$ players and attempts to guess if the mask distribution is $\cD_Y$ or $\cD_N$.) Since a sketching algorithm can be turned into a protocol for the simultaneous communication game, we are able to show that whenever $K^Y_\gamma$ and $K^N_\beta$ intersect, any  sketching algorithm that solves the $(\gamma,\beta)$-approximation version of $\maxf$ requires $\Omega(\sqrt{n})$ space. 
% This lower bound plugs in very neatly into a powerful result of Ai, Hu, Li and Woodruff~\cite{ai2016new} (which builds on the work of Li, Nguyen and Woodruff~\cite{li2014turnstile}) which shows that simultaneous communication complexity lower bounds suffice for lower bounds for dynamic streaming problems (even in the strict turnstile setting where deletions occur only after insertions). We use this connection to show that whenever $K^Y_\gamma$ and $K^N_\beta$ intersect, the $(\gamma,\beta)$-approximation version of $\maxf$ is hard (requires $\Omega(\sqrt{n})$ space) in the dynamic setting (see \autoref{thm:main-negative dynamic}). 
 
 \autoref{sec:SpaceLowerBound insertion} describes the various RMD problems discussed above and how they can be used to get proofs of \autoref{thm:main-negative dynamic} and \autoref{thm:main-negative}.

\subsection{Future questions and work}\label{ssec:future}

Some of the main questions left open in this work are listed below:
\begin{enumerate}
    \item Does the characterization given by \autoref{thm:main-detailed} actually hold for general streaming algorithms? Resolving this either way would be quite interesting. 
    \item Can the methods be extended to handle the case where the constraints come from a family of functions, rather than a single function? We believe this should be straightforward to achieve.
    \item Can we further extend the results to the setting where the constraints are not placed on literals, but rather only on variables? Such an extension seems to require new ideas beyond those in this paper.
    \item Can we extend the results to the non-Boolean setting, i.e., when the variables take on values from an arbitrary finite set, as opposed to $\{-1,1\}$. We stress that both the positive and negative results in this paper exploit restrictions of the Boolean setting! In this direction, Guruswami and Tao~\cite{GT19} proved that $(1/p+\epsilon)$-approximation for the unique games with alphabet size $p$ requires $\tilde{\Omega}(\sqrt{n})$ space in the streaming setting.
    \item Can  the lower bound for the hard problems be improved to linear space lower bounds? Such an improvement was given by Kapralov and Krachun~\cite{KK19} for the \textsf{Max-2LIN} problem ($\maxf$ where $f(x,y) = x \oplus y$) in a technical tour-de-force. Extending this work to other optimization problems seems non-trivially challenging.
    \item Finally, our work and all the questions above only consider the setting of single-pass streaming algorithms. Once this is settled, it would make sense to extend the analyses to multi-pass algorithms. While there are several multi-pass streaming algorithms and lower bounds (see, e.g., ~\cite{chakrabarti2020data,mcgregor2014graph,guha2008tight} and references therein), we note that Assadi, Kol, Saxena, and~Yu~\cite{assadi2020multi} recently suggested a multi-round version of the Boolean Hidden Hypermatching problem that allows to extend some previous single-pass results (including a lower bound for approximate \textsf{Max-2LIN}) to the multi-pass setting.
\end{enumerate}

%\cnote{Add an open question on insertion-only lower bound for non-one-wise case.}

\subsection{Structure of rest of the paper}

%\mnote{Return to this para after rest of paper is written.}

In \autoref{sec:results}, we describe our result in detail. In particular we build our convex set framework and give  an explicit criterion to distinguish the easy and hard $\maxf$ problems. We also describe sufficient conditions for the hardness of some streaming problems in the streaming setting. \autoref{sec:preliminaries} contains some of the preliminary background used in the rest of the paper.
In \autoref{sec:algorithm}, we describe and analyze our algorithm that yields our easiness result. In \autoref{sec:SpaceLowerBound insertion}, we define the central family of communication problems that lie at the heart of our lower bounds and show how the communication complexity of this problem leads to the streaming space lower bounds claimed in \autoref{sec:results}. In \autoref{sec:BHBHM 1 wise}, we first establish the desired lower bounds for a subclass of the problems using Fourier analytic methods. In \autoref{sec:BHBHM general}, we establish reductions between various communication problems that allow us to prove our most general lower bounds.
\section{Our Results}\label{sec:results}

%\mnote{Done with this section for now. Need to return to fix Sections 2.3 and 2.4}

%\mnote{Proposed restructuring: (1) key definitions, (2) results in the dynamic setting, (3) results in the insertion-only setting, (4) Examples (in insertion-only setting?}

We start with some notation needed to state our results. We use $\R^{\geq 0}$ to denote the set of non-negative real numbers. For a finite set $\Omega$, let $\Delta(\Omega)$ denote the space of all probability distributions over $\Omega$, i.e., 
$$\Delta(\Omega) = \{\cD:\Omega \to \R^{\geq 0} | \sum_{\omega \in \Omega} \cD(\omega) = 1\}.$$
We view $\Delta(\Omega)$ as being contained in $\R^{|\Omega|}$. We use $X\sim\cD$ to denote a random variable drawn from the distribution $\cD$. 

\subsection{Key definitions}

The main objects that allow us to derive our characterization are the space of distributions on constraints that either allow a large number of constraints to be satisfied, or only a few constraints to be satisfied. To see where the distributions come from, note that
distributions of constraints over $n$ variables can naturally be identified with instances of weighted constraint satisfaction problem (where the weight associated with a constraint is simply its probability). In what follows we will consider instances on exactly $k$ variables $x_1,\ldots,x_k$. Furthermore all constraints will use $x_i$ as the $i$th variable. Hence, a constraint on $k$ variables is specified by $\vecb \in \{-1,1\}^k$, specifying the constraint $f(b_1x_1,\ldots,b_kx_k)$. Thus in what follows we will equate ``instances on $k$ variables'' with distributions on $\{-1,1\}^k$.

Given $0\leq\beta\leq \gamma \leq 1$  we will consider two sets of instances/distributions. The first set $S_\gamma^Y = S_\gamma^Y(f)$ will be instances where $\gamma$ fraction of the constraints are satisfied by the assignment $1^k$. The second set $S_\beta^N = S_\beta^N(f)$ is a bit more subtle: it consists of instances where no ``independent identical distribution'' on the variables satisfies more that $\beta$-fraction of the clauses. To elaborate, recall that the only distributions on a single variable taking values in $\{-1,1\}$ are the Bernoulli distributions. Let $\bern(p)$ denote the distribution that takes the value $1$ with probability $p$ and $-1$ with probability $1-p$. Then an instance belongs to $S_\beta^N$ if for every $p$, when $(x_1,\ldots,x_k)$ gets a random assignment chosen according to $\bern(p)^k$, the expected fraction of satisfied clauses is at most $\beta$. 
The following is our formal definition.

\begin{definition}[Space of Yes/No Distributions]\label{def:sets}
For $\gamma,\beta \in \mathbb{R}$, we define
\begin{align*}
S_\gamma^Y  = S_\gamma^Y(f) & = \{\cD_Y \in \Delta(\{-1,1\}^k) ~\vert ~\Exp_{\vecb\sim \cD_Y}[f(\vecb)]\ge \gamma \}\\
\mbox{ and } S_\beta^N = S_\beta^N(f) & = \{\cD_N\in \Delta(\{-1,1\}^k) ~\vert ~\Exp_{\vecb\sim \cD_N}\Exp_{\veca \sim \bern(p)^k}[f(\vecb \odot \veca)]\leq \beta, \forall p \} \, .\\  
\end{align*}
\end{definition}
For $\gamma > \beta$ the sets $S_\gamma^Y$ and $S_\beta^N$ are clearly disjoint. But their marginals, when projected to single coordinates need not be, and this is the crux of our characterization. In what follows, we define sets $K_\gamma^Y$ and $K_\beta^N$ to be the marginals of the distributions in $S_\gamma^Y$ and $S_\beta^N$ respectively.
For a distribution $\cD \in \Delta(\{-1,1\}^k)$, let $\vecmu(\cD)$ denote its marginals, i.e., $\vecmu(\cD) = (\mu_1,\ldots,\mu_k)$ where $\mu_i = \Exp_{\vecb\sim\cD}[b_i]$.

\begin{definition}[Marginals of Yes/No Distributions]\label{def:marginals}
For $\gamma,\beta \in \mathbb{R}$, we define
\begin{align*}
K_\gamma^Y = K_\gamma^Y(f) & = \{~\vecmu(\cD_Y)~\vert ~ \cD_Y \in S_\gamma^Y \}\\
\mbox{ and }K_\beta^N = K_\beta^N(f)& = \{~\vecmu(\cD_N)~\vert ~ \cD_N \in S_\beta^N \} \, .
\end{align*}
\end{definition}

With the two definitions above in hand we are ready to describe our characterizations of easy vs. hard approximation versions of $\maxf$. 

\subsection{Results on sketching algorithms}

The following theorem now formalizes the informal statement that low space sketching algorithms (see \autoref{def:sketching alg}) can only capture the marginals of distributions.

%\cnote{Update the below to work for dynamic setting.}
\begin{theorem}[Dichotomy for Sketching Algorithms]\label{thm:main-detailed}
For every function $f:\{-1,1\}^k \to \{0,1\}$ and for every $0\leq\beta<\gamma\leq1$, the following hold:
\begin{enumerate}
    \item If $K_\gamma^Y(f) \cap K_\beta^N(f) = \emptyset$, then $(\gamma,\beta)$-$\maxf$ admits a a probabilistic linear sketching algorithm (see \autoref{def:sketching alg}) that uses $O(\log n) $ space on instances on $n$ variables.
    \item If $K_\gamma^Y(f) \cap K_\beta^N(f) \neq \emptyset$, then for every $\epsilon>0$,  every sketching algorithm for $(\gamma-\eps,\beta+\eps)$- $\maxf$ requires $\Omega(\sqrt{n})$ space\footnote{The constant hidden in the $\Omega$ notation may depend on $k$ and $\epsilon$.} on instances on $n$ variables. Furthermore, if $\gamma = 1$, then every sketching algorithm for $(1,\beta+\epsilon)$-$\maxf$ requires $\Omega(\sqrt{n})$ space. 
 \end{enumerate}
\end{theorem}

\begin{proof}[Proof of \autoref{thm:main-detailed}]
Part (1) of the theorem is restated and proved as \autoref{thm:main-positive} in \autoref{sec:algorithm}. Part (2) is proved as \autoref{thm:main-negative dynamic} in \autoref{sec:SpaceLowerBound insertion}.
\end{proof}

We now turn to the implications of this theorem.
First, to get \autoref{thm:main-intro} from \autoref{thm:main-detailed}, we need to show that the question ``Is $K_\gamma^Y\cap K_\beta^N = \emptyset$?'' can be decided in polynomial space. To this end, we first make the following observation.

\begin{lemma}\label{lem:convex}
For every $\beta,\gamma \in [0,1]$ the sets $S_\gamma^Y, S_\beta^N, K_\gamma^N$ and $K_\beta^Y$ are bounded, closed and convex. Furthermore, $K_\gamma^Y \cap K_\beta^N = \emptyset$ can be expressed in the quantified theory of the reals with $2$ quantifier alternations, $O(2^k)$ variables, and polynomials of degree at most $k+1$. 
\end{lemma}

\begin{proof} We start by considering the sets $S_\gamma^Y$ and $S_\beta^N$. It is straightforward to see that $S_\gamma^Y$ is a bounded and convex polytope in $\R^{2^k}$. 
$S_\beta^N$ is a bit more subtle due to the universal quantification over $p \in [0,1]$. It is now specified by infinitely many linear inequalities in $\R^{2^k}$ and so is still a bounded and convex set (though not necessarily a polytope). $K_\gamma^Y$ (resp. $K_\beta^N$) is obtained by a linear projection from $\R^{2^k}$ to $\R^k$. So $K_\gamma^Y$ is a bounded, closed, and convex polytope in $\R^{k}$, while $K_\beta^N$ is still a bounded, closed, and convex set.

To get an intersection detection algorithm we use one more property. Note that for variable $p$, the condition
$\Exp_{\veca\sim \cD_N}\Exp_{\vecb \sim \bern(p)^k}[f(\vecb \odot \veca)]\le \beta$ is a polynomial inequality in $p$ of degree at most~$k$, with coefficients that are linear forms in $\cD_N(\vecb)$, $\vecb \in \{-1,1\}^k$. This allows us to express the condition $K_\gamma^Y \cap K_\beta^N \not= \emptyset$ using the following system of quantified polynomial inequalities:
\begin{eqnarray}
&\exists &\cD_Y,\cD_N \in \R^{2^k}, ~ \forall p \in [0,1] \mbox{ s.t. } \nonumber\\
&&\cD_Y, \cD_N \mbox{ are distributions,}\label{eq:poly-dist}\\
&& \forall i \in [k], ~ \Exp_{\vecb \sim \cD_Y} [b_i] = \Exp_{\vecb \sim \cD_N} [b_i],\label{eq:poly-marginal} \\
&& \Exp_{\vecb\sim \cD_Y} [f(\vecb)] \geq \gamma, \label{eq:poly-gamma}\\
&& \Exp_{\vecb\sim\cD_N} \Exp_{\veca\sim\bern(p)^k} [f(\veca \odot \vecb)] \le \beta. \label{eq:poly-beta}
\end{eqnarray}
Note that Equations (\ref{eq:poly-dist}), (\ref{eq:poly-marginal}) and (\ref{eq:poly-gamma}) are just linear inequalities in the variables $\cD_Y,\cD_N$ and do not depend on $p$. As noticed above Equation~(\ref{eq:poly-beta}) is an inequality in $p$, and $\cD_N$, of degree $k$ in $p$, and $1$ in $\cD_N$. We thus get that the intersection problem can be expressed in the quantified theory of the reals by an expression with two quantifier alternations, $2^k$ variables and $O(2^k)$ polynomial inequalities, with polynomials of degree at most $k+1$. (Most of the inequalities are of the form $\cD_Y(\vecb) \geq 0$ or $\cD_N(\vecb) \geq 0$. Only $O(k)$ inequalities are not of that form; and of these, only one is non-linear.) 
\end{proof}

The quantified theory of the reals is known to be solvable in PSPACE. In particular we use the following theorem.

\begin{theorem}[\protect{\cite[Theorem 14.11, see also Remark 13.10]{BasuPR}}]\label{thm:bpr}
The truth of a quantified formula with $w$ quantifier alternations over $K$ variables and polynomial (potentially strict) inequalities can be decided in space $K^{O(w)}$ and time $2^{K^{O(w)}}$. 
\end{theorem}

(Specifically, Theorem 14.11 in \cite{BasuPR} asserts the time complexity above, and Remark 13.10 yields the space complexity.)

\autoref{thm:main-intro} now follows immediately.

\begin{proof}[Proof of \autoref{thm:main-intro}]
\autoref{thm:main-detailed} asserts that the $(\gamma,\beta)$-approximation version of $\maxf$ is easy if and only if $K_\gamma^Y \cap K_\beta^N = \emptyset$. \autoref{lem:convex} asserts that this condition is in turn expressible in the quantified theory of the reals with 2 quantifier alternations. Finally \autoref{thm:bpr} asserts that this can be decided in polynomial space. The theorem follows.
\end{proof}

We note that the literature on approximation algorithms usually considers a single parameter version of the problem. In our context we would say that an algorithm $A$ is a $\alpha$-approximation algorithm for $\maxf$ if for every instance $\Psi$, we have 
\[ \alpha \cdot \val_\Psi \leq A(\Psi) \leq \val_\Psi\, .\]
The following proposition converts our main theorem in terms of this standard notion.

\begin{proposition}\label{prop:approx-equivalence} 
Fix $f:\{-1,1\}^k$ and let $K_\gamma^Y$ and $K_\beta^N$ denote the space of marginals for this function $f$. Let 
\[\alpha = \inf_{\beta \in [0,1]} \left \{ \sup_{\gamma\in(\beta,1] \rm{~s.t~} K_\gamma^Y \cap K_\beta^N = \emptyset} \{\beta/\gamma\}\right\}.\]
Then for every $\epsilon > 0$, there is an $(\alpha-\epsilon)$-approximation algorithm for $\maxf$ that uses $O(\log n)$ space. Conversely every $(\alpha+\epsilon)$-approximation algorithm for $\maxf$ requires $\Omega(\sqrt{n})$ space.
\end{proposition}

\begin{proof}
For the positive result, let $\tau \triangleq \epsilon\cdot \rho(f)/2$, where $\rho(f) = 2^{-k} \sum_{\veca \in \{-1,1\}^k} f(\veca)$ is the fraction of clauses satisfied by a random assignment. Let 
$$A_\tau = \{(i\tau,j\tau) \in [0,1]^2 ~\vert~ i,j \in \mathbb{Z}^{\geq 0}, i > j, K_{i \tau}^Y \cap K_{j\tau}^N = \emptyset \}.$$
By \autoref{thm:main-detailed}, for every $(\gamma,\delta) \in A_\tau$ there is a $O(\log n\log(1/\tau))$-space algorithm for $(\gamma,\beta)$-$\maxf$ with error probability $1/(10\tau^2)$, which we refer to as the $(\gamma,\beta)$-distinguisher below. In the following we consider the case where all $O(\tau^{-2})$ distinguishers output correct answers, which happens with probability at least $2/3$.

Our $O(\tau^{-2}\log(1/\tau)\log n) = O(\log n)$ space $(\alpha-\epsilon)$-approximation algorithm for $\maxf$ is the following: On input $\Psi$, run in parallel all the $(\gamma,\beta)$-distinguishers on $\Psi$, for every $(\gamma,\beta) \in A_\tau$. Let
\[
\beta_0 = \arg \max_{\beta} [\exists \gamma \text{ such that the }(\gamma,\beta)\text{-distinguisher outputs YES on }\Psi] \, .
\]
Output $\beta' = \max\{\rho(f),\beta_0\}$.

We now prove that this is an $(\alpha-\epsilon)$-approximation algorithm. First note that by the correctness of the distinguisher we have $\beta' \leq \val_\Psi$. 
Let $\gamma_0$ be the smallest multiple of $\tau$ satisfying $\gamma_0 \geq (\beta_0 + \tau)/\alpha$.
By the definition of $\alpha$, we have that $K_{\gamma_0}^Y \cap K_{\beta_0+\tau}^N = \emptyset$. So $(\gamma_0,\beta_0+\tau) \in A_\tau$ and so the $(\gamma_0,\beta_0+\tau)$-distinguisher must have output NO on $\Psi$ (by the maximality of $\beta_0$). By the correctness of this distinguisher we  conclude $\val_\Psi \leq \gamma_0 \leq (\beta_0 + \tau)/\alpha + \tau \leq (\beta'+\tau)/\alpha + \tau$. We now verify that $(\beta' + \tau)/\alpha + \tau \leq \beta'/(\alpha-\epsilon)$ and this gives us the desired approximation guarantee. We have 
\[(\beta' + \tau)/\alpha + \tau \leq (\beta'+2\tau)/\alpha \leq (\beta'/\alpha) \cdot (1 + 2\tau/\rho(f)) = (\beta'/\alpha)(1+\epsilon) \leq (\beta'/(\alpha(1-\epsilon))),
\]
where the first inequality uses $\alpha \leq 1$, the second uses $\beta' \geq \rho(f)$, the equality comes from the definition of $\tau$ and the final inequality uses $(1+\epsilon)(1- \epsilon) \leq 1$. This concludes the positive result.

The negative result is simpler. Given $\gamma,\beta$ with $\beta/\gamma \geq \alpha + \epsilon$, we can use an $(\alpha+\epsilon)$-approximation algorithm $A$  to solve the $(\gamma,\beta)$-$\maxf$, by outputting YES if $A(\Psi) \geq \beta$ and NO otherwise.
\end{proof}

\subsubsection{Approximation resistance}
%\cnote{Emphasize insertion-only in this sub-section.} \mnote{\autoref{cor:approx-res} is in the dynamic model. \autoref{cor:one-wise} is in the insertion-only model.}

We now turn to \autoref{cor:approx-res}. 
Recall that for a function $f:\{-1,1\}^k \to \{0,1\}$, we define $\rho(f) = 2^{-k} \cdot |\{\veca \in \{-1,1\}^k : f(\veca)=1\}|$ to be the probability that a uniformly random assignment satisfies~$f$. Recall further that $f$ is approximation-resistant if for every $\epsilon > 0$, the $(1,\rho(f)+\epsilon)$-approximation version of $\maxf$ requires polynomial space.

\begin{proof}[Proof of \autoref{cor:approx-res}]
By \autoref{thm:main-detailed} we have that $\maxf$ is approximation-resistant if and only if $K_1^Y \cap K_{\rho(f)+\eps}^N \ne \emptyset$ for every $\eps > 0$. In turn, this is equivalent to saying $\maxf$ is approximation resistant if and only if $K_1^Y \cap K_{\rho(f)}^N \ne \emptyset$. If $K_1^Y \cap K_{\rho(f)}^N = \emptyset$, then by the property that these sets are closed, we have that there must exist $\eps > 0$ such that  $K_{1 - \eps}^Y \cap K_{\rho(f)+\eps}^N = \emptyset$. In turn this implies, again by \autoref{thm:main-detailed}, that the $(1-\eps,\rho(f)+\eps)$-approximation version of $\maxf$ can be solved by a linear sketching algorithm with $O(\log n)$ space.
Finally, from \autoref{lem:convex} and \autoref{thm:bpr} the condition ``Is $K_1^Y \cap K_{\rho(f)}^N = \emptyset$?'' can be checked in polynomial space. 
\end{proof}

\subsection{Lower bounds in the streaming setting}\label{ssec:lb-detail-insert}

For two broad sets of special cases we are able to get lower bounds =in the streaming setting with general streaming algorithms where the lower bounds match the upper bounds derived using linear sketches. To define these classes we need some definitions.

We say that a distribution $\cD \in \Delta(\{-1,1\}^k)$ is {\em one-wise-independent} if $\vecmu(\cD) = 0^k$. We say that a pair of distributions $(\cD_Y,\cD_N)$ form a {\em padded one-wise pair} if there exists $\tau \in [0,1]$ and distributions $\cD_0,\cD'_Y$ and $\cD'_N$ such that (1) $\cD'_Y$ and $\cD'_N$ are one-wise independent and (2) $\cD_Y = \tau\cD_0 + (1-\tau)\cD'_Y$ and $\cD_N = \tau\cD_0 + (1-\tau)\cD'_N$. 

%Our theorem here states that padded one-wise pairs with  $\cD_Y \in S^Y_\gamma$ and $\cD_N \in S^N_\beta$ implies hardness of the $(\gamma,\beta)$-approximation version of $\maxf$ in the {\em insertion-only} setting. 

Our main lower bound in the streaming setting asserts that if $S^Y_\gamma\times S^N_\beta$ contains a padded one-wise pair $(\cD_Y,\cD_N)$ then $(\gamma,\beta)$-$\maxf$ requires $\Omega(\sqrt{n})$-space.
%\mnote{Can someone make the below "restatable" and update/restate in section 5?}

\begin{restatable}[Streaming lower bound]{theorem}{restatethmmainnegative}
\label{thm:main-negative}
For every function $f:\{-1,1\}^k\rightarrow\{0,1\}$ and for every $0\leq\beta<\gamma\leq1$, if there exists a padded one-wise pair of distributions $\cD_Y \in S^Y_\gamma$ and $\cD_N \in S^N_\beta$ then, for every $\epsilon>0$, then every streaming algorithm that solves $(\gamma-\eps,\beta+\eps)$-$\maxf$ requires $\Omega(\sqrt{n})$ space. %\footnote{The constant hidden in the $\Omega$ notation may depend on $k$ and $\epsilon$.} 
%in the insertion-only setting. 
Furthermore, if $\gamma = 1$ then $(1,\beta+\epsilon)$-$\maxf$ requires $\Omega(\sqrt{n})$ space. 
\end{restatable}

\iffalse
%\begin{restatable}{theorem}{thmtrianglelocal}
\label{thm:triangle-local-exp}
For every prime $q$, $\varepsilon > 0$, every mixing matrix $M \in \F_q^{k \times k}$, and for every symmetric memoryless channel
$\C_{Y|Z}$ over $\F_q$,
the \Arikan\ martingale sequence associated with $M^{\otimes 2}$ and $\C_{Y|Z}$ is $(\frac{1}{k^2}, 2-\varepsilon)$-exponentially locally polarizing.
%\end{restatable}
\fi 
\autoref{thm:main-negative} is proved in \autoref{sec:lb insert}. As stated above the theorem is more complex to apply than, say, \autoref{thm:main-detailed}, owing to the fact that the condition for hardness depends on the entire distribution (and the sets $S^Y_\gamma$ and $S^N_\beta$) rather than just marginals (or the sets $K^Y_\gamma$ and $K^N_\beta$). However it can be used to derive some clean results, specifically \autoref{thm:one-wise} and \autoref{thm:main-intro-k=2}, that do depend only on the marginals. We prove these (assuming \autoref{thm:main-negative}) below.

Recall that we say that a function $f$ {\em supports one-wise independence} if there exists a one-wise-independent distribution $D$ supported on the satisfying assignments to $f$. Note that this is equivalent to saying $0^k \in K^Y_1$. \autoref{thm:one-wise} asserts that every function that supports a one-wise independent distribution is approximation resistant in the streaming setting.

\begin{proof}[Proof of \autoref{thm:one-wise}]
%\mnote{Fix this proof to cite \autoref{thm:main-negative}}
Let $\rho = \rho(f)$. We first show that the vector $0^k$ belongs to both $K_1^Y$ and $K_{\rho}^N$. 
We then note that this implies the existence of a padded one-wise pair of distributions $\cD_Y \in S_1^Y$ and $\cD_N \in S_{\rho}^N$ allowing us to apply \autoref{thm:main-negative} to get the theorem.

Let $\cD_Y$ be the distribution proving that $f$ supports a one-wise independent distribution, i.e., $\cD_Y$ is supported on $f^{-1}(1)$ and satisfies $\Exp_{\vecb \in \cD_Y}[b_i] = 0$ for every $i \in [k]$. It follows that $\cD_Y \in S_1^Y$ and $0^k \in K_1^Y$.
Let $\cD_N$ be the uniform distribution on $\{-1,1\}^k$. Note that for every $\veca \in \{-1,1\}^k$ we have $\veca \odot \vecb$ is uniformly distributed over $\{-1,1\}^k$ if $\vecb \sim \cD_N$. Consequently, for every $\veca$ we get $\Exp_{\vecb \sim \cD_N} [f(\vecb \odot \veca)] = \rho$, and so for every $p \in [0,1]$, we have  
$$\Exp_{\veca \sim\bern(p)^k} \Exp_{\vecb \sim \cD_N} [f(\vecb \circ \veca)] = \rho.$$
We conclude the $\cD_N \in S_{\rho}^N$ and so $0^k \in K_{\rho}^N$. 

Now by definition we have that $\cD_Y$ and $\cD_N$ form a padded one-wise pair (using $\cD'_Y = \cD_Y$, $\cD'_N = \cD_N$ and $\tau = 0$) and so \autoref{thm:main-negative} is applicable to show that $\maxf$  is not $(1,\rho-\eps)$-approximable and so is approximation-resistant.
\end{proof}

We now turn to the proof of \autoref{thm:main-intro-k=2}. Indeed we prove a more detailed statement along the lines of \autoref{thm:main-detailed} in this case. For this part we use the fact, proved below, that any pair of distributions $\cD_Y, \cD_N \in \Delta(\{-1,1\}^2)$ with matching marginals form a padded one-wise pair.

\begin{proposition}\label{prop:k=2padded}
If $\cD_Y, \cD_N \in \Delta(\{-1,1\}^2)$ satisfy $\mu(\cD_Y) = \mu(\cD_N)$ then $(\cD_Y,\cD_N)$ form a padded one-wise pair.
\end{proposition}

\begin{proof}
Let $\cD_Y = (p_{1,1},p_{1,-1},p_{-1,1},p_{-1,-1})$ where $p_{i,j}$ denotes the probability $\Pr_{(a,b) \sim \cD_Y}[a=i,b=j]$. If $\cD_N$ has matching marginals with $\cD_Y$ then there exists a $\delta \in [-1,1]$ such that  $\cD_N = (p_{1,1}-\delta,p_{1,-1}+\delta,p_{-1,1}+\delta,p_{-1,-1}-\delta)$. Assume without loss of generality that $\delta \geq 0$. Let
$\tau = 1-2\delta$,
$\cD_0 = \frac1{1-2\delta}(p_{1,1}-\delta,p_{1,-1},p_{-1,1},p_{-1,-1}-\delta)$,
$\cD'_Y= (1/2,0,0,1/2)$ and $\cD'_N= (0,1/2,1/2,0)$. It can be verified that $\cD'_Y$ and $\cD'_N$ are one-wise independent, $\cD_Y = \tau\cD_0 + (1-\tau)\cD'_Y$ and $\cD_N = \tau\cD_0 + (1-\tau)\cD'_N$, thus proving the proposition.
\end{proof}

Combining \autoref{prop:k=2padded} and \autoref{thm:main-negative} we immediately get the following theorem, which in turn implies  \autoref{thm:main-intro-k=2}.

\begin{theorem}\label{thm:main-detailed-k=2} For every function $f:\{-1,1\}^2 \to \{0,1\}$, and for every $0 \leq \beta < \gamma \leq 1$, the following hold:
\begin{enumerate}
    \item If $K_\gamma^Y(f) \cap K_\beta^N(f) = \emptyset$, then $(\gamma,\beta)$-$\maxf$ admits a linear sketching algorithm that uses $O(\log n) $ space. 
    \item If $K_\gamma^Y(f) \cap K_\beta^N(f) \neq \emptyset$, then for every $\epsilon>0$,  every streaming algorithm that solves $(\gamma-\eps,\beta+\eps)$-$\maxf$ requires $\Omega(\sqrt{n})$ space\footnote{The constant hidden in the $\Omega$ notation may depend on $k$ and $\epsilon$.}. Furthermore, if $\gamma = 1$, then $(1,\beta+\epsilon)$-$\maxf$ requires $\Omega(\sqrt{n})$ space. 
 \end{enumerate}
\end{theorem}

\begin{proof}
Part (1) is simply the specialization of Part (1) of \autoref{thm:main-detailed-k=2} to the case $k=2$. For Part (2), suppose $\mu \in K^Y_\gamma \cap K^N_\beta$. Let $\cD_Y \in S^Y_\gamma$ and $\cD_N \in S^N_\beta$ be distributions such that $\mu(\cD_Y) = \mu(\cD_N) = \mu$. Then by \autoref{prop:k=2padded} we have that $\cD_Y$ and $\cD_N$ form a padded one-wise pair, and so \autoref{thm:main-negative} can be applied to get Part (2). 
\end{proof}

\subsection{Examples}\label{sec:examples}

%\mnote{Need to reexamine this section. Is this obtained in the insertion-only model? Or dynamic?}

We illustrate the applicability of our results with two examples. The first is of the specific function $\maxtwoand$, i.e., $\maxf$ for $f(a,b) = a \wedge b$, i.e., $f(a,b) = 1$ if and only if $a = b = 1$. Here, since we are working with $k=2$ we get to use the stronger separation from \autoref{thm:main-detailed-k=2} (with general streaming lower bounds).

\begin{examplebox}{Example 1 (\textsf{Max-2AND}).}
For the function $f:\{-1,1\}\to\{0,1\}$ given by $f(1,1) = 1$ and $f(a,b) = 0$ otherwise, we would like to calculate the quantity $\inf_\beta\sup_{\gamma | K_\gamma^Y \cap K_\beta^N = \emptyset}\beta/\gamma$. We first note that due to the symmetry of $f$, we have $K^Y_\gamma$ is symmetric, i.e., $(\mu_1,\mu_2) \in K^Y_\gamma \Leftrightarrow (\mu_2,\mu_1) \in K^Y_\gamma$. Similarly with $K^\beta_N$. Further by convexity of $K^Y_\gamma$ and $K^N_\beta$ we get there exists a pair $(\mu_1,\mu_2) \in K^Y_\gamma\cap K^N_\beta$ if and only if there exists a $\mu$ such that $(\mu,\mu) \in K^Y_\gamma\cap K^N_\beta$. We now define two functions that will help us answer the question if such a $\mu$ exists. Let 
\[
\gamma(\mu) :=\max_{\gamma\, |\, (\mu,\mu)\in K^Y_\gamma}\{\gamma\} ~~~ \& ~~~  \beta(\mu) := \min_{\beta\, |\, (\mu,\mu)\in K^N_\beta}\{\beta\}.
\]
Note that $K_\gamma^Y \cap K_\beta^N \ne \emptyset$ if and only if there exists a $\mu$ such that $\gamma \leq \gamma(\mu)$ and $\beta \geq \beta(\mu)$. 
With some minimal calculations for $\gamma(\mu)$ and some slightly more involved ones for $\beta(\mu)$ we can show
\[
\gamma(\mu) = \frac{1+ \mu}{2} ~~~ \text{and} ~~~
\beta(\mu) = \left\{\begin{array}{ll}
|\mu|    & ,\  |\mu|\geq\frac{1}{3}\\
\frac{(1-|\mu|)^2}{4(1-2|\mu|)}    & ,\ \text{else.}
\end{array}
\right.
\]
With the above in hand we can analyze when $K^Y_\gamma \cap K^N_\beta = \emptyset$. 

 First, when $\beta < 1/4$ then $K^N_\beta = \emptyset$. Next, 
%First, 
when $\gamma\leq1/2$, note that $(0,0)\in K^Y_\gamma$ and hence $K^Y_\gamma\cap K^N_\beta\neq\emptyset$ for all $\beta\geq1/4$. When $\gamma>1/2$, we set  $\mu = 2\gamma -1$ which leads to $|\mu| = \mu = 2\gamma -1$ and so 
%(note both $\gamma(\mu)$ and $\beta(\mu)$ only depend on $|\mu|$) 
we get
\[
\beta(\mu)\big|_{\mu = 2\gamma -1 } = \left\{ \begin{array}{ll}
          \frac{(1-\gamma)^2}{3-4\gamma}  & ,\ 1/2\leq\gamma < 2/3\\
          2\gamma-1 & ,\ 2/3\leq\gamma\,.
      \end{array} \right.
\]
We thus get that the set $H^{\cap} \triangleq \{(\gamma,\beta) \in [0,1]^2 | K^Y_\gamma \cap K^N_\beta \ne \emptyset\}$ (of {\em hard} problems) is given by:
\begin{align*}
H^\cap = & ~~~~~~ \left[0,\frac12\right]\times \left[\frac14,1\right] \\
         & \cup ~~\left\{(\gamma,\beta) | \gamma \in \left[\frac12,\frac23\right], \beta \in \left[\frac{(1-\gamma)^2}{3-4\gamma},1\right]\right\} \\
         & \cup  ~~\left\{(\gamma,\beta) |\gamma\in \left[\frac23,1\right], \beta \in \left[2\gamma - 1,1\right]\right\}.
\end{align*}  

%\cnote{The last two regimes should switch.}\mnote{Coorect. Switched}
The quantity $\alpha(\beta) = \sup_{\gamma\in[\beta,1]\, |\, K_\gamma^Y \cap K_\beta^N = \emptyset} \beta/\gamma$ is minimized at $\beta = 4/15$. 

At this point $\alpha = 4/9$, which is consistent with the findings in \cite{CGV20} for the $\textsf{Max-2AND}$ problem. Our more refined analysis also shows that
$\alpha(\beta)$ approaches $1$ as $\beta \to 1$ (suggesting that ``almost-satisfiable'' instances are better approximated). 
\begin{figure}[H]
    \centering
    \includegraphics[width=15cm]{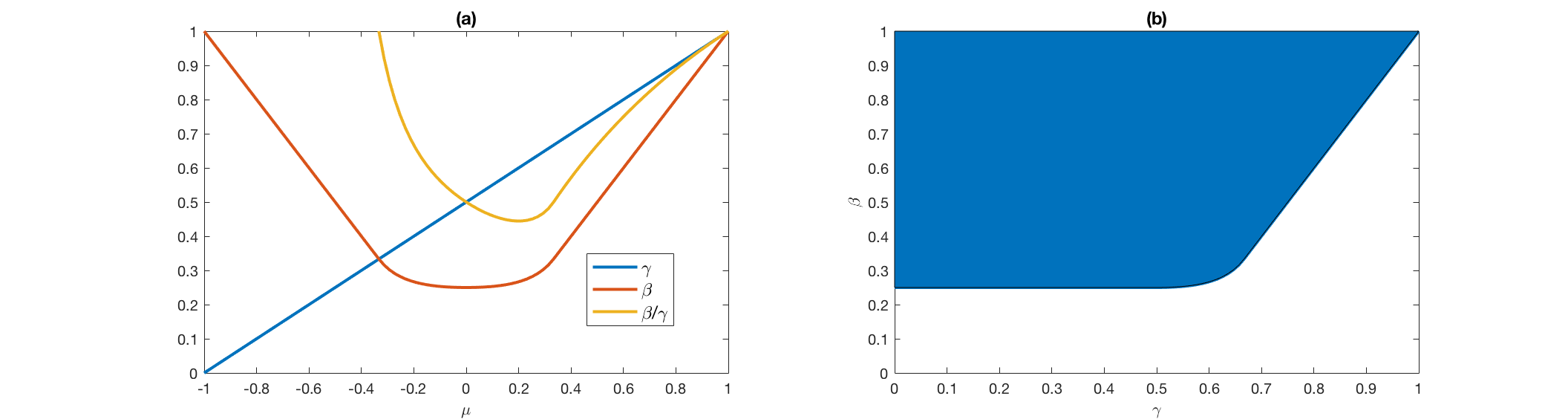}
    \caption{(a) A plot of $\gamma,\beta$, and $\beta/\gamma$ with respect to $\mu$. (b) A plot of $H^{\cap}$.
    % A plot of $\gamma$, $\beta$, and $\beta/\gamma$ with respect to $\mu$. $\beta/\gamma$ is minimized at $\beta=4/15$.\mnote{I would suggest replacing this plot with a plot of $H^\cap$.}
    }
    \label{fig:Max-2AND}
\end{figure}
\end{examplebox}

The second example we consider includes an entire family of functions.

%\mnote{Need to promote to theorem?}

\begin{lemma}[one-wise independence implies approximation resistance]\label{prop:symmetric}
For a symmetric function $f:\{-1,1\}^k\to \{0,1\}$, $\maxf$ is approximation resistant if and only if it supports a~one-wise independent distribution.
\end{lemma}

\begin{proof}
One direction of the implication directly follows from \autoref{cor:approx-res}. For the other direction, we use Fourier analysis. The necessary definitions are included in \autoref{sec:fourier}.
A symmetric function~$f$ is given by a set of ``levels'' $L = \{\ell_1,\ldots,\ell_t\} \subseteq  \{-k,\ldots,k\}$ such that $f(a_1,\ldots,a_k) = 1$ if and only if $\|\veca\|_1=\sum_{i=1}^k a_i \in L$. 
If $L$ contains $0$, or if $L$ contains both positive and negative elements, then $f$ supports a one-wise independent distribution.\footnote{Indeed, if $\ell_1,\ell_2\in L$, where $\ell_1<0$ and $\ell_2>0$, then a distribution $\cD$ that with probability $p=\ell_2/(\ell_2-\ell_1)$ samples a random $\veca$ of Hamming weight $ \|\veca\|_1=\ell_1$ and with probability $1-p$ samples a random $\veca$ of weight $\|\veca\|_1=\ell_2$ is one-wise independent and is supported on $f^{-1}(1)$.} So we conclude $L$ contains only positive elements or only negative elements. Without loss of generality we consider the case where $L$ contains only positive elements.

Let $\rho = \rho(f)$, first note that both $K_1^Y$ and $K_\rho^N$ are symmetric since $f$ is symmetric. Thus, by the convexity of the sets, it suffices to consider vectors of the form $\mu^k = (\mu,\mu,\ldots,\mu)$ in $K_1^Y$ and $K_\rho^N$. Since $L$ contains only positive elements, it follows that for $\mu^k \in K_1^Y$, we must have $\mu > 0$. To prove that $\maxf$ is not approximation resistant, it suffices to show that for $\mu > 0$, $\mu^k$  is not contained in $K_\rho^N$. Consider a distribution $\cD \in S_\rho^N$ with $\mu(\cD) = \mu^k$. It can be shown by elementary Fourier analysis that if $\veca\sim\bern(1/2+\eps)^k$ and $\vecb \sim \cD$ then 
$$\Exp_{\vecb\sim \cD}\Exp_{\veca \sim \bern(p)^k}[f(\vecb \odot \veca)] = \rho + \Omega(\mu\tau\epsilon) - O(\epsilon^2),$$
where $\tau$ is the sum of the first level Fourier coefficients of $f$ (i.e., $\tau=\sum_{||w||_1 = 1} \hat{f}(w)$), and the $\Omega(\cdot)$ and $O(\cdot)$ notations hide constants depending on $f$ and $\cD$, but not on $\eps > 0$. Due to the symmetry of $f$, all the first level Fourier coefficients are equal, and due to the positivity of $L$, all these coefficients are positive. It follows that for some sufficiently small $\epsilon > 0$, the expected probability of satisfying a constraint is strictly larger than $\rho$ thus proving $\mu^k \not\in K_\rho^N$. We conclude $K_1^Y \cap K_\rho^N = \emptyset$, and so $\maxf$ is not approximation-resistant.
\end{proof}

\section{Preliminaries}\label{sec:preliminaries}

We will follow the convention that $n$ denotes the number of variables in the CSP as well as the communication game, $m$ denotes the number of constraints in the CSP, and $k$ denotes the arity of the CSP.  We use $\N$ to denote the set of natural numbers $\{1,2,3,\ldots\}$ and use $[n]$ to denote the set $\{1,2,\dots,n\}$. By default, the Boolean variable in this paper takes value in $\{-1,1\}$.

For variables of a vector form, we write them in boldface, \textit{e.g.,} $\vecx\in\{-1,1\}^n$, and its $i$-th entry is written without boldface, \textit{e.g.,} $x_i$. For variable being a vector of vectors, we write it, for example, as $\vecb=(\vecb(1),\vecb(2),\dots,\vecb(m))$ where $\vecb(i)\in\{-1,1\}^k$. The $j$-th entry of the $i$-th vector of $\vecb$ is then written as $\vecb(i)_j$. Let $\vecx$ and $\vecy$ be two vectors of the same length, $\vecx\odot\vecy$ denotes the entry-wise product of them.

For every $p\in[0,1]$, $\textsf{Bern}(p)$ denotes the Bernoulli distribution that takes value $1$ with probability~$p$ and takes value $-1$ with probability $1-p$.

\subsection{Approximate Constraint Satisfaction}

%\mnote{Allow weights}

Let $f:\{-1,1\}^k\rightarrow\{0,1\}$ be a Boolean constraint function of arity $k$ and $x_1,\dots,x_n$ be variables. A constraint $C$ consists of $\vecj=(j_1,\dots,j_k)\in[n]^k$ and $\vecb=(b_1,\dots,b_k)\in\{-1,1\}^k$ where the $j_i$'s are distinct. The constraint $C$ reads as requiring $f(\vecb\odot\vecx|_\vecj)=f(b_1x_{j_1},\dots,b_kx_{j_k})=1$. A \textsf{Max-CSP}($f$) instance $\Psi$ contains $m$ constraints $C_1,\dots,C_m$ with non-negative weights $w_1,\ldots,w_m$ where $C_i=(\vecj(i),\vecb(i))$ and $w_i \in \R$ for each $i\in[m]$. For an assignment $\vecsigma\in\{-1,1\}^n$, the value $\val_\Psi(\vecsigma)$ of $\vecsigma$ on $\Psi$ is the fraction of weight of constraints satisfied by $\vecsigma$, i.e., $\val_\Psi(\vecsigma)=\tfrac{1}{W}\sum_{i\in[m]}w_i \cdot f(\vecb(i)\odot\vecsigma|_{\vecj(i)})$, where $W = \sum_{i=1}^m w_i$. The optimal value of $\Psi$ is defined as $\val_\Psi=\max_{\vecsigma\in\{-1,1\}^n}\val_\Psi(\vecsigma)$. The approximation version of \textsf{Max-CSP}($f$) is defined as follows.

%We first describe how instances $\Psi = (C_1,\ldots,C_m;w_1,\ldots,w_m)$ are presented in the streaming setting. 

%\begin{remark}\label{rem:stream-instance}
%In the insertion only setting we only consider the unweighted case. (Note this is a setting we only use for our lower bounds, so this makes our results stronger.) Here the input $\Psi = (C_1,\ldots,C_m)$ is presented as a stream $\sigma_1,\ldots,\sigma_m$ with $\sigma_i = C_i$ where the $C_i$'s are distinct. 
%In the dynamic setting, the input $\Psi = (C_1,\ldots,C_m; w_1,\ldots,w_m)$ is obtained by inserting and deleting (unweighted) constraints, possibly with repetitions and thus leading to a (integer) weighted instance. 
%Formally $\Psi= (C_1,\ldots,C_m; w_1,\ldots,w_m)$ is presented as a stream $\sigma_1,\ldots,\sigma_\ell$ where $\sigma_t = (C'_t,w'_t)$ and $w'_t \in \{-1,1\}$ such that $w_i = \sum_{t \in [\ell] : C_i = C'_t} w'_t$. For the algorithmic results we require that $w_i$'s are non-negative at the end of the stream. The lower bounds however work in the ``strict turnstile'' setting where at all times, the weight of every constraint is non-negative, i.e., for all $\ell' \leq \ell$ and for all constraints $C$  we have $\sum_{t \in [\ell'] : C = C'_t} w'_t \geq 0$. Given such representations of our instance, we now describe our streaming task.
%\end{remark}

\begin{definition}[$(\gamma,\beta)$-$\maxf$]
Let $f\colon\{-1,1\}^k\to\{0,1\}$ be a constraint function and $0\leq\beta<\gamma\leq1$. For each $m\in\N$, let $\Gamma_m=\{\Psi=(C_1,\dots,C_m; w_1,\ldots,w_m)\, |\, \val_\Psi\geq\gamma\}$ and $B_m=\{\Psi=(C_1,\dots,C_m; w_1,\ldots,w_m)\, |\, \val_\Psi\leq\beta\}$.

The task of $(\gamma,\beta)$-$\maxf$ is to distinguish between instances from $\Gamma=\cup_{m\leq\textsf{poly}(n)}\Gamma_m$ and instances from $B=\cup_{m\leq\textsf{poly}(n)}B_m$. Specifically we desire algorithms that output $1$ w.p. at least $2/3$ on inputs from $\Gamma$ and output $1$ w.p. at most $1/3$ on inputs from $B$. 
\end{definition}
Let $\rho(f)=2^{-k}\cdot|\{\veca\in\{-1,1\}^k\, |\, f(\veca)=1\}|$ denote the probability that a uniformly random assignment satisfies $f$. We say $f$ is \textit{streaming-approximation-resistant}  if for every $\epsilon > 0$, the $(1,\rho(f)+\epsilon)$-$\maxf$ requires $\Omega(n^\delta)$ space for some constant $\delta>0$.

We now define streaming and sketching algorithms in the context of $\maxf$. Note that the input to both algorithms are sequences of constraints. We use
$\C_{\cF,n}$ to denote the set of all constraints of $\maxf$ on $n$ variables. A stream is thus an element of $(\C_{\cF,n})^*$ and we use $\lambda$ to denote the empty stream. 

\begin{definition}[Streaming algorithm]\label{def:streaming alg}
A space $s$ general streaming algorithm $\ALG$ for $\maxf$ on $n$ variables is given by a (state-evolution) function $S: \{0,1\}^s \times \C_{\cF,n} \to \{0,1\}^s$ and a (output) function $v:\{0,1\}^s \to [0,1]$.
Let  $\wtS: (\C_{\cF,n})^*\to \{0,1\}^s$ given by $\wtS(\lambda) = 0^s$ and
$\wtS(\sigma_1,\ldots,\sigma_m)) = S(\wtS(\sigma_1,\ldots,\sigma_{m-1}),\sigma_m)$ denote the iterated state-evolution map.
Then the output of $\ALG$ on input $\sigma = (\sigma_1,\ldots,\sigma_m)$ is $v(\wtS(\sigma))$.
For the purposes of this paper, a randomized streaming algorithm is simply a distribution on the pairs $(S,v)$. 
\end{definition}

Sketching algorithms are a special class of streaming algorithms that have been widely used in both upper bounds and lower bounds.

\begin{definition}[Sketching algorithms]\label{def:sketching alg}
A (deterministic) space $s$ streaming algorithm $\ALG=(S,v)$ is a sketching algorithm if there exists a compression function $\COMP:(\C_{\cF,n})^* \to \{0,1\}^s$ and a combination function $\COMB:\{0,1\}^s\times\{0,1\}^s \to \{0,1\}^s$ such that the following hold:
\begin{itemize} 
\item $S(z,C) = \COMB(z,\COMP(C))$ for every $z \in \{0,1\}^s$ and $C \in \C_{\cF,n}$.
\item For every pair of streams $\sigma,\tau\in (\C_{\cF,n})^*$, we have
\[
\COMB(\COMP(\sigma),\COMP(\tau))=\COMP(\sigma\circ\tau)
\]
where $\sigma\circ\tau$ represents the concatenation of the streams $\sigma$ and $\tau$. A randomized algorithm $\ALG$ is a randomized sketching algorithm if it is a distribution over deterministic sketching algorithms.
\end{itemize}
\end{definition}

We remark that a linear sketching algorithm roughly
 associates with elements of a vector space $V$ (over some field) and $\COMB$ is simply vector addition in $V$. 

%\mnote{Do we need the following definition? It is distracting.}

%For $\alpha\in[0,1]$, an algorithm $\ALG$ is an $\alpha$-approximation to the $\maxf$ problem if $\ALG$ can solve $(\gamma,\beta)$-$\maxf$ with success probability at least $2/3$ for every $0\leq\beta<\gamma\leq1$ such that $\beta/\gamma\leq\alpha$.

\subsection{Total variation distance}
The total variation distance between probability distributions plays an important role in our analysis.

\begin{definition}[Total variation distance of discrete random variables]
Let $\Omega$ be a finite probability space and $X,Y$ be random variables with support $\Omega$. The total variation distance between $X$ and $Y$ is defined as follows.
\[
\|X-Y\|_{tvd} :=\frac{1}{2}
\sum_{\omega\in\Omega}\left|\Pr[X=\omega]-\Pr[Y=\omega]\right| \, .
\]
\end{definition}
We will use the triangle and data processing inequalities for the total variation distance.
\begin{proposition}[E.g.,{\cite[Claim~6.5]{KKS}}]\label{prop:tvd properties}
For random variables $X, Y$ and $W$:
\begin{itemize}
\item (Triangle inequality) $\|X-Y\|_{tvd}\geq\|X-W\|_{tvd}-\|Y-W\|_{tvd}$.
\item (Data processing inequality) If $W$ is independent of both $X$ and $Y$, and $f$ is a function, then  $\|f(X,W)-f(Y,W)\|_{tvd}\leq\|X-Y\|_{tvd}$.
\end{itemize}
\end{proposition}

\subsection{Concentration inequality}

We will use the following concentration inequality which is essentially an Azuma-Hoeffding style inequality for submartingales. The form we use is based on  \cite[Lemma~2.5]{KK19}, and allows for variables with different expectations. The analysis is a very slight modification of theirs.

\iffalse{ 
\begin{lemma}[{\cite[Lemma~2.5]{KK19}}]\label{lem:azuma}
Let $X=\sum_{i\in[N]}X_i$ where $X_i$ are Bernoulli random variables such that for any $k\in[N]$, $\Exp[X_k \, |\, X_1,\dots,X_{k-1}]\leq p$ for some $p\in(0,1)$. Let $\mu=Np$. For any $\Delta>0$,
\[
\Pr\left[X\geq\mu+\Delta\right]\leq\exp\left(-\frac{\Delta^2}{2\mu+2\Delta}\right) \, .
\]
\end{lemma}

%\mnote{Creating and proving our variant.} 

We also use a slight variant of the above which we prove by modifying the proof from \cite[Lemma~2.5]{KK19} slightly.

}\fi

\begin{lemma}\label{lem:our-azuma}
Let $X=\sum_{i\in[N]}X_i$ where $X_i$ are Bernoulli random variables such that for every $k\in[N]$, $\Exp[X_k \, |\, X_1,\dots,X_{k-1}]\leq p_k$ for some $p_k\in(0,1)$. Let $\mu=\sum_{k=1}^N p_k$. For every $\Delta>0$, we have:
\[
\Pr\left[X\geq\mu+\Delta\right]\leq\exp\left(-\frac{\Delta^2}{2\mu+2\Delta}\right) \, .
\]
\end{lemma}

\begin{proof}
Let $v = \Delta/(\mu+\Delta)$ and $u = \ln (1+v)$. We have
\[ \Exp[e^{uX}] = \Exp[\prod_{k=1}^N e^{uX_k}] \leq (1+p_N(e^u-1))\cdot \Exp[\prod_{k=1}^{N-1} e^{uX_k}] \leq \prod_{i=1}^N (1 +p_k(e^u-1))
= \prod_{i=1}^N (1 +p_kv) \leq e^{v\mu},
\] 
where the final inequality uses $1+x \leq e^x$ for every $x$ (and the definition of $\mu$). 
Applying Markov to the above, we have:
\[\Pr\left[X\geq\mu+\Delta\right]=\Pr\left[e^{uX}\geq e^{u(\mu+\Delta)}\right]\leq \Exp[e^{uX}]/ e^{u(\mu+\Delta)} \leq e^{v\mu-u\mu - u\Delta}.
\]
From the inequality $e^{v-v^2/2} \leq 1+v$ we infer $u \geq v-v^2/2$ and so the final expression above can be bounded as:
\[
\Pr\left[X\geq\mu+\Delta\right]\leq e^{v\mu-u\mu - u\Delta} \leq e^{\frac{v^2}2 (\mu+\Delta) - v \Delta} = e^{-\frac{\Delta^2}{2(\mu+\Delta)}},
\]
where the final equality comes from our choice of $v$.
\end{proof}

%\mnote{End insertion}

\subsection{Fourier analysis}\label{sec:fourier}
We will need the following basic notions from Fourier analysis over the Boolean hypercube (see, for instance,~\cite{o2014analysis}).
For a Boolean function $f: \{-1,1\}^k \to \R$ its Fourier coefficients are defined by $\widehat{f}(\vecv) = \Exp_{\veca\in\{-1,1\}^k}[f(\veca)\cdot(-1)^{\vecv^\top\veca}]$, where $\vecv\in\{0,1\}^k$. We need the following two important tools.

\begin{lemma}[Parseval's identity]\label{prop:parseval}
For every function $f:\{-1,1\}^k\to\R$, 
\[
\|f\|_2^2=\frac{1}{2^k}\sum_{\veca\in\{-1,1\}^k}f(\veca)^2=\sum_{\vecv\in\{0,1\}^k}\widehat{f}(\vecv)^2 \, .
\]
\end{lemma}

Note that for every distribution $f$ on $\{-1,1\}^k$, $\widehat{f}(0^k)=2^{-k}$. For the uniform distribution $U$ on $\{-1,1\}^k$, $\widehat{U}(\vecv)=0$ for every $\vecv\neq0^k$. Thus, by \autoref{prop:parseval}, for any distribution $f$ on $\{-1,1\}^k$:
\begin{align}\label{eq:dist}
\|f-U\|_2^2=\sum_{\vecv\in\{0,1\}^k}\left(\widehat{f}(\vecv)-\widehat{U}(\vecv)\right)^2=\sum_{\vecv\in\{0,1\}^k\backslash\{0^k\}}\widehat{f}(\vecv)^2 \, .
\end{align}

Next, we will use the following consequence of hypercontractivity for Boolean functions as given in \cite[Lemma 6]{GKKRW} which in turns relies on a lemma from \cite{KKL88}. 
\begin{lemma}\label{lem:kkl}
Let $f:\{-1,1\}^n\rightarrow\{-1,0,1\}$ and $A=\{\veca\in\{-1,1\}^n\, |\, f(\veca)\neq0\}$. If $|A|\geq2^{n-c}$ for some $c\in\N$, then for every $\ell\in\{1,\dots,4c\}$, we have
\[
\frac{2^{2n}}{|A|^2}\sum_{\substack{\vecv\in\{0,1\}^n\\\|\vecv\|_1=\ell}}\widehat{f}(\vecv)^2\leq\left(\frac{4\sqrt{2}c}{\ell}\right)^\ell \, .
\]
\end{lemma}

\section{A Streaming Approximation Algorithm for \texorpdfstring{$\maxf$}{Max-CSP(f)}}\label{sec:algorithm}

%\mnote{Done with this section. Pls. check!}
%\mnote{Section is mostly done except for Lemma 4.4. See also discussion in slack - should we have integer weights only? +1/-1 only?}

%\cnote{Define different streaming models if we haven't done this before.}

In this section we give our main algorithmic result --- a $O(\log n)$-space linear sketching algorithm for $(\gamma,\beta)$-$\maxf$ if $K_\gamma^Y = K_\gamma^Y(f)$ and $K_\beta^N = K_\beta^N(f)$ are disjoint. (See \autoref{def:marginals}.)

The algorithm in fact works in the (general) dynamic setting where the input $\Psi = (C_1,\ldots,C_m; w_1,\ldots,w_m)$ is obtained by inserting and deleting (unweighted) constraints, possibly with repetitions and thus leading to a (integer) weighted instance. Formally $\Psi= (C_1,\ldots,C_m; w_1,\ldots,w_m)$ is presented as a stream $\sigma_1,\ldots,\sigma_\ell$ where $\sigma_t = (C'_t,w'_t)$ and $w'_t \in \{-1,1\}$ such that $w_i = \sum_{t \in [\ell] : C_i = C'_t} w'_t$. For the algorithmic result to hold, we require that $w_i$'s are non-negative at the end of the stream but the intermediate values can be arbitrary. Furthermore the algorithm requires that the length of the stream be polynomial in $n$ (or else there will be a logarithmic multiplicative factor in the length of the stream in the space usage). 

We state our main theorem of this section which simply restates Part (1) of \autoref{thm:main-detailed}.

%\cnote{State the following in terms of dynamic setting.}
\begin{theorem}\label{thm:main-positive}
For every function $f:\{-1,1\}^k \to \{0,1\}$ and for every $0\leq\beta<\gamma\leq1$, if $K_\gamma^Y(f) \cap K_\beta^N(f) = \emptyset$, then $(\gamma,\beta)$-$\maxf$ admits a probabilistic streaming algorithm in the dynamic setting that uses $O(\log n) $ space and succeeds with probability at least $2/3$. \label{positive-result}
\end{theorem}

The overview of the algorithm is as follows: We use the separability of $K_\gamma^Y$ and $K_\beta^N$ to obtain a hyperplane with normal vector $\veclambda$ that separates the two sets.
We then estimate a $\veclambda$-weighted bias of a given instance $\Psi$ and accept $\Psi$ if this bias falls on the $K_\gamma^Y$ side of the hyperplane. We note that the bias can be approximated arbitrarily well using well-known $\ell_1$-norm approximators in the turnstile setting. The bulk of the work is in analyzing the correctness of our algorithm.

We will use the following streaming algorithm for approximating the~$\ell_1$ norm of a vector.
\begin{proposition}[{\protect \cite{Indyk},\cite[Theorem 2.1]{KNW10}}]
\label{prop:ell1}
Given a stream $S$ of $\poly(n)$ updates $(i, v) \in [n] \times \{-M,-(M-1),\ldots,M-1,M\}$ where $M=\textsf{poly}(n)$, let $x_i=\sum_{(i,v)\in S} v$ for $i\in[n]$. For every $\eps>0$, there exists a linear sketch that uses $O(\log{n})$ bits of memory and outputs a $(1\pm\epsilon)$-approximation to the value $\|x\|_1=\sum_i |x_i|$ with probability at least $2/3$.
\end{proposition}

\subsection{Algorithm}
Let us start with the definition of $\veclambda$-bias.
\begin{definition}[Bias (vector)]
For $\veclambda = (\lambda_1,\ldots,\lambda_k) \in \R^k$, and instance $\Psi = (C_1,\ldots,C_m; w_1,\ldots,w_m)$ of $\maxf$ where $C_i = (\vecj(i),\vecb(i))$ and $w_i \geq 0$, we let the {\em $\veclambda$-bias vector} of $\Psi$, denoted $\bias_{\veclambda}(\Psi)$, be the vector in $\R^n$ given by 
\[
\bias_{\veclambda}(\Psi)_\ell = \frac{1}W \cdot \sum_{i \in [m], t \in [k] : j(i)_t = \ell} \lambda_t w_i b(i)_t \, ,
\]
for $\ell \in [n]$, where $W = \sum_{i \in [m]} w_i$. 
The $\veclambda$-bias of $\Psi$, denoted $B_{\veclambda}(\Psi)$, is the $\ell_1$ norm of $\bias_{\veclambda}(\Psi)$, i.e., $B_{\veclambda}(\Psi) = \sum_{\ell=1}^n |\bias_{\veclambda}(\Psi)_\ell|$.
\end{definition}

By directly applying the known $\ell_1$-sketching algorithm (i.e.,~\autoref{prop:ell1}), the following lemma shows that $\veclambda$-bias can be estimated in $O(\log n)$ space.
\begin{lemma}\label{prop:ell1norm}
For every vector $\veclambda \in \R^k$ and $\epsilon > 0$, there exists a $O(\log n)$ 
space algorithm $\cA$ that on input a stream $\sigma_1,\ldots,\sigma_\ell$, representing an instance $\Psi = (C_1,\ldots,C_m;w_1,\ldots,w_m)$, outputs a $(1\pm \epsilon)$-approximation to $B_{\veclambda}(\Psi)$, i.e., for every $\Psi$, $(1-\epsilon)B_{\veclambda}(\Psi) \leq \cA(\Psi) \leq (1+\epsilon)B_{\veclambda}(\Psi)$, with probability at least $2/3$. 
\end{lemma}

\begin{proof}

Note that since $k$ and $\epsilon$ are constants with respect to $n$, we can without loss of generality assume that each entry of $\lambda$ is an integer and $\epsilon$ has constant bit complexity.~\footnote{Concretely, round $\epsilon$ to $2^{-t}$ where $t$ is the smallest integer such that $\epsilon\geq2^{-t}$. As for $\veclambda$, let $\lambda_{\min}=\min_{j\in[k]}|\lambda_j|$ and round it the same way as we did for $\epsilon$. Next, for each $j\in[k]$, scale and round $\lambda_j$ to $\ceil{\frac{4\lambda_j}{\lambda_{\min}}}$. It is not difficult to verify that scaling down the new $\veclambda$-bias by a factor of $\lambda_{\min}/4$, it is a $(1\pm\epsilon/2)$-approximation to the original $\veclambda$-bias.}

On input a stream $\sigma_1,\ldots,\sigma_\ell$ representing an instance $\Psi=(C_1,\ldots,C_m;w_1,\ldots,w_m)$ (see \autoref{rem:stream-instance}) with $\sigma_i = (C'_i = (\vecj(i),\vecb(i)),w'_i)$, the algorithm $\cA$ proceeds as follows. It implicitly maintains a vector $\vecv \in \R^n$ which is initially zero. Each stream element $\sigma_i$ is converted into $k$ updates to $\vecv$ given by $(\vecj(i)_1,w'_i\cdot\lambda_1),\ldots,(\vecj(i)_k,w'_i\cdot\lambda_k)$ (where the notation of ``updating by $(i,x)$" indicates that $x$ is added to $v_i$). It then applies the algorithm from \autoref{prop:ell1} to compute a $(1\pm \epsilon)$ approximation $B'$ to $\|v\|_1 =  \sum_{i \in [m], t \in [k] : j(i)_t = \ell} \lambda_t w_i b(i)_t$. (Note that since $\ell=\textsf{poly}(n)$ and $k$ is a constant, we know that there are only $\textsf{poly}(n)$ updates and each update is a constant integer and so the conditions of \autoref{prop:ell1} are satisfied, and so $B'$ is a $(1\pm\epsilon)$ approximation to $\|v\|_1$ with probability at least $2/3$.) Finally $\cA$ outputs $B'/W$ which is a $(1\pm\epsilon)$-approximation to $B_{\veclambda}(\Psi)$ if and only if $B'$ is a $(1\pm\epsilon)$ approximation to $\|v\|_1$.
\end{proof}

We will use the following form of the hyperplane separation theorem for convex bodies (see, e.g., \cite[Exercise 2.22]{boyd2004convex}).

\begin{proposition}\label{prop: Hyperplane separation}
Let $K^Y$ and $K^N$ be two disjoint nonempty closed convex sets in $\R^k$ at least one of which is compact. Then there exists a nonzero vector $\veclambda = (\lambda_1,\ldots,\lambda_k)$ and real numbers $\tau_Y > \tau_N$ such that 
\[
 \forall \vecx \in K^Y, ~~ \langle \veclambda,\vecx \rangle \ge \tau_Y \text{~~and~~ } \forall \vecx \in K^N, ~~ \langle \veclambda,\vecx \rangle \le \tau_N \, .
\]
\end{proposition}

We are now ready to describe our algorithm for $(\gamma,\beta)$-$\maxf$. 

\begin{algorithm}[H]
	\caption{A streaming algorithm for $(\gamma,\beta)$-$\maxf$}
	\label{alg:main alg}
\begin{algorithmic}[1]
		\Input a stream $\sigma_1,\ldots,\sigma_\ell$ representing an instance $\Psi$ of $\maxf$. 
		    \State Let $\veclambda \in \R^k$ and $\tau_N < \tau_Y$ be as given by~\autoref{prop: Hyperplane separation} separating $K_\gamma^Y(f)$ and $K_\beta^N(f)$.
			\State Let $\epsilon = \frac{\tau_Y-\tau_N}{2(\tau_Y+\tau_N)}$ (so that $(1-\epsilon)\tau_Y > (1+\epsilon)\tau_N$). 
			\State Use the algorithm $\cA$ from \autoref{prop:ell1norm} to  compute $\tilde{B}$ to be a $(1\pm\epsilon)$ approximation to $B_{\veclambda}(\Psi)$, i.e., $(1-\epsilon)B_{\veclambda}(\Psi) \leq \tilde{B} \leq (1+\epsilon)B_{\veclambda}(\Psi)$ with probability at least $2/3$. 
			\If{$\tilde{B} \leq \tau_N (1+\epsilon)$}
				\Statex {\bf Output:} NO.
			\Else
				\Statex {\bf Output:} YES.
			\EndIf
	\end{algorithmic}
\end{algorithm}

It is clear that the algorithm above runs in $O(\log n)$ space (in particular by using~\autoref{prop:ell1} via \autoref{prop:ell1norm} for Step 3). We now turn to analyzing the correctness of the algorithm.

\subsection{Analysis of the correctness of Algorithm~\ref{alg:main alg}}

\begin{lemma}\label{lem:correctness_algorithm}
\autoref{alg:main alg} correctly solves $(\gamma,\beta)$-$\maxf$, if $K_\gamma^Y(f)$ and $K_\beta^N(f)$ are disjoint.  
Specifically, for every $\Psi$, let $\tau_Y,\tau_N,\epsilon,\veclambda,\tilde{B}$ be as given in~\autoref{alg:main alg}, we have:
\begin{eqnarray*}
\val_\Psi \geq \gamma & \Rightarrow & B_{\veclambda}(\Psi) \geq \tau_Y \mbox{ and } \tilde{B} > \tau_N(1+\epsilon) \, , \\
\mbox{ and } \val_\Psi \leq \beta & \Rightarrow & B_{\veclambda}(\Psi) \leq \tau_N \mbox{ and } \tilde{B} \leq \tau_N(1+\epsilon) \, ,
\end{eqnarray*}
provided $(1-\epsilon)B_\lambda(\Psi) \leq \tilde{B} \leq (1+\epsilon)B_\lambda(\Psi)$.
\end{lemma}
In the rest of this section, we will prove~\autoref{lem:correctness_algorithm}. 
The key to our analysis is a distribution $\cD(\Psi^\veca) \in \Delta(\{-1,1\}^k)$ that we associate with every instance $\Psi$ and assignment $\veca\in \{-1,1\}^n$ to the variables of $\Psi$. Recall that in~\autoref{def:marginals}, we define $\vecmu(\cD)=(\mu_1,\dots,\mu_k)$ where $\mu_i=\Exp_{\vecb\sim\cD}[b_i]$. If $\Psi$ is $\gamma$-satisfied by assignment $\veca$, we prove that $\vecmu(\cD(\Psi^\veca)) \in K_\gamma^Y$. On the other hand, if $\Psi$ is not $\beta$-satisfiable by any assignment, we prove that for every $\veca$, $\vecmu(\cD(\Psi^\veca)) \in K_\beta^N$. Finally we also show that the bias $B_{\veclambda}(\Psi)$ relates to $\veclambda(\cD(\Psi^\veca)) \triangleq \langle\vecmu(\cD(\Psi^\veca)), \veclambda \rangle$, where the latter quantity is exactly what needs to be computed (by \autoref{prop: Hyperplane separation}) to distinguish the membership of $\vecmu(\cD(\Psi^\veca))$ in $K_\gamma^Y$ versus the membership in $K_\beta^N$.

We start with recalling some notations. For an instance $\Psi=(C_1,\ldots,C_m;w_1,\ldots,w_m)$ on $n$ variables with $C_i = (\vecj(i),\vecb(i))$, and an assignment $\veca \in \{-1,1\}^n$, let $\Psi^\veca$ denote the new instance obtained by flipping the variables according to $\veca$. Specifically $\Psi^\veca = (C_1^\veca,\ldots,C_m^\veca;w_1,\ldots,w_m)$ where $C_i^\veca = (\vecj(i),\veca|_{\vecj(i)} \odot \vecb(i))$. 

Given instance $\Psi$, let $\cD(\Psi) \in \Delta(\{-1,1\}^k)$ be the distribution obtained by sampling a constraint at random (according to its weight) from $\Psi$ and outputting the ``negation pattern''. Formally, to sample a random vector $\vecb\sim\cD(\Psi)$, we sample $i \in [m]$ with probability $w_i/W$ where $W=\sum_{i \in [m]} w_i$, and output $\vecb(i)$ where $C_i = (\vecj(i),\vecb(i))$.

The next lemma relates the $\veclambda$-bias vector of $\Psi$ to $\veclambda(\cD(\Psi^\veca))$ and uses this to relate the bias of $\Psi$ to the maximum over $\veca$ of $\veclambda(\cD(\Psi^\veca))$.

\begin{lemma}\label{lem:bias_comparsion}
For every vector $\veca\in\{-1,1\}^n$, we have $\veclambda(\cD(\Psi^\veca)) = \langle \veca, \bias_{\veclambda}(\Psi) \rangle$. Consequently we have $B_{\veclambda}(\Psi) = \max_{\veca\in\{-1,1\}^n} \{ {\veclambda}(\cD(\Psi^\veca)) \}$.
\end{lemma}

\begin{proof}%[Proof of~\autoref{lem:bias_comparsion}]
We start with the first equality. Fix $\veca \in \{-1,1\}^n$. We have
\begin{align*}
\veclambda(\cD(\Psi^\veca)) &= \langle \vecmu(\cD(\Psi^\veca)),\veclambda \rangle ~~~\mbox{(By definition of $\lambda(\cdot)$)}\\
    &= \Exp_{\vecy \sim \cD(\Psi^\veca)}[\langle \vecy, \veclambda \rangle]~~~\mbox{(By definition of $\vecmu(\cD)$ and linearity of inner product})\\
    &= \Exp_{i} \left[\langle \vecb^\veca(i), \veclambda \rangle\right]~~~\mbox{(By definition of $\cD(\Psi^\veca)$)}\\
    &= \Exp_{i} \left[\sum_{t\in[k]} b^\veca(i)_t\cdot \lambda_t \right]~~~\mbox{(Expanding the inner product)}\\
    &= \frac1W \sum_{i \in [m]} w_i \sum_{\ell \in [n]} \sum_{t \in [k]} \One_{\vecj(i)_t = \ell} \cdot \lambda_t \cdot a_\ell \cdot b(i)_t~~~\mbox{(Using definition of $\Psi^\veca$)}\\
    & = \frac1W\sum_{\ell\in[n]} a_\ell \sum_{t\in[k]}\lambda_t \sum_{i\in[m]} \One_{\vecj(i)_t = \ell} \cdot w_i \cdot b(i)_t~~~\mbox{(Exchanging summations)}\\
    & = \sum_{\ell \in [n]} a_\ell\cdot \bias_{\veclambda}(\Psi)_\ell~~~\mbox{(By definition of $\bias_{\veclambda}(\cdot)$)}\\
    & = \langle \veca, \bias_{\veclambda}(\Psi) \rangle\, ,
\end{align*}
yielding the first equality.

The second part is immediate from the observation that for every vector $\vecv \in \R^n$, we have $||\vecv||_1 = \max_{\veca \in \{-1,1\}^n} \langle \veca, \vecv \rangle$ and so 
\[
B_{\veclambda}(\Psi) = ||\bias_{\veclambda}(\Psi)||_1 = \max_{\veca \in \{-1,1\}^n} \{\langle \veca, \bias_{\veclambda}(\Psi)\rangle\} = \max_{\veca \in \{-1,1\}^n} \{{\veclambda}(\cD(\Psi^\veca))\} \, .
\]
\end{proof}

We now turn to connecting $\val_\Psi$ to properties of $\cD(\Psi^\veca)$.

\begin{lemma}\label{lem:optimum_YES}
For every $\Psi$ and $\veca$, if $\val_\Psi(\veca) \geq \gamma$ then $\cD(\Psi^\veca) \in S_\gamma^Y$. 
\end{lemma}

\begin{proof}
Follows from the fact that 
\[
\Exp_{\vecb \sim \cD(\Psi^\veca)} [f(\vecb)] = \frac{1}{W}\sum_{i\in[m]}w_i \cdot f(\vecb(i)\odot \veca|_{\vecj(i)})  = \frac{1}{W} \sum_{i\in[m]} w_i \cdot C_i(\veca) = \val_\Psi(\veca) \geq \gamma \, ,
\]
implying $\cD(\Psi^\veca)\in S_\gamma^Y$.
\end{proof}

\begin{lemma}\label{lem:optimum_NO}
For every $\Psi$, if $\val_{\Psi} \leq \beta$, then for all $\veca$, we have $\cD(\Psi^\veca) \in S_\beta^N$.
\end{lemma}

\begin{proof}
We claim if $\val_{\Psi} \leq \beta$, then  $\cD(\Psi) \in S_\beta^N$.
This suffices to prove the lemma, since for every $\veca\in\{-1,1\}^n$ we have $\val_{\Psi^\veca} = \val_{\Psi}$. So if 
$\val_\Psi\leq \beta$ then $\val_{\Psi^\veca} \leq \beta$ and so by the claim above applied to $\Psi^\veca$, we have $\cD(\Psi^\veca) \in S_\beta^N$.  

We prove the contrapositive, i.e., we assume $\cD(\Psi) \not\in S_\beta^N$ and show this implies $\val_\Psi > \beta$.
If $\cD(\Psi) \not\in S_\beta^N$, then there exists $p \in [0,1]$ such that $\Exp_{\vecb \sim \cD(\Psi)} \Exp_{\vecc \sim \bern(p)^{k}} [f(\vecb \odot \vecc)] > \beta$. But this implies, as we show below, that if $\vecsigma \sim \bern(p)^{n}$, then $\Exp_{\vecsigma\sim\bern(p)^{n}}[\val_{\Psi}(\vecsigma)] > \beta$. 
We have:
\begin{align*}
    \Exp_{\vecsigma\sim\bern(p)^{n}}[\val_{\Psi}(\vecsigma)]
    &= \Exp_{\vecsigma\sim\bern(p)^{n}} \Exp_{i}[C_i(\vecsigma)]~~~\mbox{(By definition of $\Psi$)} \\
    &= \Exp_{\vecsigma\sim\bern(p)^{n}} \Exp_{i}[f(\vecb(i) \odot\vecsigma|_{\vecj(i)})]~~~\mbox{(By definition of $C_i$)} \\
    &= \Exp_{i}  \Exp_{\vecsigma|_{\vecj(i)}\sim \bern(p)^{k}} [f(\vecb(i) \odot \vecsigma|_{\vecj(i)})]~~~\mbox{(Exchanging summations)} \\
    &= \Exp_{i}  \Exp_{\vecc\sim \bern(p)^{k}} [f(\vecb(i) \odot \vecc)]~~~\mbox{(Renaming variables)} \\
    &= \Exp_{\vecb\sim\cD(\Psi)} \Exp_{\vecc\sim \bern(p)^{k}} [f(\vecb \odot \vecc)]~~~\mbox{(By definition of $\cD(\Psi)$)} \\
    &> \beta ~~~\mbox{(By the contrapositive assumption)}
\end{align*}
Since $\val_\Psi \triangleq \max_{\vecsigma} \{\val_\Psi(\vecsigma)\} \geq \Exp_{\vecsigma\sim\bern(p)^{n}}[\val_{\Psi}(\vecsigma)]$, we get a 
contradiction to $\val_\Psi\leq\beta$. This concludes the proof of the claim and hence the lemma.
\end{proof}

Before turning to the proof of~\autoref{lem:correctness_algorithm}, we first do a quick post-analysis of the proof above. The proof above is the key reason why the definition of $S_\beta^N$ is chosen as it is: In particular, from the fact that there was an i.i.d. distribution, namely $\bern(p)^{k}$, according to which a random assignment satisfied the ``instance'' underlying $\cD(\Psi)$ with value more than $\beta$ allowed us to extend this to a (again i.i.d., but this was not necessary) distribution over assignments to $\Psi$ that also achieved value of at least $\beta$. Note that the mere existence of an assignment of value greater than $\beta$ on $\cD(\Psi)$ would have been insufficient for this step to go through, explaining our choice of definition of $S_\beta^N$. 

We are now ready to prove~\autoref{lem:correctness_algorithm}.

\begin{proof}[Proof of~\autoref{lem:correctness_algorithm}]
Let $\val_\Psi \geq \gamma$. Then there exists $\veca\in\{-1,1\}^n$ such that $\val_\Psi(\veca) \geq \gamma$. By~\autoref{lem:optimum_YES}, we have that $\cD(\Psi^\veca) \in S_\gamma^Y$. By our choice of $\veclambda$, we have $\lambda(\cD) \geq \tau_Y$ for every $\cD \in S_\gamma^Y$ and so in particular we have $\veclambda(\cD(\Psi^\veca)) \geq \tau_Y$. By~\autoref{lem:bias_comparsion}, we have $B_{\veclambda}(\Psi) = \max_{\vecc \in \{-1,1\}^n}\{\veclambda(\cD(\Psi^\vecc))\}$. Putting these together we have 
\[
B_{\veclambda}(\Psi) = \max_{\vecc \in \{-1,1\}^n}\{\veclambda(\cD(\Psi^\vecc))\}\geq \veclambda(\cD(\Psi^\veca)) \geq \tau_Y \, .
\]
Finally, since $\tilde{B} \geq (1-\epsilon)B_\lambda(\Psi)$, we get $\tilde{B} \geq (1-\epsilon)\tau_Y > (1+\epsilon)\tau_N$, where the final inequality holds by our choice of $\epsilon$. 

The case $\val_\Psi \leq \beta$ is similar. In this case, by~\autoref{lem:optimum_NO} we have $\cD(\Psi^\veca) \in S_\beta^N$ for every $\veca$. Now applying~\autoref{lem:bias_comparsion} we get that for every $\veca$, $\langle \veca, \bias_{\veclambda} \rangle = \veclambda(\cD(\Psi^\veca)) \leq \tau_N$. We conclude that $B_{\veclambda}(\Psi) = \max_{\veca\in\{-1,1\}^n} \{\langle \veca, \bias_{\veclambda} \rangle\} \leq \tau_N$. since $\tilde{B} \leq (1+\epsilon)B_{\veclambda}(\Psi)$, we get $\tilde{B} \leq (1+\epsilon)\tau_N$.
\end{proof}

We now conclude the section with a formal proof of~\autoref{thm:main-positive}.

\begin{proof}[Proof of~\autoref{thm:main-positive}]
The desired algorithm is Algorithm 1. Its space complexity is bounded by the space required for Step 3, which by~\autoref{prop:ell1norm} is $O(\log n)$. Assuming Step 3 works correctly, which happens with probability at least $2/3$,~\autoref{lem:correctness_algorithm} shows that it correctly solves $(\gamma,\beta)$-$\maxf$ whenever $K_\gamma^Y(f) \cap K_\beta^N(f) = \emptyset$. 
\end{proof}

\section{Sketching and Streaming Space Lower Bounds for \texorpdfstring{$\maxf$}{Max-CSP(f)}}\label{sec:SpaceLowerBound insertion}

In this section, we prove our main lower bound results, modulo a communication complexity lower bound which is proved in \autoref{sec:BHBHM 1 wise} and \autoref{sec:BHBHM general}. We start by recalling the results to be proved. First we restate the lower bound in the general streaming setting.
Recall that $(\cD_Y,\cD_N)$ form a padded one-wise pair if there exist $\tau \in [0,1]$, and $\cD_0, \cD'_Y, \cD'_N$ such that for $i \in \{Y,N\}$ we have $\cD_i = \tau \cD_0+ (1-\tau)\cD'_i$ and $\cD'_i$ has uniform marginals.

\restatethmmainnegative*

We also restate the lower bound against sketching algorithms from \autoref{thm:main-detailed} as a separate theorem below.

\begin{theorem}[Lower bounds against sketching algorithms]\label{thm:main-negative dynamic}
For every function $f:\{-1,1\}^k\rightarrow\{0,1\}$ and for every $0\leq\beta<\gamma\leq1$, if $K_\gamma^Y(f) \cap K_\beta^N(f) \neq \emptyset$, then for every $\epsilon>0$, every sketching algorithm for $(\gamma-\eps,\beta+\eps)$-$\maxf$ requires $\Omega(\sqrt{n})$ space\footnote{The constant hidden in the $\Omega$ notation may depend on $k$ and $\epsilon$.}. Furthermore, if $\gamma = 1$. then $(1,\beta+\epsilon)$-$\maxf$ requires $\Omega(\sqrt{n})$ space. 
\end{theorem}

%\mnote{Update this at the end.}
To prove both theorems, we introduce the \textit{Randomized Mask Detection (RMD)} communication game in \autoref{sec:BHBHM definition and theorems}. We then state a lower bound for the communication complexity of this game (\autoref{thm:communication lb matching moments}), and use the lower bound to prove \autoref{thm:main-negative} in \autoref{sec:lb insert} and  \autoref{thm:main-negative dynamic} in \autoref{sec:proof_dynamic}. The proof of \autoref{thm:communication lb matching moments} appears in \autoref{sec:BHBHM general}.

%\mnote{In the figure below perhaps exchange the two boxes in the bottom left? Seems unreasonable to "prove" a bound without defining it.}
%\cnote{Just updated.}
%\mnote{Excellent - thanks!}

\begin{figure}[h]
    \centering
    \includegraphics[width=15cm]{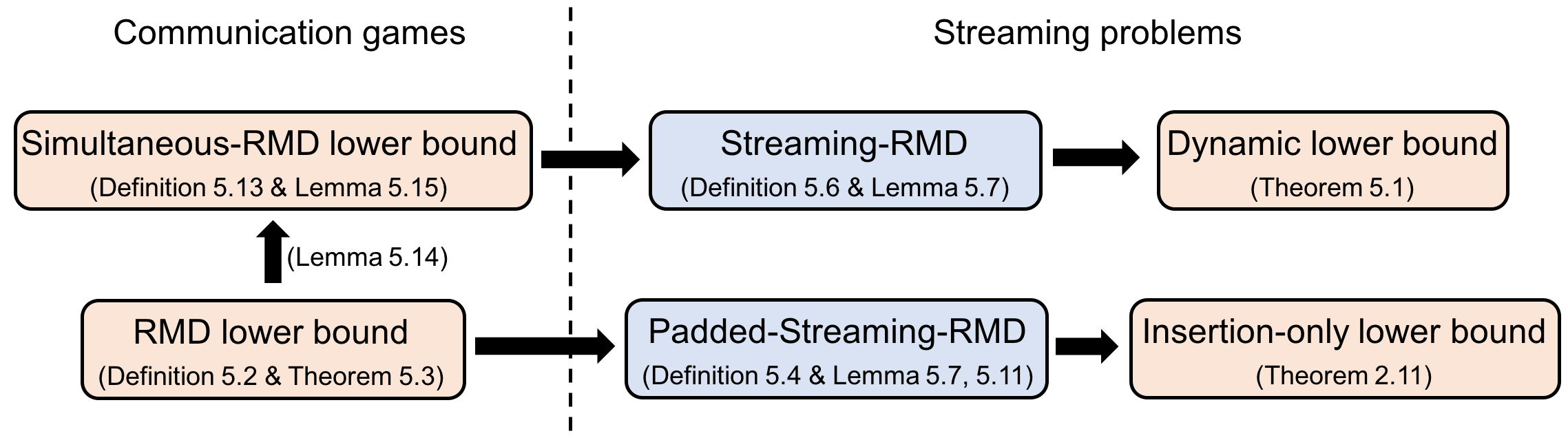}
    \caption{Roadmap of this section.}
    \label{fig:outline}
\end{figure}

\subsection{2-Player Communication Games and the Randomized Mask Detection Problem}\label{sec:BHBHM definition and theorems}

In most of this section and the rest of this paper, we will be considering the complexity of 2-player 1-way communication games. Broadly such games are described by two (parameterized set of) distributions $\cY$ and $\cN$. An instance of the game is a pair $(X,Y)$ either drawn from $\cY$ or from $\cN$ and $X$ is given as input to Alice and $Y$ to Bob. A (one-way communication) protocol $\Pi = (\Pi_A,\Pi_B)$ is a pair of functions with $\Pi_A(X) \in \{0,1\}^c$ denoting Alice's message to Bob, and $\Pi_B(\Pi_A(X),Y)\in \{\yes,\no\}$ denoting the protocol's output. We denote this output by $\Pi(X,Y)$. The complexity of this protocol is the parameter $c$ specifying the length of $\Pi_A(X)$ (maximized over all $X$). The advantage of the protocol $\Pi$ is the quantity 
$$\left| \Pr_{(X,Y)\sim\cY} [ \Pi(X,Y) = \yes] - \Pr_{(X,Y)\sim \cN} [\Pi(X,Y) = \yes] \right|. $$

The Randomized Mask Detection (RMD) communication game is an instance of such a communication game. Let $n,k\in\mathbb{N}$ and $\alpha\in(0,1)$ with $k\leq n$ and $\alpha k\leq 1$. Alice receives a private input $\vecx^*$ drawn uniformly at random from $\{-1,1\}^n$ while Bob receives private inputs of a $k$-uniform hypermatching of size $\alpha n$ and a vector $\vecz\in\{-1,1\}^{\alpha kn}$ of the form $\vecz=(\vecz(1),\dots,\vecz(\alpha n))$ where $\vecz(i)\in\{-1,1\}^k$ for each $i\in[\alpha n]$. Alice's input $\vecx^*$ encodes a random bipartition of the vertex set according to the $\pm1$ pattern. Bob's $k$-uniform hypermatching is encoded by a matrix $M\in\{0,1\}^{\alpha kn\times n}$ where the $(k(i-1)+1)$-th to the $(ki)$-th rows encode the $i$-th hyperedge by putting exactly one $1$ in each row to the corresponding vertices. During the game, Alice sends a message to Bob and Bob has to discover the hidden structure of the vector $\vecz$. The following definition formally describes the problem. 

\begin{definition}[Randomized Mask Detection (RMD) Problem]\label{def:rmd}
For $k \in \mathbb{N}$, $\alpha \in (0,1/k]$ and a pair of distributions $\cD_Y,\cD_N \in \Delta(\{-1,1\}^k)$, the $(\cD_Y,\cD_N;\alpha,k)$-\textsf{RMD} problem is the $2$-player communication game given by a family of instances $(\cY_n,\cN_n)_{n \in \mathbb{N}, n \geq 1/\alpha}$ where for a given $n$, $\cY = \cY_n$ and $\cN = \cN_n$ are as follows: 
Both $\cY$ and $\cN$ are supported on triples $(\vecx^*,M,\vecz)$ where $\vecx^* \in \{-1,1\}^n$, $M \in \{0,1\}^{k\alpha n \times n} $ and $\vecz \in \{-1,1\}^{k\alpha n}$, where $\vecx^*$ is Alice's input and the pair $(M,\vecz)$ are Bob's inputs. We now specify the distributions of $\vecx^*,M$ and $\vecz$ in $\cY$ and $\cN$:
\begin{itemize}
\item In both $\cY$ and $\cN$, $\vecx^*$ is distributed uniformly over $\{-1,1\}^n$.
\item In both $\cY$ and $\cN$ the matrix $M\in\{0,1\}^{\alpha kn\times n}$ is chosen uniformly (and independently of $\vecx^*$) among matrices with exactly one $1$ per row and at most one $1$ per column. (Thus $M$ represents a $k$-hypermatching where each block of $k$ rows describes a hyperedge.)
\item The vector $\vecz$ is obtained by ``masking'' (i.e., xor-ing) $M\vecx^*$ by a random vector $\vecb\in \{-1,1\}^{\alpha kn}$ whose distribution differs in $\cY$ and $\cN$. Specifically let $\vecb=(\vecb(1),\dots,\vecb(\alpha n))$ be sampled from one of the following distributions (independent of $\vecx^*$ and $M$):
\begin{itemize}
    \item $\cY$:  Each $\vecb(i)\in\{-1,1\}^k$ is sampled independently according to $\cD_Y$.
    \item $\cN$:  Each $\vecb(i)\in\{-1,1\}^k$ is sampled independently according to $\cD_N$.
\end{itemize}
We now set $\vecz = (M \vecx^*)\odot \vecb$ (recall that that $\odot$ denotes coordinatewise product).
\end{itemize}
\end{definition}
We will typically suppress $k$ and $\alpha$ from the notation when they are clear from context and simply refer to the $(\cD_Y, \cD_N)$-\textsf{RMD}. We will refer to $n$ as the length parameter or refer to ``instances of length $n$" when the instances are drawn from $\cY_n$ vs. $\cN_n$. 
The goal of a protocol solving \textsf{RMD} is to distinguish between case where the masks are sampled from $\cD_Y$ from the case where the masks are sampled from $\cD_N$ and advantage measures this probability of distinguishing.

We note that our communication game is slightly different from those in previous works: Specifically the problem studied in \cite{GKKRW,KKS} is called the \textit{Boolean Hidden Matching (BHM)} problem from~\cite{GKKRW} and the works \cite{KKSV17,KK19} study a variant called the \textit{Implicit Hidden Partition} problem. While these problems are similar, they are less expressive than our formulation, and specifically do not seem to capture all $\maxf$ problems. 

There are two main differences between the previous settings and our setting. The first difference is the way to encode the matching matrix $M$. In all the previous works, each edge (or hyperedge) is encoded by a single row in $M$ where the corresponding columns are assigned to $1$, so that $m = \alpha n$. However, it turns out that this encoding hides too much information and hence we do not know how to reduce the problem to general \textsf{Max-CSP}. We unfold the encoding by using $k$ rows to encode a single $k$-hyperedge (leading to the setting of $m = k\alpha n$ in our case). The second difference is that we allow the masking vector $\vecb$ to be sampled from a more general distribution. This is also for the purpose of establishing a reduction to general \textsf{Max-CSP}.
That being said, it is possible to describe some of the previous results in our language: all the papers consider the complexity of distinguishing the distribution $\cD_Y = \textsf{Unif}(\{(1,1),(-1,-1)\})$ from the distribution $\cD_N = \textsf{Unif}(\{-1,1\}^2)$. This problem is shown to have a communication lower bound of $\Omega(\sqrt{n})$ in \cite{GKKRW}. And a variant of this problem (not captured by our formulation above) is shown to have an $\Omega(n)$ lower bound in \cite{KK19}.

Due to the above two differences, it is not clear how to derive communication lower bounds for general $\cD_Y$ and $\cD_N$ by reduction from the previous works.
The main technical contribution of this part of the paper is a communication lower bound for \textsf{RMD} for general $\cD_Y$ and $\cD_N$. We summarize the result in the following theorem.

\begin{theorem}[\textsf{RMD} Lower bound for distributions with matching marginals]\label{thm:communication lb matching moments}
For every  $k\in\N$, there exists $\alpha_0(k) > 0$ such that for every $\alpha\in(0,\alpha_0(k))$ and $\delta>0$ the following holds: For every pair of distributions $\cD_Y,\cD_N\in \Delta(\{-1,1\}^k)$ with $\vecmu(\cD_Y) = \vecmu(\cD_N)$
there exists $\tau > 0$ and $n_0$ such that for every $n \geq n_0$, every protocol for $(\cD_Y,\cD_n)$-\textsf{RMD} achieving advantage $\delta$ on instances of length $n$ requires $\tau\sqrt{n}$ bits of communication. 
\end{theorem}

We prove~\autoref{thm:communication lb matching moments} in two parts. First, in~\autoref{sec:BHBHM 1 wise}, we prove a communication lower bound for the special case where the marginals of $\cD_Y$ and $\cD_N$ are all zero. While this captures many new cases, it fails to capture the more interesting scenarios (involving non-approximation resistant problems). To get lower bounds for the general case, we reduce the $0$-marginal case to the general case in~\autoref{sec:BHBHM general}. 

%\mnote{Fix this para.}
In the rest of this section, we use~\autoref{thm:communication lb matching moments} to prove \autoref{thm:main-negative} and \autoref{thm:main-negative dynamic}. 
%We first perform a standard step on bootstrapping the number of hyperedges (which corresponds to the number of clauses in \textsf{Max-CSP}) in~\autoref{sec:BHBHM boosting}. Next, we present the reduction to \textsf{Max-CSP} in~\autoref{sec:BHBHM reduction}. Finally, we wrap up the proof for~\autoref{thm:main-negative} in~\autoref{sec:BHBHM wrap up}.

\subsection{The streaming lower bound}\label{sec:BHBHM boosting}

The hardness of \textsf{RMD} suggests a natural path for hardness of $\maxf$ problems in the streaming setting. Such a reduction would take two distributions $\cD_Y \in S_\gamma^Y$ and $\cD_N \in S_\beta^N$ with matching marginals, construct  distributions $\cY$ and $\cN$ of \textsf{RMD}, and then interpret these distributions (in a natural way) as distributions over instances of $\maxf$ that are indistinguishable to small space algorithms. While the exact details of this ``interpretation'' need to be spelled out, every step in this path can be achieved. Unfortunately this does not mean any hardness for $\maxf$ since the CSPs generated by this reduction would consist of instances that have at most one constraint per variable, and such instances are easy to solve! 

To go from the instance suggested by the  \textsf{RMD} problem to hard CSP instances, we instead pick $T$ samples (somewhat) independently from the distributions $\cY$ and $\cN$ suggested by the \textsf{RMD} problem and concatenate these. With an appropriate implementation of this notion (see \autoref{def:srmd}) it turns out it is possible to use the membership of the underlying distributions in $S^Y_\gamma$ and $S^N_\beta$ to argue that the resulting instances $\Psi$ do (almost always) have $\val_\Psi \geq \gamma$ or $\val_\Psi \leq \beta$. (We prove this after appropriate definitions in \autoref{lem:csp value}.) But now to one needs to connect the streaming problem generated from the $T$-fold sampled version to the RMD problem.

To this end we formalize the $T$-fold streaming problem, which we call $(\cD_Y,\cD_N,T)$\textsf{-streaming-RMD} problem, in \autoref{def:srmd}. Unfortunately, we are not able to reduce the $(\cD_Y, \cD_N)$\textsf{-RMD} problem to $(\cD_Y,\cD_N,T)$\textsf{-streaming-RMD} problem for all 
$\cD_Y$ and $\cD_N$. (Roughly this problem arises from the fact that the $T$ samples $(\vecx^{*}(t),M(t),\vecz(t))$ are not sampled independently from $\cY$ (or $\cN$) for $t \in [T]$. Instead they are sampled independently conditioned on $\vecx^*(1) = \cdots = \vecx^*(T)$. This hidden correlation in {\em both} the \yes\ and the \no\ cases turns out to be a serious problem.) But in the setting where $\cD_Y$ and $\cD_N$ have uniform marginals, we are able to effect the reduction and thus show that the streaming problem requires large space. This is a special case of \autoref{lem:reduce to streaming} and \autoref{cor:space lb 1 wise} which we discuss next.

We are able to extend our reduction from \textsf{RMD} to \textsf{streaming-RMD} slightly beyond the uniform marginal case, to the case where $\cD_Y$ and $\cD_N$ form a padded one-wise pair, but both the streaming problem and the analysis of the resulting CSP value need to be altered to deal with this case, as elaborated next. Let $\tau \in [0,1]$ and $\cD_0, \cD'_Y, \cD'_N$ be such that for $i \in \{Y,N\}$ we have $\cD_i = \tau \cD_0+ (1-\tau)\cD'_i$ and $\cD'_i$ has uniform marginals. Our padded streaming problem, denoted $(\cD'_Y,\cD'_N,T,\cD_0,\tau)$\textsf{-padded-streaming-RMD} problem, includes an appropriately large number of constraints generated according to $\cD_0$, followed by $T$ samples chosen according to the $(\cD'_Y,\cD'_N,T)$\textsf{-streaming-RMD} problem. See \autoref{def:srmd} for a formal definition. In \autoref{lem:csp value} we show that the CSP value of the resulting streaming problem inherits the properties of $\cD_Y$ and $\cD_N$ (which is not as immediate for \textsf{padded-streaming-RMD} as for \textsf{streaming-RMD}). We then show effectively that $(\cD'_Y,\cD'_N)$-\textsf{RMD} reduces to $(\cD'_Y,\cD'_N,T,\cD_0,\tau)$\textsf{-padded-streaming-RMD}. See \autoref{lem:reduce to streaming} and \autoref{cor:space lb 1 wise}. Putting these together leads to a proof of \autoref{thm:main-negative}. 

\newcommand{\strm}{\textrm{stream}}

\subsubsection{The (Padded) Streaming RMD Problem}

%\cnote{Potentially, we can directly define the padded-streaming-RMD to be the streaming-RMD. For now, I keep both definitions for flexibility but all the lemmas and theorems are proving on the padded version.}

\begin{definition}[$(\cD_Y,\cD_N,T)$-\textsf{streaming-RMD}]\label{def:srmd}
For $k,T\in\mathbb{N}$, $\alpha\in(0,1/k]$, distributions $\cD_Y,\cD_N$ over $\{-1,1\}^k$, the streaming problem $(\cD_Y,\cD_N,T; \alpha,k)$-\textsf{streaming-RMD}  is the task of distinguishing, for every $n$, $\vecsigma \sim \cY_{\strm,n}$ from $\vecsigma \sim \cN_{\strm,n}$ where for a given length parameter $n$, the distributions $\cY_\strm = \cY_{\strm,n}$ and $\cN_\strm=\cN_{\strm,n}$ are defined as follows:
\begin{itemize}
    \item Let $\cY$ be the distribution over instances of length $n$, i.e., triples $(\vecx^*,M,\vecz)$, from the definition of $(\cD_Y,\cD_N)$-\textsf{RMD}. For $\vecx \in \{-1,1\}^n$, let $\cY|_\vecx$ denote the distribution $\cY$ conditioned on $\vecx^* = \vecx$. The stream $\vecsigma\sim\cY_\strm$ is sampled as follows: Sample $\vecx^*$ uniformly from $\{-1,1\}^n$.  Let $(M^{(1)},\vecz^{(1)}),\ldots,(M^{(T)},\vecz^{(T)})$ be sampled independently according to $\cY|_{\vecx^*}$. Let $\vecsigma^{(t)}$ be the pair $(M^{(t)},\vecz^{(t)})$ presented as a stream of edges with labels in $\{-1,1\}^k$. Specifically for $t \in [T]$ and $i \in [\alpha n]$, let $\vecsigma^{(t)}(i) = (e^{(t)}(i),\vecz^{(t)}(i))$ where $e^{(t)}(i)$ is the $i$-th hyperedge of $M^{(t)}$, i.e., $e^{(t)}(i) = (j^{(t)}(k(i-1)+1),\ldots,j^{(t)}(k(i-1)+k)$ and $j^{(t)}(\ell)$ is the unique index $j$ such that $M^{(t)}_{j,\ell} = 1$. Finally we let  $\vecsigma = \vecsigma^{(1)} \circ \cdots \circ \vecsigma^{(T)}$ be the concatenation of the $\vecsigma^{(t)}$s. 
    \item $\vecsigma\sim\cN_\strm$ is sampled similarly except we now sample $(M^{(1)},\vecz^{(1)}),\ldots,(M^{(T)},\vecz^{(T)})$  independently according to $\cN|_{\vecx^*}$ where $\cN|_\vecx$ is the distribution $\cN$ condition on $\vecx^* = \vecx$.
\end{itemize}
\end{definition}
Again when $\alpha$ and $k$ are clear from context we suppress them and simply refer to the $(\cD_Y,\cD_N,T)$-\textsf{streaming-RMD} problem. 
\begin{remark}\label{rem:uniform-srmd}
We note that when $\cD_N = \textsf{Unif}(\{-1,1\}^k)$, then the distributions $\cN|_{\vecx^*}$ are identical for all $\vecx^*$ (and the variables
$\vecz^{(t)}(i)$ is distributed uniformly over $\{-1,1\}^k$ independently for every $t,i$).
\end{remark}

For technical reasons, we need the following \textit{padded} version of \textsf{streaming-RMD} to extend our lower bound techniques in the streaming setting beyond uniform marginals.
\begin{definition}[$(\cD_Y,\cD_N,T,\cD_0,\tau)$-\textsf{padded-streaming-RMD}]\label{def:psrmd}
For $k,T\in\mathbb{N}$, $\alpha\in(0,1/k]$, $\tau\in[0,1)$, distributions $\cD_Y,\cD_N,\cD_0$ over $\{-1,1\}^k$, the streaming problem $(\cD_Y,\cD_N,T,\cD_0,\tau; \alpha,k)$-\textsf{padded-streaming-RMD} is the task of distinguishing, for every $n$, $\vecsigma \sim \cY_{\pstrm,n}$ from $\vecsigma \sim \cN_{\pstrm,n}$ where for a given length parameter $n$, the distributions $\cY_\pstrm = \cY_{\pstrm,n}$ and $\cN_\pstrm=\cN_{\pstrm,n}$ are defined as follows: Sample $\vecx^*$ from $\{-1,1\}^n$ uniformly. For each $i\in[\frac{\tau}{1-\tau}\alpha nT]$, uniformly sample a tuple $e^{(0)}(i)=(i_1,\dots,i_k)\in\binom{[n]}{k}$ and $\vecb^{(0)}(i)\sim\cD_0$, let $\vecsigma^{(0)}(i)=(e^{(0)}(i),\vecx^*|_{e^{(0)}(i)}\odot\vecb^{(0)}(i))$. Next, sample $\vecsigma^{(1)},\dots,\vecsigma^{(T)}$ according to the Yes and No distribution of $(\cD_Y,\cD_N,T)$-\textsf{streaming-RMD} respectively. Finally, let $\vecsigma = \vecsigma^{(0)} \circ \cdots \circ \vecsigma^{(T)}$ be the concatenation of the $\vecsigma^{(t)}$s.
\end{definition}

Note that when $\tau=0$, $(\cD_Y,\cD_N,T,\cD_0,\tau)$-\textsf{padded-streaming-RMD} is the same as $(\cD_Y,\cD_N,T)$-\textsf{streaming-RMD}.

\subsubsection{CSP value of \psRMD}\label{sssec:csp value}
%\mnote{It seems that if we define $\cY_{\pstrm,n}$ and $\cN_{\pstrm,n}$ then the following lemma is well-defined. So perhaps define those two, and promote this lemma to between section 5.1 and section 5.2.} 
%\cnote{Done.}

%We now complete the sequence of reductions from \RMD\ to approximating $\maxf$ by reducing \psRMD\ to $\maxf$. To this end, note that 
There is a natural way to convert instances of \psRMD\ to a $\maxf$\ problem. In this section we make this conversion explicit and show how to use properties of the underlying distributions $\cD_0,\cD_Y,\cD_N$ to get bounds on the value of the instances produced.

Note that 
an instance 
$\vecsigma$ of \psRMD\ is simply a sequence $(\sigma(1),\ldots,\sigma(m))$ where each $\sigma(i) = (\vecj(i),\vecz(i))$ with $\vecj(i) \in [n]^k$ and $\vecz(i) \in \{-1,1\}^k$. This sequence is already syntactically very close to the description of a $\maxf$ instance. The only missing ingredient is any reference to the function $f$ itself! Indeed the reduction from \psRMD\ to $\maxf$ involves just applying this function $f$ to the literals indicated by $\sigma(i)$. 

Formally, given an instance $\vecsigma = (\sigma(1),\ldots,\sigma(m))$ of \psRMD, let $\Psi(\vecsigma)$ denote the instance of $\maxf$ on variables $\vecx = (x_1,\ldots,x_n)$ with the constraints $C_1,\ldots,C_m$ with $C_i = \sigma(i) = (\vecj(i),\vecz(i))$ is the constraint satisfied if $f(\vecz(i) \odot \vecx|_{\vecj(i)})=1$.

In what follows we show that if $\tau\cD_0+(1-\tau)\cD_Y  \in S_\gamma^Y$ then for all sufficiently large constant $T$, and sufficiently large $n$, if we draw $\vecsigma \sim \cY_{\pstrm,n}$, then with high probability, $\Psi(\vecsigma)$ has value at least $\gamma-o(1)$. Conversely if $\tau\cD_0+(1-\tau)\cD_N\in S_\beta^N$, then for all sufficiently large $n$, if we draw $\vecsigma \sim \cN_{\pstrm,n}$, then with high probability $\Psi(\vecsigma)$ has value at most $\beta+o(1)$. 

\begin{lemma}[CSP value of \psRMD]\label{lem:csp value}
For every $k\in\mathbb{N}$, $f:\{-1,1\}^k \to \{0,1\}$,  $0 \leq \beta < \gamma \leq 1$, $\epsilon > 0$, $\tau=[0,1)$, distributions $\cD_Y, \cD_N, \cD_0 \in \Delta(\{-1,1\}^k)$ there exists $\alpha_0$ such that for every $\alpha\in(0,\alpha_0]$, there exists an integer $T_0$ such that for every $T \geq T_0$
%~\cnote{Suggest: change to ``for every $T\geq1000/(\epsilon^2\alpha)$''}\mnote{I prefer (and applied) the opposite.}\cnote{Then let's stick to the original (current) one?}  
the following holds:
\begin{enumerate} 
\item If $\tau\cD_0+(1-\tau)\cD_Y \in S_\gamma^Y$, then for every sufficiently large $n$, the $(\cD_Y,\cD_N,T,\cD_0,\tau)$-\psRMD\ \yes\ instance $\vecsigma \sim \cY_{\pstrm,n}$ satisfies $\Pr[\val_{\Psi(\vecsigma)} < (\gamma-\epsilon)] \leq \exp(-n)$. 
\item If $\tau\cD_0+(1-\tau)\cD_N \in S_\beta^N$, then for every sufficiently large $n$, the $(\cD_Y,\cD_N,T,\cD_0,\tau)$-\psRMD\ \no\ instance $\vecsigma \sim \cN_{\pstrm,n}$ satisfies $\Pr[\val_{\Psi(\vecsigma)} > (\beta+\epsilon)] \leq \exp(-n)$. 
\end{enumerate} 
Furthermore, if $\gamma = 1$ then $\Pr_{\vecsigma\sim\cY_{\pstrm,n}}\left[\val_{\Psi(\vecsigma)}=1\right]=1$. 
\end{lemma}

\begin{proof}%[Proof of~\autoref{lem:csp value}]
We prove the lemma for $\alpha_0 = \epsilon/(100k^2)$ and $T_0 = 1000/(\epsilon^2\alpha)$.
Roughly our proof uses the fact that the definition of $S^Y_\gamma$ is setup so that $\Psi(\vecsigma)$ achieves value $\gamma$ under the ``planted'' assignment $\vecx^*$. 
Similarly $S^N_\beta$ is setup so that for every assignment, the expected value is not more than $\beta$. 

We recall that the condition $\tau\cD_0+(1-\tau)\cD_Y \in S_\gamma^Y$ implies that $\Exp_{\veca\sim\tau\cD_0+(1-\tau)\cD_Y}[f(\veca)] \geq \gamma$. Now consider a random \yes\ instance $\vecsigma \sim \cY_{\pstrm,n}$ of $(\cD_Y,\cD_N,T,\cD_0,\tau)$-\psRMD\ and let $\vecx^*$ denote the underlying vector corresponding to this draw. We show that for $\Psi = \Psi(\vecsigma)$ we have $\val_\Psi(\vecx^*) \geq \gamma-\epsilon$ with high probability.
We consider the constraints given by $\sigma(i)$ one at a time. Let $m = \frac{\tau}{1-\tau}\alpha n T+\alpha n T=\frac{\alpha n T}{1-\tau}$ denote the total number of constraints of $\Psi$. Let $Z_i = C_i(\vecx^*) = f(\vecz(i)\odot\vecx^*|_{\vecj(i)})$ denote the indicator of the event that the $i$th constraint is satisfied by $\vecx^*$. 
By construction of $\vecz(i)$ (from \autoref{def:rmd} and passed through \autoref{def:srmd}), we have $\vecz(i) = \vecb(i) \odot \vecx^*|_{\vecj(i)}$ where $\vecb(i) \sim \cD_Y$ independently of all other choices. We thus have
$Z_i = f(\vecb(i) \odot \vecx^*|_{\vecj(i)} \odot\vecx^*|_{\vecj(i)}) = f(\vecb(i))$. Thus $Z_i$ is a random variable,  chosen independently of $Z_{1},\ldots,Z_{i-1}$, with expectation $\Exp[Z_i | Z_1,\ldots,Z_{i-1}] = \Exp_{\vecb \sim \cD_0} [f(\vecb)]$ when $i\leq\frac{\tau}{1-\tau}\alpha n T$ and $\Exp[Z_i | Z_1,\ldots,Z_{i-1}] = \Exp_{\vecb \sim \cD_Y} [f(\vecb)]$ otherwise. In particular,
\begin{align*}
\Exp\left[\sum_{i=1}^mZ_i\right] &= \frac{\tau}{1-\tau}\alpha nT\cdot\Exp_{\vecb \sim \cD_0} [f(\vecb)] + \alpha n T\cdot \Exp_{\vecb \sim \cD_Y} [f(\vecb)]\\
&= \frac{\alpha n T}{1-\tau}\cdot \Exp_{\vecb \sim \tau\cD_0+(1-\tau)\cD_Y} [f(\vecb)] \\
&\geq\frac{\alpha nT}{1-\tau}\cdot\gamma ~ = ~ \gamma \cdot m\, .
\end{align*}

By applying a concentration bound (\autoref{lem:our-azuma} suffices, though even simpler Chernoff bounds would suffice) we get that $\Pr_{\vecsigma\sim\cY_{\strm,n}}[\val_{\Psi(\vecsigma)} = \frac1m{\sum_{i=1}^m Z_i} < (\gamma-\epsilon)] \leq \exp(-\epsilon^2 m) = \exp(-\epsilon^2 \alpha T n)$. This yields Part (1) of the lemma.

Note that, if $\gamma = 1$, then $Z_i = 1$ deterministically for every $i$, and so we get $\val_\Psi = 1$ with probability $1$, yielding the furthermore part of the lemma.

We now turn to the analysis of the \no\ case. By the condition $\tau\cD_0+(1-\tau)\cD_N \in S_\beta^N$ we have that for every $p \in [0,1]$, we have 
\begin{equation}\label{eqn:no-cond}
    \Exp_{\vecb\sim \tau\cD_0+(1-\tau)\cD_N}\Exp_{\veca \sim \bern(p)^k}[f(\vecb \odot \veca)]\leq \beta.
\end{equation}
Now consider any fixed assignment $\vecnu \in \{-1,1\}^n$. In what follows we show
that for a random
\no\ instance $\vecsigma \sim \cN_{\pstrm,n}$ of $(\cD_Y,\cD_N,T,\cD_0,\tau)$-\psRMD\ if we let $\Psi = \Psi(\vecsigma)$, then $\Pr[\val_{\Psi}(\vecnu) > (\beta + \epsilon)] \le c^{-n}$ for $c > 2$. This allows us to take a union bound over the $2^n$ possible $\vecnu$'s to claim $\Pr[\val_{\Psi} > (\beta + \epsilon)] \le 2^n\cdot c^{-n} = o(1)$. 

We thus turn to analyzing $\val_{\Psi}(\vecnu)$. Recall that $\vecsigma$ is chosen by picking $\vecx^* \in \{-1,1\}^n$ uniformly and then picking $\sigma(i)$'s based on this choice --- but our analysis will work for every choice of $\vecx^*\in \{-1,1\}^n$. Fix such a choice and let $\vecnu^* = \vecnu \odot \vecx^*$. Now  for $i \in [m]$ (where $m=\frac{\alpha n T}{1-\tau}$) let $Z_i$ denote the indicator of the event that $\vecnu$ satisfies $C_i$. We have $Z_i = f(\vecb(i) \odot \vecx^*_{\vecj(i)} \odot \vecnu_{\vecj(i)}) = f(\vecb(i) \odot \vecnu^*|_{\vecj(i)})$. Our goal is to prove that $\Pr[\sum_{i=1}^m Z_i > (\beta + \epsilon) \cdot m] \leq c^{-n}$. 
To this end,  let $\eta_i = \Exp[Z_i]$. Below we prove the following: (1) $\sum_{i=1}^m \eta_i \leq  (\beta+o(1))\cdot m$, and (2) For every $i$, and $Z_1,\ldots,Z_{i-1}$,
$\Exp[Z_i | Z_1,\ldots,Z_{i-1}] \leq \eta_i + \epsilon/2$. With (1) and (2) in hand, a straightforward application of Azuma's inequality yields that $\Pr[\sum_i Z_i > (\beta+\epsilon)\cdot m] \leq c_1^{-m}$ for some $c_1 > 1$. Picking $T$ large enough now ensures this is at most $c^{-n}$ for some $c > 2$.

We start by analyzing the $\eta_i$'s. Let $m_0 = \tau\cdot m$ and let $m_1 = \alpha n$ so that $m = m_0 + m_1 \cdot T$.
Let $p$ be the fraction of $1$'s in $\vecnu^*$.
When $i \leq m_0$,
we have
$\eta_i \le \Exp_{\vecb(i)\sim\cD_0} \Exp_{\veca \sim \bern(p)^k} [f(\vecb(i) \odot \veca)] + o(1)$, where the $o(1)$ term accounts for the difference between sampling $k$ elements from $n$ distinct elements with repetition and without. %\vnote{I think it is technically incorrect to have equality here. When we sample $\vecj(i)$, we sample each coordinate without repetition. Lets replace it with $\eta_i \le  \Exp_{\vecb(i)\sim\cD_0} \Exp_{\veca \sim \bern(p)^k} [f(\vecb(i) \odot \veca)] + o(1)?$ }
When $i > m_0$, we have 
$\eta_i \le \Exp_{\vecb(i)\sim\cD_N} \Exp_{\veca \sim \bern(p)^k} [f(\vecb(i) \odot \veca)] + o(1)$. 
By linearity of expectations, we now get 
\begin{eqnarray*} 
\sum_{i=1}^m \eta_i & \le & \tau m \cdot \Exp_{\vecb(i)\sim\cD_0} \Exp_{\veca \sim \bern(p)^k} [f(\vecb(i) \odot \veca)]
+ (1-\tau)m \cdot \Exp_{\vecb(i)\sim\cD_N} \Exp_{\veca \sim \bern(p)^k} [f(\vecb(i) \odot \veca)] +o(1)\cdot m \\
& \le & m \cdot \Exp_{\vecb(i)\sim\tau\cD_0 + (1-\tau)\cD_N} \Exp_{\veca \sim \bern(p)^k} [f(\vecb(i) \odot \veca)] +o(1)\cdot m
 \\
 & \leq & (\beta+o(1)) \cdot m \,.  \mbox{~~~ (By \autoref{eqn:no-cond})}
\end{eqnarray*}  
This yields (1).

Turning to part (2) we need to understand how the $Z_i$'s depend on each other. 
For the case $i \leq m_0$, note that $Z_1,\ldots,Z_{m_0}$ are independent by construction (\autoref{def:psrmd}).
So we have
$\Exp[Z_i\, |\, Z_1,\ldots,Z_{i-1}] = \eta_i$ in this case. 
We now consider $i>m_0$.  For $t \in [T]$, let us partition the variables $Z_{m_0+1},\ldots,Z_{m}$ into $T$ block $B_1,\ldots,B_T$ with $B_t = (Z_{m_0+(t-1)m_1 + 1},\ldots,Z_{m_0+t m_1})$ for $t \in [T$]. By construction (\autoref{def:srmd} via \autoref{def:psrmd}) we have that the blocks $B_1,\ldots,B_T$ are identically distributed and independent conditioned on $Z_1,\ldots,Z_{m_0}$. Thus the only dependence between the $Z_i$'s is within the $Z_i$'s in the same block. Within a block two $Z_i$'s may depend on each other due to the restriction that the underlying set of variables are disjoint. Thus, in particular when choosing the variables of $\vecsigma^{(t)}(i')$ (corresponding to $Z_{m_0+(t-1)m_1+i'}$, some subset $S \subseteq [n]$ of the variables may already be involved in constraints of the $t$-th block. Let $S$ be the set of variables not assigned to constraints in the $t$-th block at this time, and let $p_S$ denote the fraction of 1's in $\vecnu^*|_S$. (Note both $S$ and $p_S$ are random variables.) Since $|S| \geq n - k\alpha n$ and $\alpha \leq \epsilon/(100k^2)$ we have $|S|\geq (1 - \epsilon/(100k))n$ and so $|p - p_S| \leq \epsilon/(100k)$. In turn this implies that $\|\bern(p)^k - \bern(p_S)^k\|_{tvd} \leq \epsilon/100$. Using these bounds we now have:
\begin{align*}
\Exp[Z_i | Z_1,\ldots,Z_{i-1}] & = \Exp[Z_i | S]  \\
& = 
\Exp_{\vecj(i),\vecb(i)} [f(\vecb(i) \odot \vecnu^*|_{\vecj(i)})]\\
&\leq \Exp_{\vecb(i)\sim\cD_N} \Exp_{\veca \sim \bern(p_S)^k} [f(\vecb(i) \odot \veca)] + \frac{k^2}{|S|}\\
&\leq \Exp_{\vecb(i)\sim\cD_N} \Exp_{\veca \sim \bern(p)^k} [f(\vecb(i) \odot \veca)] + \frac{\epsilon}{100} + \frac{k^2}{|S|}\\
&\leq \Exp_{\vecb(i)\sim\cD_N} \Exp_{\veca \sim \bern(p)^k} [f(\vecb(i) \odot \veca)] + \frac{\epsilon}{2} \\
&= \eta_i + \frac{\epsilon}{2}\,. 
\end{align*}
(In the second equality above $\vecj(i)$ denotes a random sequence of distinct variables from $S$. The next inequality comes from the difference between sampling $k$ elements from $S$ with repetition and without. The following inequality is the key one, using $\|\bern(p)^k - \bern(p_S)^k\|_{tvd} \leq \epsilon/100$. The final inequality uses the fact that $n$ and hence $|S|$ are large enough, and the final equality uses the value of $\eta_i$ from Part (1).)
This concludes Part (2). 

Finally we use a version of a concentration bound for submartingales to combine (1) and (2) to get the desired bound on $\Pr[\val_{\Psi}(\vecnu) > (\beta + \epsilon)] \le c^{-n}$. Specifically, we apply \autoref{lem:our-azuma} with $N=m$, $p_i = \eta_i$ for $i \in [N]$ and $\Delta = (\epsilon/2)N$ to conclude that $\Pr[\val_{\Psi}(\vecnu) > (\beta + \epsilon)] = \Pr[\sum_i Z_i > (\beta + \epsilon)\cdot m] \leq \exp(-O(\epsilon^2 \alpha n T))$. Given $c$ we can choose $T$ to be large enough so that this is at most $c^{-n}$.
This concludes the analysis of the \no\ case, and thus the lemma.
\end{proof}

\subsubsection{Reduction from one-way \texorpdfstring{$(\cD_Y,\cD_N)$-\RMD}{(DY,DN)-RMD)}\ to \texorpdfstring{$\psRMD$}{padded-streaming-RMD}}
%\subsubsection{Reduction from $\psRMD$ to one-way $(\cD_Y,\cD_N)$-\RMD}

We start by reducing \textsf{RMD} to \textsf{padded-streaming-RMD} in the special case where $\cD_N = \textsf{Unif}(\{-1,1\}^k)$. Note that since the former is hard in this case for all $\cD_Y$ with uniform marginals, applying this argument twice shows hardness of \textsf{padded-streaming-RMD} for all $\cD_Y$ and $\cD_N$ with uniform marginals.

\begin{lemma}\label{lem:reduce to streaming}
Let $T,k\in\mathbb{N}$, $\alpha \in (0,\alpha_0(k)]$, $\tau\in[0,1)$, and $\cD_Y,\cD_N,\cD_0\in\Delta(\{-1,1\}^k)$ with $\vecmu(\cD_Y)=0^k$ and $\cD_N = \textsf{Unif}(\{-1,1\}^k)$. Suppose there is a streaming algorithm $\ALG$ solves $(\cD_Y,\cD_N,T,\cD_0,\tau)$-\textsf{padded-streaming-RMD} on instances of length $n$ with advantage $\Delta$ and space $s$, then there is a one-way protocol for $(\cD_Y,\cD_N)$-\textsf{RMD} on instances of length $n$ using at most $sT$ bits of communication achieving advantage at least $\Delta/T$.
\end{lemma}

%\mnote{Update this}
The proof of~\autoref{lem:reduce to streaming} is based on a hybrid argument (\textit{e.g.,}~\cite[Lemma 6.3]{KKS}). We provide a proof here based on the proof of \cite[Lemma 4.11]{CGV20}.
% (We note that previous lemmas of this form only considered the case where $\cD_N$ is the uniform distribution, and the proofs used some special properties of this setting. Generalizing it to arbitrary $\cD_N$ involves a little extra care as we do below.)
%Later, in~\autoref{sec:BHBHM reduction}, we show a reduction from \textsf{streaming-RMD} to $\maxf$ thus completing the objective of this section.

\begin{proof}[Proof of~\autoref{lem:reduce to streaming}]
Note that since we are interested in distributional advantage, we can fix the randomness in $\ALG$ so that it becomes a deterministic algorithm. By an averaging argument the randomness can be chosen to ensure the advantage does not decrease. Let $\Gamma$ denote the evolution of function of $\ALG$ as it processes a block of $\alpha n$ edges. That is, if the algorithm is in state $s$ and receives a stream $\vecsigma$ of length $\alpha n$ then it ends in state $\Gamma(s,\vecsigma)$. Let $s_0$ denote its initial state. \\

%Also, without loss of generality we consider $\cD_N$ be the uniform distribution over $\{-1,1\}^k$. Note that in this case $\vecmu(\cD_N)=0^k$ and hence by a triangle inequality argument on the advantage of \textsf{RMD} would give the desired conclusion.

We consider the following collection of (jointly distributed) random variables: Let $\vecx^* \sim \textsf{Unif}(\{-1,1\}^n)$. Denote $\cY=\cY_{\pstrm,n}$ and $\cN=\cN_{\pstrm,n}$. Let $(\vecsigma_Y^{(0)},\vecsigma_Y^{(1)},\ldots,\vecsigma_Y^{(T)}) \sim \cY|_{\vecx^*}$. Similarly, let $(\vecsigma_N^{(0)},\vecsigma_N^{(1)},\ldots,\vecsigma_N^{(T)}) \sim \cN|_{\vecx^*}$. Recall by \autoref{rem:uniform-srmd} that since $\cD_N$ is the uniform distribution, we have $\cN|_{\vecx^*}$ is independent of $\vecx^*$, a feature that will be crucial to this proof. 

Let $S_t^Y$ denote the state of 
$\ALG$ after processing $\vecsigma_Y^{(0)},\ldots,\vecsigma_Y^{(t)}$, i.e., $S_0^Y = \Gamma(s_0,\vecsigma_Y^{(0)})$ and $S_t^Y = \Gamma(S_{t-1}^Y,\vecsigma_Y^{(t)})$ where $s_0$ is the fixed initial state (recall that $\ALG$ is deterministic). 
Similarly let $S_t^N$ denote the state of $\ALG$ after processing $\vecsigma_N^{(0)},\ldots,\vecsigma_N^{(t)}$. Note that since $\vecsigma_Y^{(0)}$ has the same distribution (conditioned on the same $\vecx^*$) as $\vecsigma_N^{(0)}$ by definition, we have $\|S_{0}^Y-S_{0}^N\|_{tvd}=0$.

Let $S^Y_{a:b}$ denote the sequence of states $(S_a^Y,\ldots,S_b^Y)$ and similarly for $S^N_{a:b}$. Now let $\Delta_t = \|S_{0:t}^Y-S_{0:t}^N\|_{tvd}$. 
Observe that $\Delta_0=0$ while $\Delta_T\geq\Delta$. (The latter is based on the fact that $\ALG$ distinguishes the two distributions with advantage $\Delta$.)
Thus $\Delta \leq \Delta_T - \Delta_0 = \sum_{t=0}^{T-1} (\Delta_{t+1} - \Delta_t)$ and so there exists $t^*\in\{0,1,\dots,T-1\}$ such that
\[
\Delta_{t^*+1} - \Delta_{t^*} = \|S_{0:t^*+1}^Y-S_{0:t^*+1}^N\|_{tvd}-\|S_{0:t^*}^Y-S_{0:t^*}^N\|_{tvd}\geq\frac{\Delta}{T} \, .
\]

Now consider the random variable $\tilde{S} = \Gamma(S_{t^*}^Y,\vecsigma_N^{(t^*+1)})$ (so the previous state is from the \yes\ distribution and the input is from the \no\ distribution). We claim below that $\| S_{t^*+1}^Y - \tilde{S} \|_{tvd} = \Exp_{A \sim_d S_{0:t^*}^Y}[\| S_{t^*+1}^Y|_{S_{0:t^*}^Y=A} - \tilde{S}|_{S_{0:t^*}^Y=A} \|_{tvd} ] \geq \Delta_{t^*+1} - \Delta_{t^*}$.
Once we have the claim, we show how to get a space $T\cdot s$ protocol for $(\cD_Y,\cD_n)$-\textsf{RMD} with advantage $\Delta_{t^*+1} - \Delta_{t^*}$ concluding the proof of the lemma.

\begin{claim}
$\| S_{t^*+1}^Y - \tilde{S} \|_{tvd} \geq \Delta_{t^*+1} - \Delta_{t^*}$.
% $\Exp_{A \sim_d S_{0:t^*}^Y}[\| S_{t^*+1}^Y|_{S_{0:t^*}^Y=A} - \tilde{S}|_{S_{0:t^*}^Y=A} \|_{tvd} ] \geq \Delta_{t^*+1} - \Delta_{t^*}$.
\end{claim}

\newcommand{\cpl}{\textrm{couple}} 

\begin{proof}
First, by triangle inequality for the total variation distance, we have
\[
\|S_{t^*+1}^Y - \tilde{S}\|_{tvd} \geq \|S_{t^*+1}^Y - S_{t^*+1}^N\|_{tvd} - \|\tilde{S} - S_{t^*+1}^N\|_{tvd} \, .
\]
Recall that $\tilde{S} = \Gamma(S_{t^*}^Y,\vecsigma_N^{(t^*+1)})$ and $S^N_{t^*+1} = \Gamma(S_{t^*}^N,\vecsigma_N^{(t^*+1)})$. Also, note that $\vecsigma_N^{(t^*+1)}$ is uniformly distributed over $(\{-1,1\}^k)^{\alpha n}$ and in particular is independent of $S_{t^*}^Y$ and $S_{t^*}^N$. (This is where we rely crucially on the property $\cD_N = \textsf{Unif}(\{-1,1\}^k)$.)
Furthermore $\Gamma$ is a deterministic function, and so we can apply the data processing inequality (Item (2) of~\autoref{prop:tvd properties} with $X=S^Y_{t^*}$, $Y=S^N_{t^*}$, $W=\vecsigma_N^{(t^*+1)}$, and $f=\Gamma$) to conclude 
\[\|\tilde{S} - S_{t^*+1}^N\|_{tvd} = \|\Gamma(S_{t^*}^Y,\vecsigma_N^{(t^*+1)}) -  \Gamma(S_{t^*}^N,\vecsigma_N^{(t^*+1)})\|_{tvd}\leq\|S_{t^*}^Y-S_{t^*}^N\|_{tvd}.\]
Combining the two inequalities above we get 
%By the data processing inequality, we have $\|\tilde{S} - S_{t^*+1}^N\|_{tvd}\leq\|S_{t^*}^Y-S_{t^*}^N\|_{tvd}$. Concretely, we apply item 2 of~\autoref{prop:tvd properties} with $X=S^Y_{t^*}$, $Y=S^N_{t^*}$, $W=\vecsigma_N^{(t^*+1)}$, and $f=\Gamma$. Note that because $\cD_N$ is the uniform distribution over $\{-1,1\}^k$, $W=\vecsigma_N^{(t^*+1)}$ is independent to $\vecx^*$ and hence is also independent to both $X$ and $Y$. Thus, the application of data processing inequality is legitimate. 
%\mnote{This is where we had an error in the last paper. So we should be extra careful in explaining this part. Why do we use a DPI from CGV rather than the one from Section 3? What makes an application of DPI (il)legitimate. Have we properly accounted for the fact that $\sigma_Y^t$ may be correlated with the hidden $x^*$? We should stress this (and the non-correlation of $\sigma^t_N$) more.}
%\cnote{The DPI from CGV is the same as the one from Section 3 and now I referenced to our Prop 3.4. The subtle part is $W$ has to be independent to both $X$ and $Y$ and this is the place we failed previously. Is the current explanation clear enough?}
%\mnote{Thanks. I think this gives me enough to see why the claim is true, and then to exposit on it.}
\[
\|S_{t^*+1}^Y - \tilde{S}\|_{tvd} \geq \|S_{t^*+1}^Y - S_{t^*+1}^N\|_{tvd} - \|S_{t^*}^Y - S_{t^*}^N\|_{tvd}=\Delta_{t^*+1}-\Delta_{t^*}
\]
as desired.

\end{proof}

We now show how a protocol can be designed for $(\cD_Y,\cD_N)$-\textsf{RMD} that achieves advantage at least $\theta = \Exp_{A \sim_d S_{0:t^*}^Y}[\| S_{t^*+1}^Y|_{S_{0:t^*}=A} - \tilde{S}|_{S_{0:t^*}=A} \|_{tvd} ] \geq \Delta_{t^*+1} - \Delta_{t^*}$ concluding the proof of the lemma. The protocol uses the distinguisher $T_A:\{0,1\}^{s} \to \{0,1\}$ such that $\Exp_{A,S_{t^*+1}^Y,\tilde{S}}[T_A(S_{t^*+1}^Y)] - \Exp[T_A(\tilde{S})] \geq \theta$ which is guaranteed to exist by the definition of total variation distance.

Our protocol works as follows: Let Alice receive input $\vecx^*$ and Bob receive inputs $(M,\vecz)$ sampled from either $\cY_{\RMD}|_{\vecx^*}$ or $\cN_{\RMD}|_{\vecx^*}$ where $\cY_{\RMD}$ and $\cN_{\RMD}$ are the Yes and No distribution of $(\cD_Y,\cD_N)$-$\RMD$ respectively.
\begin{enumerate}
\item Alice samples $(\vecsigma^{(0)},\vecsigma^{(1)},\ldots,\vecsigma^{(T)}) \sim \cY|_{\vecx^*}$ and computes $A = S_{0:t^*}^Y \in \{0,1\}^{(t^*+1)s}$ and sends $A$ to Bob.
\item Bob extracts $S_{t^*}^Y$ from $A$, computes $\widehat{S} = \Gamma(S_{t^*}^Y,\vecsigma)$, where $\vecsigma$ is the encoding of $(M,\vecz)$ as a stream, and outputs \yes\ if $T_A(\widehat{S})=1$ and \no\ otherwise.
\end{enumerate}
Note that if $(M,\vecz)\sim \cY_{\RMD}|_{\vecx^*}$ then $\widehat{S} \sim_d S_{t^*+1}^Y|_{S^Y_{0:t^*} = A}$ while if $(M,\vecz)\sim \cN_{\RMD}|_{\vecx^*}$ then $\widehat{S} \sim \tilde{S}_{S^Y_{0:t^*} = A}$. 
It follows that the advantage of the protocol above exactly equals $\Exp_A[T_A(S_{t^+1}^Y)] - \Exp_A[T_A(\tilde{S})] \geq \theta \geq \Delta_{t^*+1}-\Delta_{t^*} \geq \Delta/T$.
This concludes the proof of the lemma.
\end{proof}

By combining~\autoref{lem:reduce to streaming} with~\autoref{thm:communication lb matching moments}, we immediately have the following consequence.

\begin{lemma}\label{cor:space lb 1 wise}
For $k \in \N$ let  $\alpha_0(k)$ be as given by \autoref{thm:communication lb matching moments}.
Let $T\in\mathbb{N}$, $\alpha\in(0,\alpha_0(k)], \tau\in[0,1)$, and $\cD_0,\cD_Y,\cD_N$ be three distributions over $\{-1,1\}^k$ with $\vecmu(\cD_Y) = \vecmu(\cD_N) = 0^k$. Then every streaming algorithm $\ALG$ solving $(\cD_Y,\cD_N,T,\cD_0,\tau)$-\textsf{padded-streaming-RMD} %in the insertion-only setting 
with advantage $1/8$ for all lengths uses space $\Omega(\sqrt{n})$.
\end{lemma}
\begin{proof}
Let $\ALG$ be an algorithm using space $s$ solving  $(\cD_Y,\cD_N,T)$-\textsf{padded-streaming-RMD} with advantage $1/8$.
Let $\cD_M = \textsf{Unif}(\{-1,1\}^k)$. Then by the triangle inequality $\ALG$ solves either the $(\cD_Y,\cD_M,T,\cD_0,\tau)$-\textsf{padded-streaming-RMD} with advantage $1/16$ or it solves the 
$(\cD_N,\cD_M,T,\cD_0,\tau)$-\textsf{padded-streaming-RMD} with advantage $1/16$. Assume without loss of generality it is the former. 
Then by \autoref{lem:reduce to streaming}, there exists a  one-way protocol for $(\cD_Y,\cD_M)$-\textsf{RMD} using at most $sT$ bits of communication with advantage at least $1/(16T)$. Applying \autoref{thm:communication lb matching moments} with $\delta = 1/(16T) > 0$, we now get that $s = \Omega(\sqrt{n})$. 

%We get the lemma by combining \autoref{lem:reduce to streaming} and \autoref{thm:communication lb matching moments}.  Then 

\end{proof}

% \subsection{Reduction from \texorpdfstring{$(\cD_Y,\cD_N,T)$}{(DY,DN,T)}-\textsf{simultaneous-RMD} to dynamically approximating \texorpdfstring{\textsf{Max-CSP($f$)}}{Max-CSP(f)}}\label{sec:reduction T-RMD to CSP}

\subsubsection{Proof of the streaming lower bound}\label{sec:lb insert}

We are now ready to prove \autoref{thm:main-negative}.

\begin{proof}[Proof of \autoref{thm:main-negative}]
%\cnote{Please double check the way of setting $\alpha,T$ depending on $\epsilon$.}
%\mnote{Made adjustments.}
%\cnote{The adjustments look good.}
We combine \autoref{thm:communication lb matching moments}, \autoref{cor:space lb 1 wise} and \autoref{lem:csp value}. So in particular we set our parameters $\alpha$ and $T$ so that the conditions of these statements are satisfied. Specifically  
$k$ and $\epsilon>0$, let $\alpha^{(1)}_0$ be the constant from \autoref{thm:communication lb matching moments}
and let $\alpha^{(2)}_0$ be the constant from \autoref{lem:csp value}. Let $\alpha_0 = \min\{\alpha_0^{(1)},\alpha_0^{(2)}\}$, 
Given $\alpha \in (0,\alpha_0)$ let $T_0$ be the constant from \autoref{lem:csp value} and let $T = T_0$. (Note that these choices allow for both \autoref{thm:communication lb matching moments} and \autoref{lem:csp value} to hold.)
%Let $T= \lceil 10000/(\epsilon^2\alpha)\rceil $. \mnote{Why?}
Suppose there exists a streaming algorithm $\ALG$ that solves $(\gamma-\epsilon,\beta+\epsilon)$-\textsf{Max-CSP}($f$). 
Let $\tau\in[0,1)$ and $\cD_Y, \cD_N,\cD_0$ be distributions such that (i) $\vecmu(\cD_Y)=\vecmu(\cD_N)=0^k$, (ii) $\tau\cD_0+(1-\tau)\cD_Y\in S_\gamma^Y(f)$, (iii) $\tau\cD_0+(1-\tau)\cD_N\in S_\beta^N(f)$, and (iv) $\vecmu=\tau\vecmu(\cD_0)$.
Let $n$ be sufficiently large and let $\cY_{\strm,n}$ and $\cN_{\strm,n}$ denote the distributions of \yes\ and \no\ instances of $(\cD_Y,\cD_N,T,\cD_0,\tau)$-\psRMD\ of length $n$. 
Since $\alpha$ and $T$ satisfy the conditions of \autoref{lem:csp value}, we have for every sufficiently large $n$
\[
\Pr_{\vecsigma\sim\cY_{\strm,n}}\left[\val_{\Psi(\vecsigma)}<\left(\gamma-\epsilon\right)\right]=o(1)
\text{~~~and~~~}
\Pr_{\vecsigma\sim\cN_{\strm,n}}\left[\val_{\Psi(\vecsigma)}>\left(\beta+\epsilon\right)\right]=o(1) \, .
\]

We conclude that $\ALG$ can distinguish \yes\ instances of \textsf{Max-CSP}($f$)  from \no\ instances  with advantage at least $1/4-o(1)\geq 1/8$. 
However, since $\cD_Y,\cD_N$ and $\alpha$ satisfy the conditions of \autoref{cor:space lb 1 wise} (in particular $\vecmu(\cD_Y) = \vecmu(\cD_N) = 0^k$ and $\alpha \in (0,\alpha_0(k))$) such an algorithm requires space at least $\Omega(\sqrt{n})$. Thus, we conclude that any streaming algorithm that solves $(\gamma-\epsilon,\beta+\epsilon)$-\textsf{Max-CSP}($f$) requires $\Omega(\sqrt{n})$ space.

Finally, note that if $\gamma = 1$ then in \autoref{lem:csp value}, we have $\val_\Psi=1$ with probability one. Repeating the above reasoning with this information, shows that $(1,\beta+\epsilon)-\maxf$ requires $\Omega(\sqrt{n})$-space.

\end{proof}

\subsection{The lower bound against sketching algorithms}\label{sec:lb dynamic} 

In the absence of a reduction from \textsf{RMD} to \textsf{streaming-RMD} for general $\cD_Y$ and $\cD_N$, we turn to other means of using the hardness of \textsf{RMD}. In particular, we use lower bounds on the communication complexity of a $T$-player communication game in the {\em simultaneous communication} setting --- one which is significantly easier to obtain lower bounds for than the one-way setting. %Here the result of 
%Ai, Hu, Li and Woodruff~\cite{ai2016new} (building up on Li, Nguyen and Woodruff~\cite{li2014turnstile}) comes to our aid.  \cite{ai2016new} use the special properties of the dynamic setting to set up a more efficient reduction.  In particular, they manage to show lower bounds on the space complexity of the dynamic streaming problems by using lower bounds on the communication complexity of a $T$-player communication game in the {\em simultaneous communication} setting --- one which is significantly easier to obtain lower bounds for than the one-way setting.
%Below we explain their setup and show how our problems fit in their general class of problems. %\mnote{concrete pointer?} 
Below we describe a family of  $T$-player simultaneous communication games, which we call $(\cD_Y,\cD_N,T)$-\textsf{simultaneous-RMD} (See \autoref{def:simrmd}.) We then show a simple reduction from $(\cD_Y,\cD_N)$-\textsf{RMD} to $(\cD_Y,\cD_N,T)$-\textsf{simultaneous-RMD}. Combining this reduction with our lower bounds on \textsf{RMD} and the reduction from $(\cD_Y,\cD_N,T)$-\textsf{simultaneous-RMD} to streaming complexity leads to the proof of \autoref{thm:main-negative dynamic}.

\subsubsection{\texorpdfstring{$T$}{T}-Player Simultaneous Version of RMD}\label{sec:T-RMD}

\newcommand{\simul}{\textrm{simul}}

In this section, we consider the complexity of \textit{$T$-player number-in-hand simultaneous message passing communication games} (abbrev. $T$-player simultaneous communication games). Such games are described by two distributions $\cY$ and $\cN$. An instance of the game is a $T$-tuple $(X^{(1)},\dots,X^{(T)})$ either drawn from $\cY$ or from $\cN$ and $X^{(t)}$ is given as input to the $t$-th player. A (simultaneous communication) protocol $\Pi = (\Pi^{(1)},\dots,\Pi^{(T)},\Pi_{\text{ref}})$ is a $(T+1)$-tuple of functions with $\Pi^{(t)}(X^{(t)}) \in \{0,1\}^c$ denoting the $t$-th player's message to the \textit{referee}, and $\Pi_{\text{ref}}(\Pi^{(1)}(X^{(1)}),\dots,\Pi^{(T)}(X^{(T)}))\in \{\yes,\no\}$ denoting the protocol's output. We denote this output by $\Pi(X^{(1)},\dots,X^{(T)})$. The complexity of this protocol is the parameter $c$ specifying the maximum length of $\Pi^{(1)}(X^{(1)}),\dots,\Pi^{(T)}(X^{(T)})$ (maximized over all $X$). The advantage of the protocol $\Pi$ is the quantity 
$$\left| \Pr_{(X^{(1)},\dots,X^{(T)})\sim\cY} [ \Pi(X^{(1)},\dots,X^{(T)}) = \yes] - \Pr_{(X^{(1)},\dots,X^{(T)})\sim \cN} [\Pi(X^{(1)},\dots,X^{(T)}) = \yes] \right|. $$

\begin{definition}[$(\cD_Y,\cD_N,T)$-\textsf{simultaneous-RMD}]\label{def:simrmd}
For $k,T\in\N$, $\alpha\in(0,1/k]$, distributions $\cD_Y,\cD_N$ over $\{-1,1\}^k$, the $(\cD_Y,\cD_N,T)$-\textsf{simultaneous-RMD} is a $T$-player communication game given by a family of instances $(\cY_{\simul,n},\cN_{\simul,n})_{n\in\N,n\geq1/\alpha}$ where for a given $n$, $\cY=\cY_{\simul,n}$ and $\cN=\cN_{\simul,n}$ are as follows: Both $\cY$ and $\cN$ are supported on tuples $(\vecx^*,M^{(1)},\dots,M^{(T)},\vecz^{(1)},\dots,\vecz^{(T)})$ where $\vecx^*\in\{-1,1\}^n$, $M^{(t)}\in\{0,1\}^{k\alpha n\times n}$, and $\vecz^{(t)}\in\{-1,1\}^{k\alpha n}$, where the pair $(M^{(t)},\vecz^{(t)})$ are the $t$-th player's inputs for all $t\in[T]$. We now specify the distributions of $\vecx^*$, $M^{(t)}$, and $\vecz^{(t)}$ in $\cY$ and $\cN$:
\begin{itemize}
\item In both $\cY$ and $\cN$, $\vecx^*$ is distributed uniformly over $\{-1,1\}^n$.
\item In both $\cY$ and $\cN$, the matrix $M^{(t)}\in\{0,1\}^{\alpha kn\times n}$ is chosen uniformly (and independently of $\vecx^*$) among matrices with exactly one $1$ per row and at most one $1$ per column.
\item The vector $\vecz^{(t)}$ is obtained by ``masking'' (i.e., xor-ing) $M^{(t)}\vecx^*$ by a random vector $\vecb^{(t)}\in \{-1,1\}^{\alpha kn}$ whose distribution differs in $\cY$ and $\cN$. Specifically, let $\vecb^{(t)}=(\vecb^{(t)}(1),\dots,\vecb^{(t)}(\alpha n))$ be sampled from one of the following distributions (independent of $\vecx^*$ and $M$):
\begin{itemize}
    \item $\cY$:  Each $\vecb^{(t)}(i)\in\{-1,1\}^k$ is sampled independently according to $\cD_Y$.
    \item $\cN$:  Each $\vecb^{(t)}(i)\in\{-1,1\}^k$ is sampled independently according to $\cD_N$.
\end{itemize}
We now set $\vecz^{(t)} = (M^{(t)} \vecx^*)\odot \vecb^{(t)}$ (recall that that $\odot$ denotes coordinatewise product).
\end{itemize}
Given an instance $\vecsigma = (\vecx^*,M^{(1)},\dots,M^{(T)},\vecz^{(1)},\dots,\vecz^{(T)})$ and a function $f:\{-1,1\}^k\to\{0,1\}$,  we will let $\Psi(\vecsigma)$ represent the instance of $\maxf$ it corresponds to, presented as a stream of $T\alpha n$ constraints. 
\end{definition}

Note that the instance $\Psi(\vecsigma)$ obtained in the \yes\ and \no\ cases of $(\cD_Y,\cD_N,T)$-\textsf{simultaneous-RMD} are distributed exactly according to instances derived in the \yes\ and \no\ cases of $(\cD_Y,\cD_N,T,\cD_0,\tau=0)$-\textsf{padded-streaming-RMD} and thus \autoref{lem:csp value} can still be applied to conclude that \yes\ instances usually have $\val_{\Psi(\vecsigma)} \geq \gamma - o(1)$ and \no\ instances usually have $\val_{\Psi(\vecsigma)} \leq \beta - o(1)$. We will use this property when proving \autoref{thm:main-negative dynamic}.

{%\color{blue} %\mnote{Trying alternate plan}

We start by showing the \textsf{simultaneous-RMD} problems above do not have low-communication protocols when the marginals of $\cD_Y$ and $\cD_N$ match.

\begin{lemma}\label{lem:reduce RMD to simul-RMD}
Let $k,T\in\N$, $\cD_Y,\cD_N\in \Delta(\{-1,1\}^k)$, and let $\alpha\in(0,1/k]$.
Suppose there is a protocol $\Pi$ that solves $(\cD_Y,\cD_N,T)$-\textsf{simultaneous-RMD} on instances of length $n$ with advantage $\Delta$ and space $s$, then there is a one-way protocol for $(\cD_Y,\cD_N)$-\textsf{RMD} on instances of length $n$ using at most $s(T-1)$ bits of communication achieving advantage at least $\Delta/T$.
\end{lemma}
\begin{proof}
Let us first fix the randomness in $\Pi$ so that it becomes a deterministic protocol. Note that by an averaging argument the advantage of $\Pi$ does not decrease. Recall that $\cY$ and $\cN$ are Yes and No input distribution of $(\cD_Y,\cD_N,T)$-\textsf{simultaneous-RMD} and we have
\[
\Pr_{X\sim\cY}[\Pi(X)=\yes]-\Pr_{X\sim\cN}[\Pi(X)=\yes]\geq\Delta \, .
\]
Now, we define the following distributions $\cD_0,\dots,\cD_{T}$. Let $\cD_0=\cY$ and $\cD_{T}=\cN$. For each $t\in[T-1]$, we define $\cD_t$ to be the distribution of input instances of $(\cD_Y,\cD_N,T)$-\textsf{simultaneous-RMD} by sampling $\vecb^{(t')}(i)$ independently according to $\cD_Y$ (resp. $\cD_N$) for all $t'\leq t$ (resp. $t'>t$) and $i$ (see~\autoref{def:simrmd} to recall the definition).  Next, for each $t\in[T]$, let
\[
\Delta_t=\Pr_{X\sim\cD_t}[\Pi(X)=\yes]-\Pr_{X\sim\cD_{t-1}}[\Pi(X)=\yes] \, .
\]
Observe that $\sum_{t\in[T]}\Delta_t=\Delta$ and hence there exists $t^*\in[T]$ such that $\Delta_{t^*}\geq\Delta/T$.

Now, we describe a protocol $\Pi'$ for $(\cD_Y,\cD_N)$-\textsf{RMD} as follows. On input $(\vecx^*,M,\vecz)$, Alice receives $\vecx^*$ and Bob receives $(M,\vecz)$. Alice first samples matrices $M^{(1)},\dots,M^{(t^*-1)},M^{(t^*+1)},\dots,M^{(T)}$ as the second item in~\autoref{def:simrmd}. Next, Alice samples $\vecb^{(t')}$ according to $\cD_Y$ (resp. $\cD_N$) for all $t'<t^*$ (resp. $t'>t^*$) and sets $\vecz^{(t')}=(M^{(t')}\vecx^*)\odot\vecb^{(t')}$ as the third item in~\autoref{def:simrmd}. Note that this is doable for Alice because she possesses $\vecx^*$. Finally, Alice sends $\{\Pi^{(t')}(M^{(t')},\vecz^{(t')})\}_{t'\in[T]\backslash\{t^*\}}$ to Bob. After receiving Alice's message $(X^{(1)},\dots,X^{(t^*-1)},X^{(t^*+1)},\dots,X^{(T)})$, Bob computes $\Pi^{(t^*)}(M,\vecz)$ and outputs $\Pi'(M,\vecz)=\Pi_{\text{ref}}(X^{(1)},\dots,X^{(t^*-1)},\Pi^{(t^*)}(M,\vecz),X^{(t^*+1)},\dots,X^{(T)})$.

It is clear from the construction that the protocol $\Pi'$ uses at most $s(T-1)$ bits of communication. To see $\Pi'$ has advantage at least $\Delta/T$, note that if $(\vecx^*,M,\vecz)$ is sampled from the Yes distribution $\cY_{\textsf{RMD}}$ of $(\cD_Y,\cD_N)$-\textsf{RMD}, then $((M^{(1)},\vecz^{(1)}),\dots,(M^{(t^*-1)},\vecz^{(t^*-1)}),(M,\vecz),(M^{(t^*+1)},\vecz^{(t^*+1)}),\dots,(M^{(T)},\vecz^{(T)}))$ follows the distribution $\cD_{t^*}$. Similarly, if $(\vecx^*,M,\vecz)$ is sampled from the No distribution $\cN_{\textsf{RMD}}$ of $(\cD_Y,\cD_N)$-\textsf{RMD}, then $((M^{(1)},\vecz^{(1)}),\dots,(M^{(t^*-1)},\vecz^{(t^*-1)}),(M,\vecz),(M^{(t^*+1)},\vecz^{(t^*+1)}),\dots,(M^{(T)},\vecz^{(T)}))$ follows the distribution $\cD_{t^*-1}$. Thus, the advantage of $\Pi'$ is at least
\begin{align*}
&\Pr_{(M,\vecz)\sim\cY_{\textsf{RMD}},\Pi'}[\Pi'(M,\vecz)=\yes]-\Pr_{(M,\vecz)\sim\cN_{\textsf{RMD}},\Pi'}[\Pi'(M,\vecz)=\yes]\\
=\, &\Pr_{X\sim\cD_{t^*}}[\Pi(X)=\yes]-\Pr_{X\sim\cD_{t^*-1}}[\Pi(X)=\yes]=\Delta_{t^*}\geq\Delta/T \, .
\end{align*}
We conclude that there is a one-way protocol for $(\cD_Y,\cD_N)$-\textsf{RMD} using at most $s(T-1)$ bits of communication achieving advantage at least $\Delta/T$.
\end{proof}

As an immediate consequence of~\autoref{thm:communication lb matching moments} and~\autoref{lem:reduce RMD to simul-RMD} we get that $(\cD_Y,\cD_N,T)$-\textsf{simultaneous-RMD} requires $\Omega(\sqrt{n})$ bits of communication when the marginals of $\cD_Y$ and $\cD_N$ match.

\begin{lemma}\label{cor:communication lb matching moments sim}
For every $k\in\N$, there exists $\alpha_0 > 0$ such that for every $\alpha\in(0,\alpha_0)$ and $\delta>0$ the following holds: For every $T\in\N$ and every pair of distributions $\cD_Y,\cD_N\in \Delta(\{-1,1\}^k)$ with $\vecmu(\cD_Y) = \vecmu(\cD_N)$,
there exists $\tau > 0$ and $n_0$ such that for every $n \geq n_0$, every protocol for $(\cD_Y,\cD_N,T)$-\textsf{simultaneous-RMD} achieving advantage $\delta$ on instances of length $n$ requires $\tau\sqrt{n}$ bits of communication. 
\end{lemma}

We are now ready to prove \autoref{thm:main-negative dynamic}.

\subsubsection{Proof of Theorem~\ref{thm:main-negative dynamic}}\label{sec:proof_dynamic}

\begin{proof}[Proof of \autoref{thm:main-negative dynamic}]
The proof is a straightforward combination of \autoref{lem:csp value} and \autoref{cor:communication lb matching moments sim}
and so we pick parameters so that all these are applicable. 
%\mnote{Fix parameters.}
Given $\epsilon$ and $k$, let $\alpha^{(1)}_0$ be as given by \autoref{lem:csp value} and let $\alpha^{(2)}_0$ be as given by \autoref{cor:communication lb matching moments sim}. Let $\alpha = \min\{\alpha^{(1)}_0,\alpha^{(2)}_0\}$. Given this choice of $\alpha$, let $T_0$ be as given by \autoref{lem:csp value}. We set $T = T_0$ below. Let $n$ be sufficiently large. 

%Let $\alpha^1_0$ be as given by \autoref{cor:communication lb matching moments sim}. 
%Given $\epsilon>0$, let $\alpha=\alpha_0/2$  and $T= \lceil 10000/(\epsilon^2\alpha)\rceil $.
Throughout this proof we will be considering integer weighted instances of $\maxf$ on $n$ variables with constraints. Note that such an instance $\Psi$ can be viewed as a vector in $\Z^N$ where $N = O(n^k)$ represents the number of possibly distinct constraints applications on $n$ variables. 
Let $\Gamma = \{\Psi | \val_\Psi \geq \gamma - \epsilon\}$. 
Let $B = \{\Psi | \val_\Psi \leq \beta + \epsilon\}$. 
Suppose there exists a sketching algorithm $\ALG_1$ that solves $(\gamma-\epsilon,\beta+\epsilon)$-\textsf{Max-CSP}($f$) using at most $s(n)$ bits of space. Note that $\ALG_1$ must achieve advantage at least $1/3$ on the problem $(\Gamma,B)$. 
By running several independent copies of $\ALG_1$ and thresholding appropriately, we can get an algorithm $\ALG_2$ with space $O(s)$ and advantage $1 - \frac{1}{100}$ solving $(\Gamma,B)$.

Now, let $\COMP$ and $\COMB$ be the compression and combination functions as given by this sketching algorithm (see \autoref{def:sketching alg}). We use these to design a protocol for $(\cF,\cD_Y,\cD_N,T)$-\textsf{simultaneous-RMD} as follows.

%Now, by~\autoref{thm:AHLW} we get that there exists a linear sketching algorithm $\ALG_3$ to solve to solve the $(\Gamma,B)$ distinguishing problem  with advantage at least $1-12/100$. Let $q$, $A$ and $P$ be as given by this linear sketching algorithm. We use these to design a protocol for $(\cD_Y,\cD_N,T)$-\textsf{simultaneous-RMD} as follows.

Let $(M^{(t)},\vecz^{(t)})$ denote the input to the $t$-th player in $(\cD_Y,\cD_N,T)$-\textsf{simultaneous-RMD}. Each player turn his/her inputs into $\vecsigma^{(t)}=(\sigma^{(t)}_1,\dots,\sigma^{(t)}_{\alpha n})$ where $\sigma^{(t)}_i$ corresponds to the constraint $(\vecj^{(t)}(i),\vecz^{(t)}_i)$ with $\vecj^{(t)}_i\in[n]^k$ the indicator vector for the $i$-th hyperedge of $M^{(t)}$. Next, the players use shared randomness to compute the sketch of his/her input $\COMP(\vecsigma^{(t)})$ and send it to the referee. Finally, the referee computes the sketch for all streams $\COMB( \COMP(\vecsigma^{(1)}),\COMP(\vecsigma^{(2)}),\dots,\COMP(\vecsigma^{(T)})$ and outputs the corresponding answer.

%Next, the players use the shared randomness to sample a linear sketch matrix $A$ and send $A\sum_i\vecsigma^{(t)}_i \pmod q$ to the referee. Finally, the referee outputs  $P(\sum_t A\sum_i\vecsigma^{(t)}_i \pmod q)$. 

To analyze the above, note that the communication is $O(s)$. Next, by the advantage of the sketching algorithm, we have that 
\begin{equation}
\min_{\Psi \in \Gamma}[\ALG_2(\Psi) = 1] - \max_{\Psi \in B}[\ALG_2(\Psi) = 1] \geq 1-12/100.
\label{eq:dyn-one}    
\end{equation} 
Now we consider what happens when $\Psi \sim \cY_{\simul,n}$ and $\Psi \sim \cN_{\simul,n}$. By \autoref{lem:csp value} we have that $\Pr_{\Psi \sim \cY_{\simul,n}}[\Psi \in \Gamma] \geq  1 - o(1)$ and $\Pr_{\Psi \sim \cN_{\simul,n}}[\Psi \in B] \geq  1 - o(1)$. Combining with \autoref{eq:dyn-one} we thus get
\[
\Pr_{\Psi \sim \cY_{\simul,n}}[\ALG_2(\Psi) = 1] - \Pr_{\Psi \sim \cN_{\simul,n}}[\ALG_2(\Psi) = 1] \geq 1 - 12/100 - o(1) \geq 1/2,
\]
We thus get that there is a $O(s)$ simultaneous communication protocol for $(\cD_Y,\cD_N,T)$-\textsf{simultaneous-RMD} with advantage at least $1/2$. 

Now we conclude by applying \autoref{cor:communication lb matching moments sim} with $\delta = 1/2$ to get that $s = \Omega(\sqrt{n})/T = \Omega(\sqrt{n})$, thus yielding the theorem.
\end{proof}

}

\section{Communication Lower Bound: A Special Case of 1-wise Independence}\label{sec:BHBHM 1 wise}

The goal of this section is to prove a special case of~\autoref{thm:communication lb matching moments} when the distributions are $1$-wise independent, i.e., their marginals are all $0$. The main theorem of this section is summarized below.

\begin{theorem}[Lower bound for 1-wise distributions]\label{thm:communication lb 1 wise}
For every $k\geq 2$, there exists an $\alpha_0 > 0$ such that for every  $\alpha \in (0,1/\alpha_0)$, $\delta \in (0,1/2)$, and every $\cD_Y,\cD_N \in \Delta(\{-1,1\}^k)$ with $\mu(\cD_Y) = \mu(\cD_N) = 0^k$, there exists $\tau > 0$, and $n_0$ such that for every $n\geq n_0$, we have that every protocol for $(\cD_Y,\cD_N)$-RMD with parameter $\alpha$ that achieves advantage $\delta$ requires at least $\tau\sqrt{n}$ bits of communication on instances of length $n$.
\end{theorem}

Our proof of \autoref{thm:communication lb 1 wise} follows the methodology of \cite{GKKRW} with minor modifications as required by the RMD formulation. Their proof uses Fourier analysis to reduce the task of proving a communication lower bound to that of proving some combinatorial identities about randomly chosen matchings. We follow the same approach and this leads us to slightly different conditions about randomly chosen hypermatchings which requires a fresh analysis (though at the end our bounds are qualitatively similar to those in \cite{GKKRW}).

The proof is by contradiction. We show that if the number of bits communicated is $o(\sqrt{n})$, then the \textit{posterior distribution} of Bob's input $\vecz$ is close to the uniform distribution in total variation distance, and hence contradicts the assumed advantage of the protocol. In \autoref{thm:tvd 1 wise} we show that this total variation distance is small when Alice's message is a ``typical'' one, in that the number of Alice inputs leading to this message is not too small. We show immediately after stating \autoref{thm:tvd 1 wise} how to go from the case of typical messages to all messages, and this gives a proof of \autoref{thm:communication lb 1 wise}. 

For each $k$-uniform hypermatching $M$, distribution $\cD$ over $\{-1,1\}^k$, and a fixed Alice's message, the posterior distribution function $p_{M,\cD}:\{-1,1\}^{\alpha kn}\rightarrow[0,1]$ is defined as follows. For each $\vecz\in\{-1,1\}^{\alpha kn}$, let
\[
p_{M,\cD}(\vecz) := \Pr_{\substack{\vecx^*\in \{-1,1\}^n\\\vecb\sim\cD^{\alpha n}}}[\vecz=(M\vecx^*)\odot\vecb\ |\ M,\ \text{Alice's message}]=\Exp_{\vecx^*\in A}\Exp_{\vecb\sim\cD^{\alpha n}}[\mathbf{1}_{\vecz=(M \vecx^*)\odot \vecb}] \,,
\]
where $A\subset\{-1,1\}^n$ is the set of Alice's inputs that correspond to the message. If the number of bits communicated is at most $c$, then there exists a message such that the corresponding $A$ satisfies $|A|\geq2^{n-c}$.

\begin{theorem}\label{thm:tvd 1 wise}
For every $k\in\mathbb{N}$, there exists $\alpha_0 > 0$ such that for every  $\alpha\in(0,\alpha_0)$, $\delta\in(0,1/2)$, there exists a $\tau_0 > 0$ such that the following holds for every sufficiently large $n$. 
Let $A\subseteq\{-1,1\}^n$ be a set satisfying $|A| \geq 2^{n - \tau_0 \sqrt{n}}$, and let $\cD$ be a distribution over $\{-1,1\}^k$ satisfying $\Exp_{\veca\sim\cD}[a_j]=0$ for all $j\in[k]$.
Then 
\begin{equation}\label{eq:1 wise ub}
\Exp_M\left[\|p_{M,\cD}-U\|_{tvd}^2\right]\leq\delta^2 \, ,
\end{equation}
where $U$ denotes the uniform distribution over $\{-1,1\}^{k \alpha n}$.
\end{theorem}

Assuming \autoref{thm:tvd 1 wise}, we prove \autoref{thm:communication lb 1 wise} below.

\begin{proof}[Proof of~\autoref{thm:communication lb 1 wise}]
Let $\delta$ be as in the theorem statement and let $\delta'=\delta/8$.
Let $\tau_0$ be the constant given by \autoref{thm:tvd 1 wise} when invoked with parameter $\alpha$ and $\delta'$.
Let $\tau = \tau_0/2$, $c' = \tau_0\sqrt{n}$, and $c = c' - \log(1/\delta')$. Note that for large enough $n$, we have $c \geq \tau\sqrt{n}$. 

We will prove the theorem for this choice of $\tau$. The proof is by contradiction. 
Suppose there exists a protocol for $(\cD_Y,\cD_N)$-\textsf{RMD} on instances of length $n$ with advantage at least $\delta$ using at most $c$ bits of communication. Let $\cD_{unif}$ be the uniform distribution over $\{-1,1\}^k$.  By triangle inequality, there is a protocol for either $(\cD_Y,\cD_{unif})$-\textsf{RMD} or $(\cD_N,\cD_{unif})$-\textsf{RMD} with advantage at least $\frac{\delta}{2}$ using at most $c$ bits of communication. Without loss of generality, suppose there is protocol for $(\cD_Y,\cD_{unif})$-\textsf{RMD} with advantage at least $\frac{\delta}{2}$. We have
\[
\|p_{M,\cD_Y}-p_{M,\cD_{unif}}\|_{tvd}\geq\frac{\delta}{2} \, .
\]

Next, by Yao's principle~\cite{yao1977probabilistic} we may assume that the message sent by Alice is deterministic. Namely, the message partitions the set $\{-1,1\}^n$ of $\vecx^*$ into $2^{c}$ sets $A_1,A_2,\dots,A_{2^{c}}$. Using a simple counting argument, we can show that with probability at least $1-\delta'$, the message sent by Alice corresponds to a set $A_i\subset\{-1,1\}^n$ of size at least $2^{n-c-\log1/\delta'}\geq2^{n-c'}$. We call such an event $\textsf{GOOD}$. That is, 
\[
\textsf{GOOD}=\bigcup_{i\in[2^{c}]:|A_i|\geq2^{n-c'}}A_i \, .
\]

Now for each $A_i$ with $|A_i|\geq2^{n-c'}$, we apply~\autoref{thm:tvd 1 wise} with parameters $\alpha$ and $\delta'$ to get 
\[
\|p_{M,\cD_Y}-p_{M,\cD_{unif}}\|_{tvd}|_{\vecx^*\in A_i}=\Exp_M[\|p_{M,\cD_Y}-U\|_{tvd}|_{\vecx^*\in A_i}]\leq\delta' \, .
\]
Now, for $\vecx^*\sim\textsf{Unif}(\{-1,1\}^n)$, we have
\begin{align*}
\|p_{M,\cD_Y}-U\|_{tvd}&=\Pr[\vecx^*\in\textsf{GOOD}]\cdot\|p_{M,\cD_Y}-U\|_{tvd}|_{\vecx^*\in\textsf{GOOD}}\\
&+\Pr[\vecx^*\not\in\textsf{GOOD}]\cdot\|p_{M,\cD_Y}-U\|_{tvd}|_{\vecx^*\not\in\textsf{GOOD}}\\
&\leq1\cdot\delta'+\delta'\cdot1 = \frac{\delta}4 <\frac{\delta}{2} \, .
\end{align*}
But this contradicts our assumption that 
\[
\|p_{M,\cD_Y}-U\|_{tvd}=\|p_{M,\cD_Y}-p_{M,\cD_{unif}}\|_{tvd}\geq\frac{\delta}{2} \, .
\]
This completes the proof of~\autoref{thm:communication lb 1 wise}.
\end{proof}

The rest of this section is devoted to the proof of~\autoref{thm:tvd 1 wise}. In~\autoref{sec:1 wise reduce to combo}, we reduce the upper bound for~\autoref{eq:1 wise ub} to a combinatorial problem. Next, we analyze the combinatorial problem in~\autoref{sec:1 wise combo ub}, and finally complete the proof of~\autoref{thm:tvd 1 wise} in~\autoref{sec:1 wise wrap up}.

\subsection{Reduction to a combinatorial problem}\label{sec:1 wise reduce to combo}
Let $A\subseteq\{-1,1\}^n$ be the set of Alice’s inputs that correspond to the message. We define $f:\{-1,1\}^n\rightarrow\{0,1\}$ to be the indicator function of $A$, i.e., $f(\vecx^*)=1$ iff $\vecx^*\in A$. In this subsection, we apply Fourier analysis on the left hand side of~\autoref{eq:1 wise ub} and get an upper bound in terms of a combinatorial quantity related to the random matching and the Fourier coefficients of $f$. The reduction is summarized in the following lemma.

In what follows, we will write a vector $\vecs\in\{0,1\}^{\alpha k n}$ as a concatenation of $\alpha n$ vectors, i.e.,  $\vecs = (\vecs(1),\ldots,\vecs(\alpha n))$ where $\vecs(i) \in \{0,1\}^k$. We use $|\vecs(i)|$ to denote the Hamming weight of $\vecs(i)$.

\begin{lemma}\label{lem:step 1 1 wise}
Let $A\subseteq\{-1,1\}^n$ and $f:\{-1,1\}^n\rightarrow\{0,1\}$ be its indicator function. Let $k\in\mathbb{N}$ and $\alpha\in(0,1/100k)$. Let $\cD$ be a distribution over $\{-1,1\}^k$ such that $\mathbb{E}_{\veca\sim\cD}[a_j]=0$ for all $j\in[k]$. For each $\ell\in[n]$, let us denote by $\vecv_\ell\in\{0,1\}^n$, the vector where the first $\ell$ entries are $1$, and the remaining entries are $0$. We have
\[
\Exp_M[\|p_{M,\cD}-U\|_{tvd}^2]\leq\frac{2^{2n}}{|A|^2}\sum_{\ell\geq2}^{\alpha kn}g(\ell)\cdot\sum_{\substack{\vecv\in\{0,1\}^n\\|\vecv|=\ell}}\widehat{f}(\vecv)^2 \, ,
\]
where
\[
g(\ell)=\Pr_M\left[\exists \vecs\in\{0,1\}^{\alpha kn}\backslash\{0^{\alpha kn}\},\ |\vecs(i)|\neq1\, \forall i,\ M^\top \vecs=\vecv_\ell\right] \, .
\]
\end{lemma}
\begin{proof}%[Proof of~\autoref{lem:step 1  1 wise}]
By Cauchy–Schwarz inequality and \autoref{eq:dist}, 
\begin{align}
\Exp_M\left[\|p_{M,\cD}-U\|_{tvd}^2\right]&\leq2^{2\alpha kn}\Exp_M\left[\|p_{M,\cD}-U\|_2^2\right]\nonumber\\
&=2^{2\alpha kn}\Exp_M\left[\sum_{\vecs\in\{0,1\}^{\alpha kn}\backslash\{0^{\alpha kn}\}}\widehat{p_{M,\cD}}(\vecs)^2\right] \, . \label{eq:upper bound tvd 1 wise}
\end{align}

The following claim shows that the expected sum of the Fourier coefficients (corresponding to non-empty subsets of $[\alpha k n]$) of the posterior distribution $p_{M,\cD}$ can be upper bounded by an expected sum of certain Fourier coefficients of the indicator function $f$.

\begin{claim}\label{claim:1 wise fourier bound}
\[
\Exp_M[\|p_{M,\cD}-U\|_{tvd}^2]\leq\frac{2^{2n}}{|A|^2}\sum_{\vecs\in\GOOD\backslash\{0^{\alpha kn}\}}\Exp_M\left[\widehat{f}(M^\top \vecs)^2\right] \, . 
\]
\end{claim}
\begin{proof}

For every $\vecs\in\{0,1\}^{\alpha kn}\backslash\{0^{\alpha kn}\}$, consider $\vecs\in\{0,1\}^{\alpha kn}$ to be $\alpha n$ blocks $\vecs(1),\dots,\vecs(\alpha n)\in\{0,1\}^{k}$ of length $k$. Observe that
\begin{align*}
\widehat{p_{M,\cD}}(\vecs)&= \frac{1}{2^{\alpha kn}}\sum_{\vecz\in\{-1,1\}^{\alpha kn}}p_{M,\cD}(\vecz)\prod_{\substack{i\in[\alpha n],j\in[k]\\s(i)_{j}=1}}z(i)_{j} \, .
\intertext{By substituting $p_{M,\cD}(\vecz)=\Exp_{\vecx^*\in A}\Exp_{\vecb\sim\cD^{\alpha n}}[\mathbf{1}_{\vecz=M \vecx^*\odot \vecb}]$, the equation becomes}
\widehat{p_{M,\cD}}(\vecs)&=\frac{1}{2^{\alpha kn}}\cdot\Exp_{\vecx^*\in A}\left[\prod_{\substack{i\in[\alpha n],j\in[k]\\s(i)_{j}=1}}(M \vecx^*)_{i,j}\right]\Exp_{\vecb\sim\cD^{\alpha n}}\left[\prod_{\substack{i\in[\alpha n],j\in[k]\\s(i)_{j}=1}}b(i)_{j}\right] \\
&=\frac{1}{2^{\alpha kn}}\cdot\Exp_{\vecx^*\in A}\left[\prod_{\substack{i\in[\alpha n],j\in[k]\\s(i)_{j}=1}}(M \vecx^*)_{i,j}\right]\prod_{i\in [\alpha n]}\Exp_{\vecb(i)\sim\cD}\left[\prod_{\substack{j\in[k]\\s(i)_{j}=1}}b(i)_{j}\right]\, .
\intertext{Since $\Exp_{\veca\sim\cD}[a_j]=0$ for all $j\in[k]$, the right hand side expression becomes zero if there exists $i\in[\alpha n]$ such that $|\vecs(i)|=1$. Define $\GOOD:=\{\vecs\in\{0,1\}^{\alpha kn}\, |\, |\vecs(i)|\neq1\ \forall i\}$. We have}
\widehat{p_{M,\cD}}(\vecs)&\leq\frac{1}{2^{\alpha kn}}\cdot\left|\Exp_{\vecx^*\in A}\left[\prod_{\substack{i\in[\alpha n],j\in[k]\\s(i)_{j}=1}}(M \vecx^*)_{i,j}\right]\right|\cdot\mathbf{1}_{\vecs\in\GOOD} \, .
\intertext{Since each row and column in $M$ has at most one non-zero entry, we can rewrite the right hand side as}
&=\frac{1}{2^{\alpha kn}}\cdot\left|\Exp_{\vecx^*\in A}\left[\prod_{\substack{i\in[n]\\(M^\top \vecs)_{i}=1}} \vecx^*_{i}\right]\right|\cdot\mathbf{1}_{\vecs\in\GOOD}
\end{align*}
Now we relate the above quantity to the Fourier coefficients of $f$. Recall that $f$ is the indicator function of the set $A$ and hence for each $\vecv\in\{0,1\}^n$, we have
\[
\widehat{f}(\vecv)=\frac{1}{2^n}\sum_{\vecx^*}f(\vecx^*)\prod_{i\in[n]:v_i=1}\vecx^*_i=\frac{1}{2^n}\sum_{\vecx^*\in A}\prod_{i\in[n]:v_i=1}\vecx^*_i \, .
\]
Thus, the Fourier coefficient of $p_M$ corresponding to a set $\vecs\in\{0,1\}^{\alpha kn}$ can be bounded as follows:
\begin{equation}\label{eq:1 wise fourier coefficient}
\widehat{p_{M,\cD}}(\vecs)\leq\frac{1}{2^{\alpha kn}}\cdot\frac{2^n}{|A|}\left|\widehat{f}(M^\top \vecs)\right|\cdot\mathbf{1}_{\vecs\in\GOOD} \, .
\end{equation}
By plugging~\autoref{eq:1 wise fourier coefficient} into~\autoref{eq:upper bound tvd 1 wise}, we have the desired bound, and this completes the proof of~\autoref{claim:1 wise fourier bound}.
\end{proof}

It follows from~\autoref{claim:1 wise fourier bound} that
\begin{align*}
\Exp_M[\|p_{M,\cD}-U\|_{tvd}^2]&\leq\frac{2^{2n}}{|A|^2}\sum_{\vecs\in\GOOD\backslash\{0^{\alpha kn}\}}\Exp_M\left[\widehat{f}(M^\top \vecs)^2\right] \, . 
\intertext{Since for a fixed $M$, the map $M^\top$ is injective, the right hand side of the above inequality has the following combinatorial form.}
&\frac{2^{2n}}{|A|^2}\sum_{\vecv\in\{0,1\}^n\backslash\{0^n\}}\Pr_M\left[\exists \vecs\in\GOOD\backslash\{0^{\alpha kn}\},\ M^\top \vecs=\vecv\right]\widehat{f}(\vecv)^2 \, . 
\intertext{By symmetry, the above probability term will be the same for $\vecv$ and $\vecv'$ which have the same Hamming weight. For each $\ell\in[n]$, denote $g(\ell)=\Pr_M\left[\exists \vecs\in\GOOD\backslash\{0^{\alpha kn}\},\ M^\top \vecs=\vecv_\ell\right]$. Therefore, the expression simplifies to}
&\frac{2^{2n}}{|A|^2}\sum_{\ell\geq1}^ng(\ell)\cdot\sum_{\substack{\vecv\in\{0,1\}^n\\|\vecv|=\ell}}\widehat{f}(\vecv)^2 \, .
\intertext{
Note that for $\ell=1$ or $\ell>\alpha kn$, $g(\ell)=0$ by definition. Thus, the above expression further simplifies to the following:}
&\frac{2^{2n}}{|A|^2}\sum_{\ell\geq2}^{\alpha kn}g(\ell)\cdot\sum_{\substack{\vecv\in\{0,1\}^n\\|\vecv|=\ell}}\widehat{f}(\vecv)^2 \, .
\end{align*}
We conclude that
\[
\Exp_M[\|p_{M,\cD}-U\|_{tvd}^2]\leq\frac{2^{2n}}{|A|^2}\sum_{\ell\geq2}^{\alpha kn}g(\ell)\cdot\sum_{\substack{\vecv\in\{0,1\}^n\\|\vecv|=\ell}}\widehat{f}(\vecv)^2 \, .
\]

This completes the proof of~\autoref{lem:step 1 1 wise}.
\end{proof}

\subsection{An upper bound for the combinatorial problem}\label{sec:1 wise combo ub}
In this subsection, we upper bound the combinatorial term $g(\ell)$ in~\autoref{lem:step 1 1 wise}. The result is summarized in the following lemma.

\begin{lemma}\label{lem:step 2 1 wise}
For every $k$, there exists an $\alpha_0 > 0$ such that for every $\alpha \in (0,\alpha_0)$,
and for every
$n$ and $\ell \leq n/2$, we have 
\[
g(\ell)=\Pr_M\left[\exists \vecs\neq0,\ |\vecs(i)|\neq1\ \forall i,\ M^\top \vecs=\vecv_\ell\right]\leq \left(\frac{\ell}n\right)^{\ell/2} \, . 
\]
\end{lemma}

\begin{proof}
We set $\alpha_0= (1/(2e^2k))^k$ so that 
$2 \alpha_0 ^{1/k} e^{3/2} k \leq 1$. We reformulate our events. Instead of fixing $\vecv = \vecv_\ell$ and picking the matching $M$ at random, we note that it is equivalent to fixing the matching $M$ and letting $\vecv$ be a uniformly random vector of weight $\ell$. We thus let $M$ be the matching $e_1,\ldots,e_{\alpha n}$, where $e_i = \{(i-1)k + 1,\ldots,(i-1)k+k\}$. Letting $V$ denote the support of the vector $\vecv$, the event we wish to consider is: ``$V \subseteq [k \alpha n]$ and $|V \cap e_i|\neq 1$ for every $i \in [\alpha n]$.''

\begin{figure}[ht]
    \centering
    \includegraphics[width=6cm]{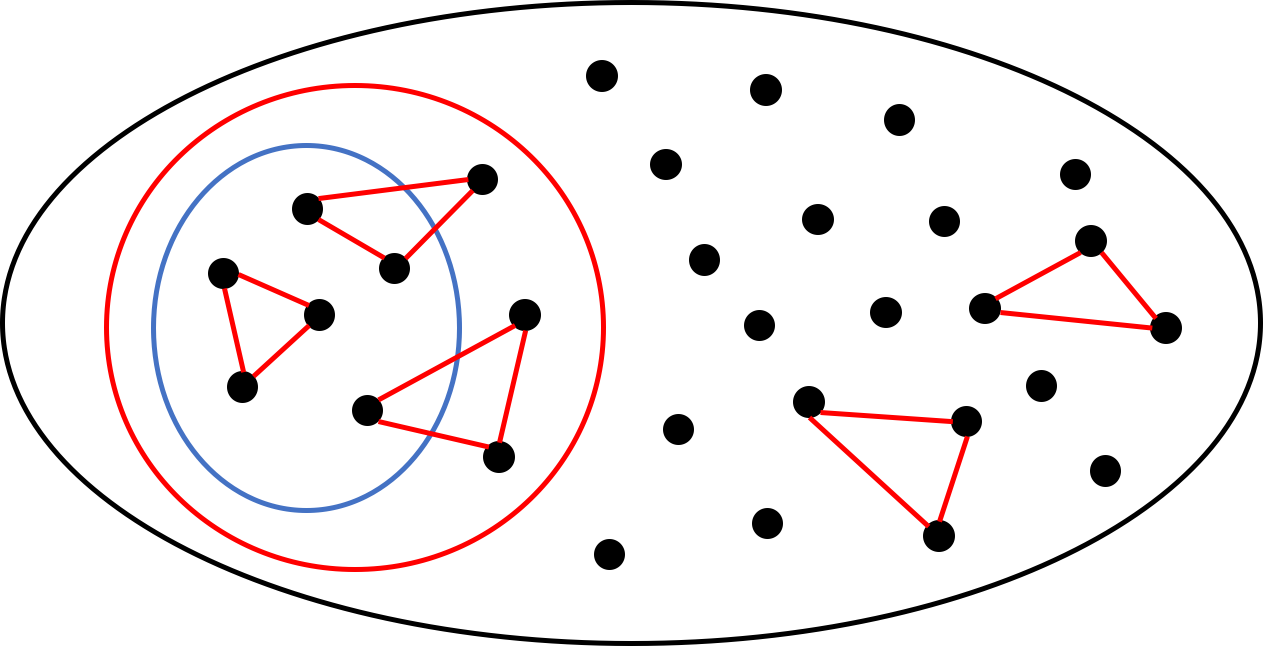}
    \caption{An example with $n=30$, $k=3$, $\alpha=0.5$, and $\ell=6$. The red triangles denote the edges $e_1,\dots,e_5$. The blue circle denotes the set $V$ and the red circle denotes the set $T$ with $t = 3$. Note that the figure illustrates the over-counting we do in the proof of the lemma - the set $V$ actually intersects one of the edges just once, and so should not be counted. Our counting will nevertheless include the set since it is contained in at most $\ell/2 = 3$ edges.}
    \label{fig:matching 1}
\end{figure}

We bound the probability as follows. Let $T = \{i \in [\alpha n]\mid  V\cap e_i \ne \emptyset\}$ denote the set of edges that touch $V$, and let $|T| = t$. Note that $\ell/k \leq t \leq \ell/2$, where the latter inequality follows from the fact that every intersection is of size at least $2$. We pick $V$ by first picking $T$ (there are at most 
$\binom{\alpha n}t$ ways of doing this), and then picking $V$ as a subset of the vertices incident to the edges of $T$ (there are $\binom{kt}\ell$ ways of doing this). (See \autoref{fig:matching 1}.)  Summing over $t$ and dividing by the total number of choices of $V$ gives the final bound. We give the calculations below (which use the inequalities $(a/b)^b \leq \binom{a}b \leq (ea/b)^b$). 
\begin{align*}   
& \Pr_{V} [ V \subseteq [k \alpha n], |V \cap e_i| \neq 1, \forall i\in [\alpha n] ] \\
& \leq  \frac{\sum_{t=\ell/k}^{\ell/2} \binom{\alpha n}t \binom{kt}\ell}{\binom{n}\ell} \\
& \leq \sum_{t=\ell/k}^{\ell/2} \left(\frac{e \alpha n}{t}\right)^t \cdot \left(\frac{e kt}{\ell}\right)^\ell  \cdot \left(\frac{n}{\ell}\right)^{-\ell} \\
& =  \sum_{t=\ell/k}^{\ell/2} e^{t+\ell} \alpha^t k^\ell (t/n)^{\ell-t} \\
& \leq  \alpha^{\ell/k} e^{3\ell/2} k^\ell (\ell/n)^{\ell/2} \sum_{t'=0}^{\infty} (\ell/n)^{t'}\\
& \leq  2 (\alpha^{1/k} e^{3/2} k)^\ell (\ell/n)^{\ell/2}\\
& \leq (2 \alpha^{1/k} e^{3/2} k)^\ell (\ell/n)^{\ell/2}\\
& \leq (\ell/n)^{\ell/2} \, .
\end{align*}

\end{proof}

\subsection{Proof of Theorem~\ref{thm:tvd 1 wise}}\label{sec:1 wise wrap up}
\begin{proof}[Proof of~\autoref{thm:tvd 1 wise}]

By~\autoref{lem:step 1 1 wise} and~\autoref{lem:step 2 1 wise}, we have
\begin{align*}
\Exp_M[\|p_{M,\cD}-U\|_{tvd}^2]&\leq\frac{2^{2n}}{|A|^2}\cdot\sum_{\ell=2}^{\alpha kn}\frac{\ell^{\ell/2}}{n^{\ell/2}} \sum_{\substack{\vecv\in\{0,1\}^n\\|\vecv|=\ell}}\widehat{f}(\vecv)^2 \, .
\end{align*}

We use \autoref{lem:kkl} to upper bound the sum of level-$\ell$ Fourier coefficients for small $\ell$ as follows. Let $c = \tau_0\sqrt{n}$ so that $|A| \geq 2^{n-c}$. 
For $\ell\in[4c]$, we have
\begin{align*}
\frac{2^{2n}}{|A|^2}\sum_{\substack{\vecv\in\{0,1\}^n\\|\vecv|=\ell}}\widehat{f}(\vecv)^2&\leq\left(\frac{4\sqrt{2}c}{\ell}\right)^\ell \, .
\end{align*}

Next, we apply Parseval's inequality (\autoref{prop:parseval}) and have $\sum_{\vecv}\widehat{f}(\vecv)^2\leq1$. Thus,
\begin{align*}
\Exp_M[\|p_{M,\cD}-U\|_{tvd}^2]
&\leq \sum_{\ell=2}^{4c}\frac{\ell^{\ell/2}}{n^{\ell/2}} \cdot\left(\frac{4\sqrt{2}c}{\ell}\right)^\ell+\frac{2^{2n}}{|A|^2}\cdot\max_{4c<\ell\leq\alpha kn}\left\{\frac{\ell^{\ell/2}}{n^{\ell/2}}\right\}
\intertext{The second term on the right hand side is maximized at $\ell=4c+1$, and hence}
\Exp_M[\|p_{M,\cD}-U\|_{tvd}^2]&\leq\sum_{\ell=2}^{4c}\left(\frac{32c^2}{\ell \cdot n}\right)^{\ell/2} + \left(\frac{8c}{n}\right)^{2c} \\
&\leq \sum_{\ell=2}^{4c} (16 \tau_0^2)^{\ell/2} + (8\tau_0)^{2c}\\
&\leq\delta^2 \, ,
\end{align*}
where the final expression determines our choice of $\tau_0$. Specifically, we set 
$\tau_0 = \delta/2^6$ so that each term is at most $\delta^2/2$. 
This completes the proof of~\autoref{thm:tvd 1 wise}.
\end{proof}
\section{Communication Lower Bound: General Case}\label{sec:BHBHM general}

In this section we finally prove \autoref{thm:communication lb matching moments}. In other words we
show that for every $\cD_Y,\cD_N\in \Delta(\{-1,1\}^k)$ with matching marginals, any protocol for $(\cD_Y,\cD_N)$-\textsf{RMD} with positive advantage requires $\Omega(\sqrt{n})$ bits of communication. We start with an overview.

The first step is to observe that we can prove indistinguishability of {\em some} distributions with matching non-zero marginals. For example, given that $\cD_1 = \textsf{Unif}(\{(-1,-1),(1,1)\}$ is indistinguishable from $\cD_2 = \textsf{Unif}(\{-1,1\}^2)$, it can also be shown that $\cD'_1 = \frac12\{(1,1)\} + \frac12\cD_1$ is indistinguishable from $\cD'_2 = \frac12\{(1,1)\} + \frac12\cD_2$ (see \autoref{lem:polarization indis} for a related statement). Note that $\cD'_1$ and $\cD'_2$ are distributions with non-zero but matching marginals.

The bulk of this section is devoted to proving that for every pair of distributions $\cD_Y$ and $\cD_N$, we can find a path (a sequence) of intermediate distributions $\cD_Y = \cD_0,\cD_1,\ldots,\cD_L = \cD_N$ such that adjacent pairs in this sequence are indistinguishable by a ``basic'' argument, where a basic argument is a combination of an indistinguishability result from~\autoref{thm:communication lb 1 wise} and a shifting argument formalized in \autoref{lem:polarization indis}. Our proof comes in the following steps:
\begin{enumerate}
    \item For every marginal vector $\vecmu$, we identify a {\em canonical} distribution $\cD_{\vecmu}$ that we use as the endpoint of the path. So it suffices to prove that for all $\cD$, $\cD$ is indistinguishable from $\cD_{\vecmu(\cD)}$, i.e., there is a path of finite length from $\cD$ to $\cD_{\vecmu(\cD)}$.
    \item We identify a measure $\Phi(\cD)$ associated with distributions that helps measure progress on a path. Among distributions with marginal $\vecmu(\cD)$, this measure is uniquely maximized by $\cD_{\vecmu(\cD)}$. We show that for every distribution $\cD$ that is not canonical one can take a basic step that increases $\vecmu(\cD)$. 
    Unfortunately the measure $\Phi$ is real-valued and the increases per step can be by arbitrarily small amounts, so we are not done.
    \item We give a combinatorial proof that there is a path of finite length (some function of $k$) that takes us from an arbitrary distribution to the canonical one.
\end{enumerate}
Putting the three ingredients together, along with a proof that a ``basic step'' is indistinguishable  gives us the final theorem.  

We start with the definition of the chain and the canonical distribution. For a distribution $\cD \in \Delta(\{-1,1\}^k)$, its support is the set  $\supp(\cD) = \{\veca\in \{-1,1\}^k\, |\, \cD(\veca) > 0\}$. Next, we consider the following partial order on $\{-1,1\}^k$. For vectors $\veca,\vecb \in \{-1,1\}^k$ we use the notation $\veca \leq \vecb$ if $a_i \leq b_i$ for every $i \in [k]$. Further we use $\veca < \vecb$ if $\veca \leq \vecb$ and $\veca \ne \vecb$.

\begin{definition}[Chain]\label{def:chain}
We refer to a sequence $\veca(0)< \veca(1) < \cdots < \veca(\ell)$, $\veca(i) \in \{-1,1\}^k$ for every $i \in \{0,\ldots,\ell\}$,  as a {\em chain} of length $\ell$. Note that chains in $\{-1,1\}^k$ have length at most $k$.
\end{definition}

\begin{definition}[Canonical distribution]\label{def:canonical}
Given a vector of marginals $\vecmu = (\mu_1,\ldots,\mu_k) \in [-1,1]^k$, the {\em canonical distribution} associated with $\vecmu$, denoted $\cD_{\vecmu}$, is defined as follows: Let $\rho:[k]\to[k]$ be a permutation such that $-1\leq \mu_{\rho(1)} \leq \cdots \leq \mu_{\rho(k)} \leq 1$. For $i \in \{0,\ldots,k\}$, let $\veca(i) \in \{-1,1\}^k$ be given by $\veca(i)_j = -1$ if $j \in \{\rho(1),\ldots,\rho(k-i)\}$ and $\veca(i)_j = 1$ otherwise. (Note that $\veca(0) < \cdots < \veca(k)$.) Then $\cD_{\vecmu}(\veca(i)) = \frac{1}{2}(\mu_{\rho(k-i+1)} - \mu_{\rho(k-i)})$, where we define $\mu_{\rho(0)} = -1$ and $\mu_{\rho(k+1)} = 1$. Finally, $\cD_{\vecmu}(\veca) = 0$ for all $\veca\notin \{\veca(0),\ldots,\veca(k)\}$. 
\end{definition}

It is easy to verify that $\cD_{\vecmu}$ is indeed a distribution, and that it has the desired marginals, i.e., $\vecmu(\cD_{\vecmu}) = \vecmu$. 
Note that a distribution is a canonical distribution if and only if its support is a chain. Furthermore, the canonical distribution is uniquely determined even though $\rho$, and hence the chain $\veca(0),\ldots,\veca(k)$, may not be uniquely determined. This is so since $\rho$ is non-unique only if $\mu_{\rho(i)} = \mu_{\rho(i+1)}$ for some $i$, and in this case $\cD_{\vecmu}(\veca(i)) = 0$ so the ``non-uniqueness of $\veca(i)$ does not affect $\cD_{\vecmu}$. 

Next we define a potential associated with distributions. For a distribution $\cD \in \Delta(\{-1,1\}^k)$ define its potential to be 
$$
\Phi(\cD) = \Exp_{\vecb\sim\cD}\left[ ~\left(\sum_{j\in[k]} b_j\right)^2 ~\right]\, .
$$ 
We will show shortly that $\cD_{\vecmu}$ is the distribution with maximum potential among all distributions with marginal $\vecmu$. In the process of showing this we will introduce a ``polarization operator'' which maps a distribution $\cD$ to a new one that increases the potential for typical distributions. Since this operator is useful also for further steps, we start with defining this operator and analyzing its effect on the potential.

\subsection{Polarization}
Briefly, suppose the support of a distribution contains both $(-1)^i(1)^{k-i}$ and $1^i(-1)^{k-i}$. Then the polarization operator moves some of this mass (as much as possible while maintaining the property that the result is a distribution) to the more ``polarized'' points $(-1)^k$ and $1^k$. The operator is defined more generally to allow the two starting points to agree on some coordinates. To define this operator, the following notation will be useful.

For $\vecu,\vecv \in \{-1,1\}^k$, let 
$\vecu \wedge \vecv = (\min\{u_1,v_1\},\ldots,\min\{u_k,v_k\})$ and let 
$\vecu \vee \vecv = (\max\{u_1,v_1\},\ldots,\max\{u_k,v_k\})$. We say $\vecu$ and $\vecv$ are incomparable if $\vecu \not\leq \vecv$ and $\vecv \not\leq \vecu$. Note that if $\vecu$ and $\vecv$ are incomparable then $\{\vecu,\vecv\}$ and $\{\vecu \vee \vecv,\vecu\wedge \vecv \}$ are disjoint\footnote{To see this, suppose $\vecu=\vecu\wedge\vecv$, then we have $u_j=\min\{u_j,v_j\}$ for all $j\in[k]$ and hence $\vecu\leq\vecv$, which is a contradiction. The same analysis works for the other cases.}.

\begin{definition}[Polarization (update) operator]\label{def:polarization operator}
Given a distribution $\cD \in \Delta(\{-1,1\}^k)$ and incomparable elements $\vecu,\vecv\in \{-1,1\}^k$, we define the $(\vecu,\vecv)$-{\em polarization} of $\cD$, denoted $\cD_{\vecu,\vecv}$, to be the distribution as given below. Let $\epsilon = \min\{\cD(\vecu),\cD(\vecv)\}$. 
\[
\cD_{\vecu,\vecv}(\vecb) =\left\{\begin{array}{ll}
\cD(\vecb)-\epsilon     & ,\ \vecb\in\{\vecu,\vecv\} \\
\cD(\vecb)+\epsilon       & ,\ \vecb \in \{\vecu \vee \vecv,\vecu\wedge \vecv \} \\
\cD(\vecb) & ,\ \text{otherwise.}
\end{array}\right.
\]
We refer to $\epsilon(\cD,\vecu,\vecv) = \min\{\cD(\vecu),\cD(\vecv)\}$ as the polarization amount.
\end{definition}

It can be verified that the polarization operator preserves the marginals, i.e., $\vecmu(\cD) = \vecmu(\cD_{\vecu,\vecv})$. Note also that this operator is non-trivial, i.e., $\cD_{\vecu,\vecv}=\cD$, if  $\{\vecu,\vecv\} \not\subseteq \supp(\cD)$.
By correlating the ``$+1$''s and ``$-1$''s, the polarization operator makes the support of $\cD$ more polarized in the sense quantified in the following lemma.

\begin{lemma}[Polarization increases potential]\label{lem:polarization}
Let $\cD\in \Delta(\{-1,1\}^k)$ be a distribution with marginal vector $\vecmu=\vecmu(\cD)$ and let $\vecu,\vecv\in \supp(\cD)$ be incomparable. Then we have
\[
\Phi(\cD_{\vecu,\vecv}) = \Phi(\cD) + 8 \cdot \epsilon \cdot s \cdot t 
\]
where $\epsilon = \epsilon(\cD,\vecu,\vecv)$ is the polarization amount, and $s = |\{ j \in [k]\, |\, u_j = -v_j = 1\}|$ and $t = |\{ j \in [k]\, |\, u_j = -v_j = -1\}|$. In particular $\Phi(\cD_{\vecu,\vecv}) > \Phi(\cD)$.
\end{lemma}

\begin{proof}
We look at the difference $\Phi(\cD_{\vecu,\vecv}) - \Phi(\cD)$. Let $\ell = \sum_{j \in [k] : u_j = v_j} u_j$. We have:
\begin{align*}
\Phi(\cD_{\vecu,\vecv}) - \Phi(\cD) &= \sum_{\vecb \in \{-1,1\}^k} (\cD_{\vecu,\vecv}(\vecb)-\cD(\vecb)) \cdot \Phi(\vecb) \\
&= \epsilon\cdot (
\Phi(\vecu \wedge \vecv)
+ \Phi(\vecu \vee \vecv)
- \Phi(\vecu)
- \Phi(\vecv))\\
&= \epsilon\cdot ((\ell + s + t)^2 + (\ell - s - t)^2 - (\ell + s - t)^2 - (\ell - s + t)^2) \\
&= 8\cdot \epsilon \cdot s \cdot t \, .
\end{align*}
Finally note that $s,t > 0$ since $\vecu$ and $\vecv$ are incomparable, and $\epsilon > 0$ since $\vecu,\vecv \in \supp(\cD)$, thus yielding $\Phi(\cD_{\vecu,\vecv}) > \Phi(\cD)$.
\end{proof}

\begin{lemma}[$\cD_{\vecmu}$ maximizes potential]\label{lem:potential polarized}
For every distribution $\cD \in \Delta(\{-1,1\}^k)$ with $\vecmu = \vecmu(\cD)$ we have $\Phi(\cD) \leq \Phi(\cD_{\vecmu})$. Furthermore the inequality is strict if $\cD \ne \cD_{\vecmu}$. 
\end{lemma}
\begin{proof}
Let $\cD^*$ be a  distribution with marginal $\vecmu$ that maximized $\Phi(\cD)$. Suppose there exist incomparable $\vecu,\vecv\in\supp(\cD^*)$, then by~\autoref{lem:polarization} we have that $\Phi(\cD^*)<\Phi(\cD^*_{\vecu,\vecv})$ contradicting the maximality of $\cD^*$. It follows that there are no incomparable elements in $\supp(\cD^*)$, or in other words, $\supp(\cD^*)$ is a chain. We now show that this implies $\cD^* = \cD_{\vecmu}$.

More specifically we show that any distribution $\cD$ supported on a chain is uniquely determined by its marginal $\vecmu$. 
To see this, let $\rho:[k]\to[k]$ be a bijection such that $\mu_{\rho(j)} \leq \mu_{\rho(j+1)}$ for all $j$. Let $\tau_0 < \tau_1 < \cdots < \tau_\ell$ be the attainable values of $\vecmu$, i.e., $\{\tau\, |\, \exists j \in [k]\, s.t.\, \mu_j = \tau\} = \{\tau_0,\ldots,\tau_\ell\}$. For $0 \leq i \leq \ell$, let $\veca(i)$ be given by $\veca(i)_j = -1$ if $\mu_j \leq \tau_{\ell-i}$ and $\veca(i)_j = 1$ otherwise. Note that $\veca(0)<\cdots<\veca(\ell)$. It can be verified that $\supp(\cD^*) = \{\veca(0),\ldots,\veca(\ell)\}$, and $\cD^*(\veca(i))$ is uniquely defined for all $i$.

\begin{claim}
$\supp(\cD^*) = \{\veca(0),\ldots,\veca(\ell)\}$, and $\cD^*(\veca(i)) = (\tau_{\ell-i+1}-\tau_{\ell-i})/2$, where $\tau_{-1} = -1$ and $\tau_{\ell+1} = 1$.
\end{claim}
\begin{proof}
For the sake of contradiction, assume $\supp(\cD^*)=\{\veca'(0),\ldots,\veca'(\ell')\}\neq\{\veca(0),\ldots,\veca(\ell)\}$ where $\veca'(0)<\veca'(1)<\cdots<\veca'(\ell')$ is a chain. Let $0\leq i\leq\min\{\ell,\ell'\}$ be the smallest $i$ such that $\veca(i)\neq\veca'(i)$. Consider the following three situations: (i) $\veca(i)<\veca'(i)$, (ii) $\veca(i)>\veca'(i)$, and (iii) $\veca(i)$ and $\veca'(i)$ are incomparable.

For (i) and (iii), due to the construction of $\{\veca(0),\ldots,\veca(\ell)\}$ and the fact that $\{\veca'(0),\ldots,\veca'(\ell')\}$ is a chain, we have that for each $j,j'\in[k]$ with $\tau_{i-2}<\mu_j,\mu_{j'}\leq\tau_i$, $\veca'(i')_j=\veca'(i')_{j'}$ for all $0\leq i'\leq\ell'$. This implies that $\mu_j=\mu_{j'}$ which is a contradiction because there are two attainable values $\tau_i$ and $\tau_{i-1}$ lie in the interval $(\tau_{i-2},\tau_i]$. Similar argument also works for situation (ii).

We conclude that $\supp(\cD^*) = \{\veca(0),\ldots,\veca(\ell)\}$. It is immediate to see that $\cD^*(\veca(i))$ is uniquely defined for all $i$ by solving the following linear system.
\[
\vecmu=
\begin{bmatrix}
|&|&&|\\
\veca(0)&\veca(1)&\cdots&\veca(\ell)\\
|&|&&|\\
\end{bmatrix}
\begin{bmatrix}
\cD^*(\veca(0))\\\cD^*(\veca(1))\\\cdots\\\cD^*(\veca(\ell))
\end{bmatrix} \, .
\]
Note that by the construction of $\{\veca(0),\ldots,\veca(\ell)\}$, the matrix has full rank, and, hence, there is a unique solution. It can be verified that the solution is given by $\cD^*(\veca(i)) = (\tau_{\ell-i+1}-\tau_{\ell-i})/2$, where $\tau_{-1} = -1$ and $\tau_{\ell+1} = 1$.
\end{proof}

In summary, $\cD^*$ is uniquely determined by $\vecmu(\cD)$ and its support is a chain. This implies that $\cD^* = \cD_{\vecmu}$, so $\cD_{\vecmu}$ is the unique distribution that maximizes the potential.
\end{proof}

\subsection{Indistinguishability of a polarization update}
Our next observation is that for every distribution $\cD$ with incomparable elements $\vecu$, $\vecv$ in their support, $\cD$ is indistinguishable, in the RMD problem, from its $(\vecu,\vecv)$-polarization $\cD_{\vecu,\vecv}$.

\newcommand{\pivar}{\pi_{\mathrm{var}}}
\newcommand{\pimat}{\pi_{\mathrm{matching}}}

\begin{lemma}[Polarization update preserves indistinguishability]\label{lem:polarization indis}
Let $\alpha_0(k)$ be as given in \autoref{thm:communication lb 1 wise}.
Let $k\in\N$, $\alpha\in(0,\alpha_0)$, $\delta\in(0,1/2)$. Then for every distribution $\cD\in \Delta(\{-1,1\}^k)$ and incomparable $\vecu,\vecv\in\supp(\cD)$ there exists $\tau > 0$ and $n_0$ such that for every $n \geq n_0$ 
every protocol for $(\cD,\cD_{\vecu,\vecv})$-\textsf{RMD} achieving advantage $\delta$ on instances of length $n$ requires $\tau\sqrt{n}$ bits of communication. 
\end{lemma}

We prove~\autoref{lem:polarization indis} by a reduction. We show that there exists a pair of distributions $\cD_Y$ and $\cD_N$ with marginals being zero such that given a protocol $\Pi$ for $(\cD,\cD_{\vecu,\vecv})$-\textsf{RMD}, we can get a protocol $\Pi'$ for $(\cD_Y,\cD_N)$-\textsf{RMD}. We then use \autoref{thm:communication lb 1 wise} to get a lower bound on the communication of $\Pi'$ and thus of $\Pi$.
Specifically, we divide the proof into three steps. In step one, we define $\cD_Y$ and $\cD_N$ and provide intuition on the reduction. Next, we formally describe the reduction by designing a protocol for $(\cD_Y,\cD_N)$-\textsf{RMD} from a protocol for $(\cD,\cD_{\vecu,\vecv})$-\textsf{RMD}. Finally, we prove the correctness of the reduction and wrap up the proof of~\autoref{lem:polarization indis}.

\paragraph{Step 1: The auxiliary distributions $\cD_Y$ and $\cD_N$.}
We start by defining $\cD_Y$ and $\cD_N$. Let $S=\{i\in[k]\, |\, u_i\neq v_i\}$.
Let $k'=|S|$.
Without loss of generality, we re-index the coordinates and assume $S=\{1,2,\dots,k'\}$. 
Let $\veca=\vecu|_{S}$ so that $\vecv|_{S}=-\veca$. We also let $\tilde{\vecu} = \vecu|_{\bar{S}}$ denote the common parts of $\vecu$ and $\vecv$. 
Let $\cD_Y$ be the  uniform distribution over $\{\veca,-\veca\}$, and $\cD_N$ be the  uniform distribution over $\{1^{k'},(-1)^{k'}\}$. Note that $\vecmu(\cD_Y) = \vecmu(\cD_N) = 0^{k'}$.
Let $\cD_1 = \textsf{Unif}(\{\vecu,\vecv\})$ and $\cD_2 = \textsf{Unif}(\{\vecu\vee \vecv,\vecu \wedge \vecv\})$.
Let $\epsilon = \epsilon(\cD,\vecu,\vecv)$ be the polarization amount. Let $\cD_0 \in \Delta(\{-1,1\}^k)$ be such that 
$\cD = (1-2\epsilon)\cD_0 + 2\epsilon \cD_1$. Note that $\cD_{\vecu,\vecv}=(1-2\epsilon)\cD_0 + 2\epsilon \cD_2$.

We give an informal idea now, before giving the (potentially notationally complex) details. 
The rough idea is that Alice and Bob first pad their inputs with lots of dummy variables (whose values are known to both) and expand the masks from $\cD_Y$ (or $\cD_N$) into masks that are from $\cD_1$ (respectively $\cD_2$). They then augment the sequence of masks from $\alpha n'$ to $\alpha n = \Omega(\alpha n'/\epsilon)$, injecting many random masks from $\cD_0$. This gives them an instance of $(\cD,\cD_{\vecu,\vecv})$-\textsf{RMD} to solve for which they use the protocol $\Pi$. It is not too hard to see all this can be done locally by Alice and Bob; and this is proved formally below.

\paragraph{Step 2: A reduction from $(\cD_Y,\cD_N)$-\textsf{RMD} to $(\cD,\cD_{\vecu,\vecv})$-\textsf{RMD}.}
Consider a protocol $\Pi = (\Pi_A,\Pi_B)$ for $(\cD,\cD_{\vecu,\vecv})$-\textsf{RMD} with parameter $\alpha \leq 1/(200k)$ using $C(n)$ bits of communication to achieve an advantage of $\delta$ on instances of length $n$. We let $n' = (k'\epsilon/k)n$ where $k'$ was chosen in the previous step. We also let $\alpha' = (2k/k')\alpha$ so that $\alpha' \leq 1/(100k')$. We use $\Pi$ to design a protocol $\Pi'$ for
$(\cD_Y,\cD_N)$-\textsf{RMD} with parameter $\alpha'$ achieving advantage of at least $\delta/2$ on instances of length $n'$ with communication $C'(n') = C(n)$. We conclude by \autoref{thm:communication lb 1 wise} that there exists a constant $\tau'$ such that $C(n) \geq \tau'\sqrt{n'} = \tau\sqrt{n}$, where $\tau = \tau' \sqrt{\epsilon k'/k} > 0$ as desired.

Our protocol $\Pi'$ uses shared randomness between Alice and Bob (while we assume $\Pi$ is deterministic). Let $n'' = kn'/k'$ so that $n = n''/(2\epsilon)$. Let $\alpha'' = \alpha'n'/n'' = k\alpha/k'$.
Recall that an instance of $(\cD_Y,\cD_N)$-\textsf{RMD} is determined by a four tuple $(\vecx',M',\vecz',\vecb')$ with
$\vecx' \in \{-1,1\}^{n'}$, $M' \in \{0,1\}^{k'\alpha' n' \times n'}$ and $\vecz',\vecb' \in \{-1,1\}^{k'\alpha' n'}$ with $\vecz' = M'\vecx' \odot \vecb'$. See~\autoref{fig:polarization indis 1} for a pictorial description.

\begin{figure}[ht]
    \centering
    \includegraphics[width=14cm]{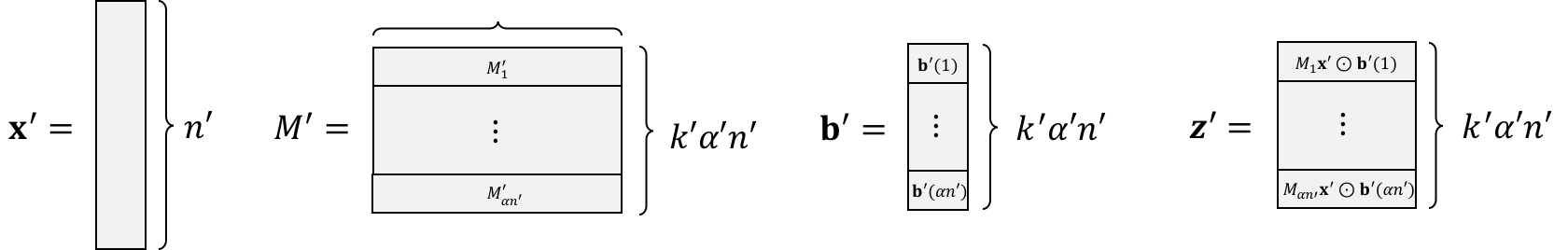}
    \caption{Pictorial description of $(\vecx',M',\vecb',\vecz')$.}
    \label{fig:polarization indis 1}
\end{figure}

We give two maps using shared randomness $R'$ and $R''$:
\begin{enumerate}[label=(\roman*)]
\item \textbf{From $(\cD_Y,\cD_N)$-\textsf{RMD} to $(\cD_1,\cD_2)$-\textsf{RMD}}: $(\vecx',M',\vecb',\vecz',R') \mapsto (\vecx'',M'',\vecb'',\vecz'')$ where $\vecx'' \in \{0,1\}^{n''}$,
$M'' \in \{0,1\}^{k\alpha'' n'' \times n''}$ and $\vecb'',\vecz'' \in \{-1,1\}^{k\alpha'' n''}$.
\item \textbf{From $(\cD_1,\cD_2)$-\textsf{RMD} to $(\cD,\cD_{\vecu,\vecv})$-\textsf{RMD}}: $(\vecx'',M'',\vecb'',\vecz'',R'') \mapsto (\vecx,M,\vecb,\vecz)$, where $\vecx \in \{0,1\}^{n}$,
$M \in \{0,1\}^{k\alpha n \times n}$ and $\vecb,\vecz \in \{-1,1\}^{k\alpha n}$.
\end{enumerate}
Before describing the two maps, let us first state the desired conditions.

\begin{reduction}{Success conditions for the reduction}
\begin{enumerate}[label=(\arabic*)]
\item \textbf{The reduction is locally well-defined.} Namely, there exist random strings $R'$ and $R''$ so that (i) Alice can get $\vecx$ through the maps $(\vecx',R') \mapsto \vecx''$ and $(\vecx'',R'') \mapsto \vecx$ while Bob can get $(M,\vecz)$ through the maps $(M',\vecz',R') \mapsto (M'',\vecz'')$ and $(M'',\vecz'',R'') \mapsto (M,\vecz)$.

\item \textbf{The reduction is sound and complete.} Namely, (i) $\vecz'' = M''\vecx'' \odot \vecb''$ and $\vecz = M\vecx \odot \vecb$. (ii) If $\vecb' \sim \cD_Y^{\alpha'n'}$ then $\vecb'' \sim \cD_1^{\alpha''n''}$ and $\vecb \sim \cD^{\alpha n}$. Similarly if $\vecb' \sim \cD_N^{\alpha'n'}$ then $\vecb'' \sim \cD_2^{\alpha''n''}$ and $\vecb \sim \cD_{\vecu,\vecv}^{\alpha n}$. (iii) $\vecx''\sim \textsf{Unif}(\{-1,1\}^{n''})$, $\vecx\sim \textsf{Unif}(\{-1,1\}^{n})$ and $M$ is a uniformly random matrix conditioned on having exactly one ``$1$'' per row and at most one ``$1$'' per column.
\end{enumerate}
\end{reduction}

In~\autoref{claim:polarization indis map 1} and~\autoref{claim:polarization indis map 2} we show that the above conditions hold except for an error event that occurs with tiny ($\exp(-n)$) probability. For now, let us show that these conditions imply the success of the reduction. Assuming conditions (1) and (2) the rest is simple. Alice computes $\vecx$ from $\vecx',R'$ and $R''$ and sends $m = \Pi_A(\vecx)$ to Bob, who computes $(M,\vecz)$ from $M',\vecz',R'$ and $R''$ and outputs $\Pi_B(m,M,\vecz)$. Conditions (1)-(2) combined with the bound on the error event imply that if $\Pi$ has advantage $\delta$ then $\Pi'$ has advantage at least $\delta - \exp(-n) \geq \delta/2$ as desired.

In the rest of this subsection, we describe the two maps and show that they satisfy the described success conditions. We wrap up the reduction and the proof of~\autoref{lem:polarization indis} in the end.

\paragraph{Step 3: Specify and analyze the first map.}
We now turn to specifying the maps mentioned above and proving that they satisfy conditions (1)-(2). We start with $(\vecx',M',\vecb',\vecz',R') \mapsto (\vecx'',M'',\vecb'',\vecz'')$. For this part, we let $R' \sim \textsf{Unif}(\{-1,1\}^{n''-n'})$. We set $\vecx'' = (\vecx',R')$. To get $M''$, $\vecz''$ and $\vecb''$ we need some more notations. First, note that $\alpha'n'=\alpha''n''$ due to the choice of parameters. Next, note that $M''$ can be viewed as the stacking of matrices $M'_1,\ldots,M'_{\alpha' n'} \in \{0,1\}^{k' \times n'}$. We first extend $M'_i$ by adding all-zero columns at the end to get $N''_i \in \{0,1\}^{k' \times n''}$. We then stack $N''_i$ on top of $P''_i \in \{0,1\}^{(k-k')\times n''}$ to get $M''_i$, where $(P''_i)_{j\ell} = 1$ if and only if $\ell = n' + (i-1)k + j$. See~\autoref{fig:polarization indis 3} for a pictorial description of $N''_i$ and $P''_i$. We let $M''$ be the stacking of $M''_1,\ldots,M''_{\alpha'' n''}$. Next we turn to $\vecb''$. Let $\vecb' = (\vecb'(1),\cdots,\vecb'(\alpha'' n''))$. Let $\tilde{\vecu} = (u_{k'+1},\ldots,u_{k})$ denote the common parts of $\vecu$ and $\vecv$. We let
$\vecb''(i) = (\vecb'(i),\tilde{\vecu})$ and $\vecb'' = (\vecb''(1),\cdots,\vecb''(\alpha'' n''))$. Finally we let $\vecz'' = M''\vecx'' \odot \vecb''$ as required. See~\autoref{fig:polarization indis 2} for a pictorial description.

\begin{figure}[ht]
    \centering
    \includegraphics[width=14cm]{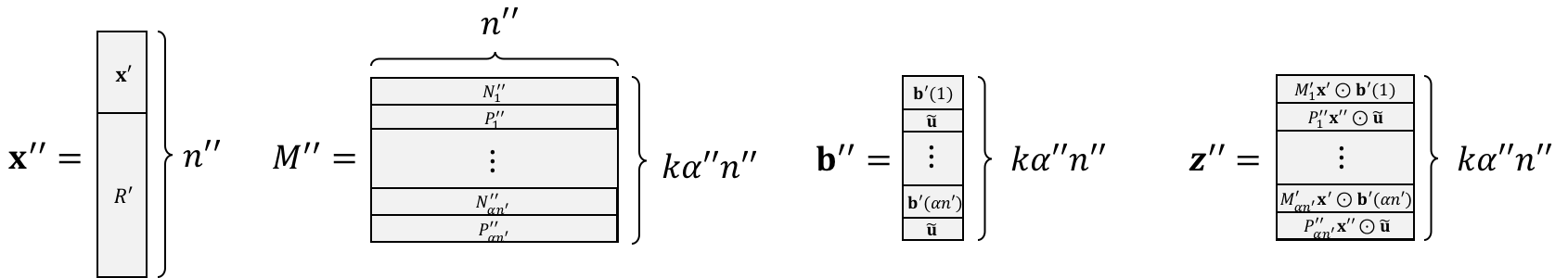}
    \caption{Pictorial description of $(\vecx'',M'',\vecb'',\vecz'')$.}
    \label{fig:polarization indis 2}
\end{figure}

Now, we verify that the first map satisfies the success conditions mentioned above.

\begin{claim}\label{claim:polarization indis map 1}
The first map in the reduction is locally well-defined, sound, and complete.
\end{claim}
\begin{proof}
To see that the first map is locally well-defined, note that Alice can compute $\vecx''= (\vecx',R')$ locally. Similarly, Bob can compute $M''$ locally by construction. As for $\vecz''$, note that $\vecz''$ interleaves (in a predetermined order) the bits of $\vecz'$ and those of $(P_i\vecx'' \odot \tilde{\vecu})_{i\in [\alpha n']}$. Furthermore $P_i\vecx''$ depends only on $R'$ (since the first $n'$ columns of all $P_i$s are zero). Thus Bob can locally compute $P_i\vecx''$ for every $i$, and since $\tilde{\vecu}$ is also known Bob can compute $\vecz''$ locally.

To see the first map is sound and complete, (i) $\vecz''=M''\vecx''\odot\vecb''$ follows from the construction. As for (ii), for each $i\in[\alpha'n']=[\alpha''n'']$, if $\vecb'_i\sim\cD_Y=\unif(\{\veca,-\veca\})$, then $\vecb''_i\sim\unif(\{(\veca,\tilde{\vecu}),(-\veca,\tilde{\vecu})\})$. Note that $\veca$ is chosen to be the uncommon part of $\vecu$ and $\vecv$ and hence $(\veca,\tilde{\vecu})=\vecu$ and $(-\veca,\tilde{\vecu})=\vecv$. Thus, $\vecb''_i\sim\unif(\{\vecu,\vecv\})=\cD_1$ as desired. Similarly, one can show that if $\vecb'_i\sim\cD_N$, then $\vecb''_i\sim\cD_2$. Finally, we have $\vecx''\sim\unif(\{-1,1\}^{n''})$ by construction and hence (iii) holds.

This completes the proof of conditions (1)-(2) for the first step of the reduction.
\end{proof}

\begin{figure}[ht]
    \centering
    \includegraphics[width=15cm]{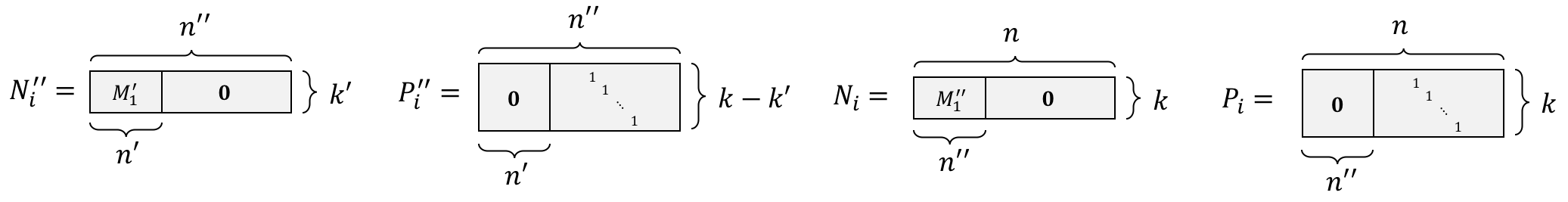}
    \caption{Pictorial description of $N''_i,P''_i,N_i,P_i$.}
    \label{fig:polarization indis 3}
\end{figure}

\paragraph{Step 4: Specify and analyze the second map.}
We now turn to the second map. Here $R''$ will be composed of many smaller parts which we introduce now.
Let $\vecy \sim \textsf{Unif}(\{-1,1\}^{n-n''})$, $\vecw \sim \bern(2\eps)^{\alpha n}$. Let $\Gamma \in \{0,1\}^{n\times n}$ be a uniform permutation matrix. Let $\vecc = (\vecc(1),\ldots,\vecc((n-n'')/k))$ where $\vecc(i) \sim \cD_0$ are chosen independently. We let $R'' = (\vecy,\vecw,\Gamma,\vecc)$. 
Let $\#_w(i) = |\{j \in [i]\, |\, w_j = 1\}|$ denote the number $1$'s among the first $i$ coordinates of $\vecw$. 
If $\#_w(\alpha n) \geq \alpha'' n''$ or if $\alpha n - \#_w(\alpha n) \geq (n-n'')/k$
we declare an error, Note $\Exp[\#_w(n)] = \alpha''n''/2$ so the probability of error is negligible (specifically it is $\exp(-n)$).

We now define the elements of $(\vecx,M,\vecb,\vecz)$.
We set $\vecx = \Gamma (\vecx'',\vecy)$ so $\vecx$ is a random permutation of the concatenation of $\vecx''$ and $\vecy$. Next, let $M'' = (M''_1,\ldots,M''_{\alpha''n''})$ where $M''_i \in \{0,1\}^{k \times n''}$. We extend $M''_i$ to $N_i \in \{0,1\}^{k \times n}$ by adding all-zero columns to the right. For $i \in \{1,\ldots,(n-n'')/k\}$, let $P_i \in \{0,1\}^{k \times n}$ be given by $(P_i)_{j\ell} =1$ if and only if $\ell = n'' +(i-1)k + j$. See~\autoref{fig:polarization indis 3} for a pictorial description of $N_i$ and $P_i$. Next we define a matrix $\tilde{M} \in \{0,1\}^{k\alpha n \times n} = (\tilde{M}_1,\ldots,\tilde{M}_{\alpha n})$ where $\tilde{M}_i \in \{0,1\}^{k \times n}$ is defined as follows: If $w_i = 1$ then we let $\tilde{M}_i = N_{\#_w(i)}$ else we let $\tilde{M}_i = P_{i - \#_w(i)}$. Finally we let $M = \tilde{M} \cdot \Gamma^{-1}$. Next we turn to $\vecb$. Again let $\vecb'' = (\vecb''(1),\ldots,\vecb''(\alpha''n''))$. We let $\vecb = (\vecb(1),\ldots,\vecb(\alpha n))$ where $\vecb(i)$ is defined as follows: If $w_i = 1$ then $\vecb(i)= \vecb''(\#_w(i))$, else $\vecb(i)= \vecc(i-\#_w(i))$. Finally, $\vecz = M\vecx \odot \vecb$. See~\autoref{fig:polarization indis 4} for a pictorial description. This concludes the description of the map and we turn to analyzing its properties. 

\begin{figure}[ht]
    \centering
    \includegraphics[width=15cm]{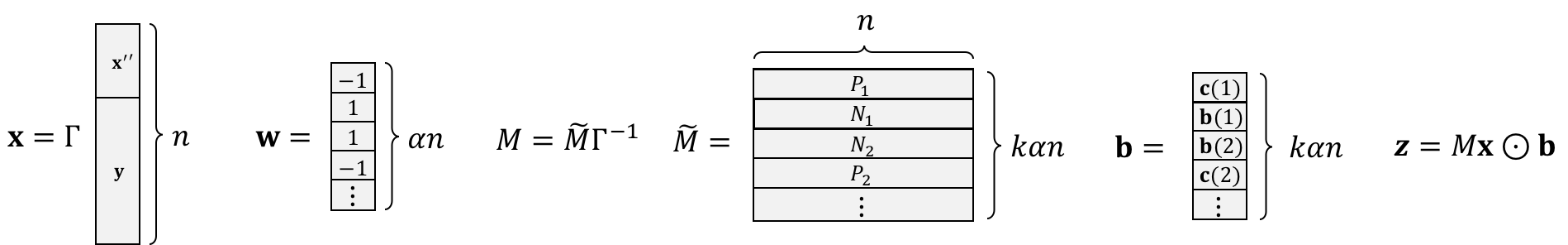}
    \caption{Pictorial description of $\vecx,\vecw,M,\vecb,\vecz$.}
    \label{fig:polarization indis 4}
\end{figure}

Now, we verify that the first map satisfies the success conditions mentioned above.
\begin{claim}\label{claim:polarization indis map 2}
If $\#_w(\alpha n) \leq \alpha'' n''$ and  $\alpha n - \#_w(\alpha n) \leq (n-n'')/k$, then the second map in the reduction is locally well-defined, sound, and complete. In particular, the error event happens with probability at most $\exp(-\Omega(n))$ over the randomness of $R''$.
\end{claim}
\begin{proof}
To see that the second map is locally well-defined, first note that Alice can compute $\vecx=\Gamma(\vecx'',\vecy)$ from $\vecx''$ and the shared randomness $R''$ locally. As for Bob, note that the maximum index needed for $N$ and $\vecb''$ (resp. $P$ and $\vecc$) is at most $\#_w(\alpha n)$ (resp. $\alpha n-\#_w(i)$). Namely, if $\#_w(\alpha n) \leq \alpha'' n''$ and $\alpha n - \#_w(\alpha n) \leq (n-n'')/k$, then $M$ and $\vecb$ are well-defined. Also, using similar argument as in the proof of~\autoref{claim:polarization indis map 1}, one can verify that $M$ and $\vecb$ can be locally computed by $M''$, $\vecb''$, and the shared randomness $R''$.

To see the second map is sound and complete, (i) $\vecz=M\vecx\odot\vecb$ directly follows from the construction. As for (ii), if $\vecb'\sim\cD_Y^{\alpha'n'}$, from~\autoref{claim:polarization indis map 1} we know that $\vecb''\sim\cD_1^{\alpha'n'}=\unif(\{\vecu,\vecv\})^{\alpha'n'}$. Now, for each $i\in[\alpha n]$, $\vecb(i)=\vecb''(\#_w(i))$ with probability $2\epsilon$ and $\vecb(i)=\vecc(i-\#_w(i))$ with probability $1-2\epsilon$. As $\vecb''(i')\sim\cD_1$ for every $i'\in[\alpha'n']$ and $\vecc(i'')\sim\cD_0$ for every $i''\in[(n-n'')/k]$, we have $\vecb(i)\sim(1-2\epsilon)\cD_0+2\epsilon\cD_1=\cD$ as desired. Similarly, one can show that for every $i\in[\alpha'n']=[\alpha''n'']$, if $\vecb'(i)\sim\cD_N^{\alpha'n'}$, then $\vecb(i)\sim\cD_{\vecu,\vecv}$. Finally, we have $\vecx\sim\unif(\{-1,1\}^n)$ and $M$ is a uniformly random matrix with exactly one ``$1$'' per row and at most one ``$1$'' per column (due to the application of a random permutation $\Gamma$) by construction.

This completes the proof of conditions (1)-(2) for the second step of the reduction.
\end{proof}

\paragraph{Step 5: Proof of~\autoref{lem:polarization indis}.}

\begin{proof}[Proof of~\autoref{lem:polarization indis}]
Let us start with setting up the parameters. Given $k, \alpha\in(0,\alpha_0),n,\cD$, and incomparable pair $(\vecu,\vecv)\in\supp(\cD)$ and polarization amount $\epsilon=\epsilon(\cD,\vecu,\vecv)$, let $k'=|\{i\in[k]\, |\, u_i\neq v_i\}|$, $n'=(k'\eps/k)n$, $\alpha'=(2k/k')\alpha$, $n''=kn'/k'$, $\alpha''=\alpha'n'/n''$, and $\delta'=\delta/2$.

Now, for the sake of contradiction, we assume that there exists a protocol $\Pi=(\Pi_A,\Pi_B)$ for $(\cD,\cD_{\vecu,\vecv})$-\textsf{RMD} with advantage $\delta$ and at most $\tau\sqrt{n}$ bits of communication.

First, observe that $n-n''=(1-\epsilon)n$ and $\alpha''n''=2\epsilon\alpha n$. As $\vecw\sim\bern(\epsilon)^{\alpha n}$, we have $\#_w(\alpha n) \leq \alpha'' n''$ and $\alpha n - \#_w(\alpha n) \leq (n-n'')/k$ with probability at least $1-\exp(-\Omega(n))$. Thus, combine with~\autoref{claim:polarization indis map 1} and~\autoref{claim:polarization indis map 2}, if $(\vecx',M',\vecz')$ is a Yes (resp. No) instance of $(\cD_Y,\cD_N)$-\textsf{RMD}, then the output of the reduction, i.e., $(\vecx,M,\vecz)$, is a Yes (resp. No) instance of $(\cD,\cD_{\vecu,\vecv})$-\textsf{RMD} with probability at least $1-\exp(-\Omega(n))$. Moreover,~\autoref{claim:polarization indis map 1} and~\autoref{claim:polarization indis map 2} also show that the reduction can be implemented locally and hence Alice and Bob can run the protocol $\Pi$ on $(\vecx,M,\vecz)$. In particular, Alice and Bob computes $\vecx$ and $(M,\vecz)$ using their inputs and shared randomness respectively. Then, Alice sends $m=\Pi_A(\vecx)$ to Bob and Bob outputs $\Pi_B(m,M,\vecz)$. By the correctness of the reduction as well as that of the protocol, we know that Alice and Bob have advantage at least $\delta-\exp(-\Omega(n))\geq\delta/2=\delta'$ in solving $(\cD_Y,\cD_N)$-\textsf{RMD} with at most $\tau\sqrt{n}=\tau\sqrt{(k/(k'\epsilon))n'}$ bits of communication.

Finally, by~\autoref{thm:communication lb 1 wise}, we know that there exists a constant $\tau_0>0$ such that any protocol for $(\cD_Y,\cD_N)$-\textsf{RMD} with advantage $\delta'$ requires at least $\tau_0\sqrt{n'}$ bits of communication. This implies that $\tau\geq\tau_0\sqrt{k'\epsilon/k}$. We conclude that any protocol for $(\cD,\cD_{\vecu,\vecv})$-\textsf{RMD} with advantage $\delta$ requires at least $\tau\sqrt{n}$ bits of communication.
\end{proof}

\subsection{Finite upper bound on the number of polarization steps}
In this section we prove that there is a finite upper bound on the number of polarization steps needed to move from a distribution $\cD \in \Delta(\{-1,1\}^k)$ to the canonical distribution with marginal $\vecmu(\cD)$, i.e., $\cD_{\vecmu(\cD)}$. Together with the indistinguishability result from \autoref{lem:polarization indis} this allows us to complete the proof of \autoref{thm:communication lb matching moments} by going from $\cD_Y$ to $\cD_{\vecmu(\cD_Y)}=\cD_{\vecmu(\cD_N)}$ and then to $\cD_N$ by using the triangle inequality for indistinguishability.

\newcommand{\vecA}{\mathbf{A}}
\newcommand{\cfk}{\cF(\{-1,1\}^k)}

In this section we extend our considerations to functions $A:\{-1,1\}^k \to \R^{\geq 0}$. Let $\cfk = \{A:\{-1,1\}^k \to \R^{\geq 0}\}$.
For $A \in \cfk$, let $\mu_0(A) = \sum_{\veca\in\{-1,1\}^k} A(\veca)$. 
Note $\Delta(\{-1,1\}^k) \subseteq \cfk$ and $A \in \Delta(\{-1,1\}^k)$ if and only if $A\in\cfk$ and $\mu_0(A) = \sum_{\veca \in \{-1,1\}^k} A(\veca) = 1$. We extend the definition of marginals, support, canonical distribution, potential and polarization operators to $\cfk$. In particular we let $\vecmu(A) = (\mu_0,\mu_1,\ldots,\mu_k)$ where $\mu_0 = \mu_0(A)$ and $\mu_j = \sum_{\veca \in \{-1,1\}^k} a_j A(\veca)$ for $j \in [k]$. We also define canonical function and polarization operators so as to preserve $\vecmu(A)$. So given arbitrary $A$, let $\cD = \frac{1}{\mu_0(A)}\cdot A$. Note $\cD\in \Delta(\{-1,1\}^k)$. For $\vecmu  = (\mu_0,\mu_1,\ldots,\mu_k) \in \R^{k+1}$, we define $A_{\vecmu} = \mu_0\cdot \cD_{\vecmu'}$ where $\vecmu' = (\mu_1/\mu_0,\ldots,\mu_k/\mu_0)$ to be the canonical function associated with $\vecmu$. 
We remark that by~\autoref{lem:polarization} and~\autoref{lem:potential polarized},  $A_{\vecmu(A)}$ is the unique function such that (i) it has the same marginals as $A$ and (ii) it supports a chain.

\begin{definition}[Polarization length]\label{def:polarization-length}
For distribution $A\in\cfk$, let $N(A)$ be the smallest~$t$ such that there exists a sequence $\vecA = A_0,A_1,\ldots,A_t$ such that $A_0 = A$, $A_t = A_{\vecmu(A)}$ is canonical and for every $i \in [t]$ it holds that there exists incomparable  $\vecu_i,\vecv_i \in \supp(A_{i-1})$ such that $A_i = (A_{i-1})_{\vecu_i,\vecv_i}$. If no such finite sequence exists then let $N(A)$ be infinite. Let $N(k) = \sup_{A\in\cfk}  \{N(A)\}$. Again, if $N(A) = \infty$ for some $A$ or if no finite upper bound exists, $N(k)$ is defined to be $\infty$. 
\end{definition}

Note that if $\cD \in \Delta(\{-1,1\}^k)$ so is every element in the sequence, so the polarization length bound below applies also to distributions.
Our main lemma in this subsection is the following:

\begin{lemma}[A finite upper bound on $N(k)$]\label{lem:polarization finite}
$N(k)$ is finite for every finite $k$. Specifically $N(k) \leq (k^2+3)(1+ N(k-1))$.
\end{lemma}

We prove~\autoref{lem:polarization finite} constructively in the following four steps.

\paragraph{Step 1: Description of the algorithm {\sc Polarize}.}
Let us start with some notations. For $A\in\cF(\{-1,1\}^k)$ we let $A|_{x_\ell = b}$ denote the function $A$ restricted to the subcube $\{-1,1\}^{\ell-1} \times \{b\} \times \{-1,1\}^{k-\ell}$. Note that $A$ restricted to subcubes is effectively a $(k-1)$-dimensional function and we will use this reduction in dimension in our recursive algorithm. 

\begin{algorithm}[H]
	\caption{$\textsc{Polarize}(\cdot)$}
	\label{alg:polarization}
    \begin{algorithmic}[1]
	    \Input $A\in\cF(\{-1,1\}^k)$.
		\If{k=2}
		    \State {\bf Output:} $A_{(-1,1),(1,-1)}$.
		\EndIf
		\State $(A_0)|_{x_k = -1} \gets $ {\sc Polarize}$(A|_{x_k = -1})$ ; $(A_0)|_{x_k = 1} \gets $ {\sc Polarize}$(A|_{x_k = 1})$ ; $t \gets 0$.
		\State Let $(-1)^k = \veca_t(0) < \cdots < \veca_t(k-1) = (1^{k-1},-1)$ be a chain supporting $(A_t)|_{x_k = -1}$.
        \State Let $((-1)^{k-1},1) = \vecb_t(0) < \cdots < \vecb_t(k-1) = 1^k$ be a chain supporting $(A_t)|_{x_k = 1}$.
		\While{$\exists (i,j)$ with $j < k-1$ such that $\veca_t(i) \vee \vecb_t(j) = 1^k$ and $A_t(\veca_t(i)), A_t(\vecb_t(j))>0$}
		    \State Let $(i_t,j_t)$ be the lexicographically smallest such pair $(i,j)$.
		    \State $B_t \gets (A_t)_{\veca_t(i_t),\vecb_t(j_t)}$.
		    \State $(A_{t+1})|_{x_k = -1} \gets $ {\sc Polarize}($B_t|_{x_k = -1}$) ; $(A_{t+1})|_{x_k = 1} \gets (B_t)|_{x_k = 1}$.
		    \State $t\gets t+1$.
		    \State Let $(-1)^k = \veca_t(0) < \cdots < \veca_t(k-1) = (1^{k-1},-1)$ be a chain supporting $(A_t)|_{x_k = -1}$.
            \State Let $((-1)^{k-1},1) = \vecb_t(0) < \cdots < \vecb_t(k-1) = 1^k$ be a chain supporting $(A_t)|_{x_k = 1}$.
		\EndWhile
		\State Let $\ell \in [k]$ be such that for every $\veca \in \{-1,1\}^k \setminus \{1^k\}$ we have $A_t(\veca)>0 \Rightarrow a_\ell = -1$.
		\State $(A_{t+1})|_{x_\ell = -1} \gets ${\sc Polarize}$(A_t)|_{x_\ell = -1}$. $(A_{t+1})|_{x_\ell = 1} \gets (A_{t})|_{x_\ell = 1}$.
		\State {\bf Output:} $A_{t+1}$.
	\end{algorithmic}
\end{algorithm}

The goal of the rest of the proof is to show that~\autoref{alg:polarization} terminates after a finite number of steps and outputs $A_{\vecmu(A)}$.

\paragraph{Step 2: Correctness assuming {\sc Polarize} terminates.}

\begin{claim}[Correctness condition of {\sc Polarize}]\label{claim:polarization correctness}
For every $A\in\cF(\{-1,1\}^k)$, if $\textsc{Polarize}$ terminates, then $\textsc{Polarize}(A)=A_{\vecmu(A)}$. In particular, $\textsc{Polarize}(A)$ has the same marginals as $A$ and is supported on a chain.
\end{claim}
\begin{proof}
First, by the definition of the polarization operator (\autoref{def:polarization operator}), the marginals of $A_t$ are the same for every $t$. So in the rest of the proof, we focus on inductively showing that if {\sc Polarize} terminates, then $\textsc{Polarize}(A)$ is supported on a chain.

For the base case where $k=2$, we always have $\textsc{Polarize}(A)=A_{(-1,1),(1,-1)}$ supported on a chain as desired. 

When $k>2$, note that when the algorithm enters the Clean-up stage, if we let $m$ and $n$ denote the largest indices such that $A_t(\veca_t(m)), A_t(\vecb_t(n)) > 0$ and $ A_t(\vecb_t(n))\neq 1^k$, then the condition that $\veca_t(m)\vee \vecb_t(n) \ne 1^k$ implies that there is a coordinate $\ell$ such that $\veca_t(m)_\ell = \vecb_t(n)_\ell = -1$. Since every $\vecc$ such that $A_t(\vecc) > 0$ and $c_k = -1$ satisfies $\vecc \leq \veca_t(m)$, we have $A_t(\vecc) > 0$ implies $c_\ell = -1$. Similarly for every $\vecc\ne 1^k$ such that $c_k = 1$, we have $A_t(\vecc)>0$ implies $c_\ell = -1$.
We conclude that $A_t$ is supported on $\{1^k\} \cup \{\vecc\, |\, c_\ell = -1\}$. 
Thus, by the induction hypothesis, after polarizing the subcube $x_\ell = -1$ and leaving the subcube $x_\ell = 1$ unchanged, we get that the resulting function $A_{t+1}$ is supported on a chain as desired and complete the induction. We conclude that if {\sc Polarize} terminates, we have $\textsc{Polarize}(A)=A_{\vecmu(A)}$.
\end{proof}

\paragraph{Step 3: Invariant in {\sc Polarize}.} 
Now, in the rest of the proof of~\autoref{lem:polarization finite}, the goal is to show that for every input $A$, the number of iterations of the while loop in~\autoref{alg:polarization} is finite.
The key claim (\autoref{clm:key-terminate}) here asserts that the sequence of pairs $(i_t,j_t)$ is monotonically increasing in lexicographic order. Once we establish this claim, it follows that there are at most $k^2$ iterations of the while loop and so $N(k)  \leq (k^2 + 3)\cdot (1 + N(k-1))$, proving \autoref{lem:polarization finite}.
Before proving \autoref{clm:key-terminate}, we establish the following properties that remain invariant after every iteration of the while loop.

\begin{claim}\label{clm:chain-supp}
For every $t\ge0$, we have  $\forall b\in\{-1,1\}$, $(A_t)|_{x_k=b}$ is supported on a chain.
\end{claim}

\begin{proof}
For $b = -1$, the claim follows from the correctness of the recursive call to {\sc Polarize}. For $b=1$, we claim by induction on $t$ that the supporting chain $\vecb_t(0) < \cdots < \vecb_t(k-1)$ never changes (with $t$). To see this, note that $\vecb_t(k-1)=1^k$ is the only point in the subcube $\{x_k = 1\}$ that increases in value compared to $A_t$, and this is already in the supporting chain. Thus $\vecb_t(0) < \cdots < \vecb_t(k-1)$ continues to be a supporting chain for $(A_{t+1})|_{x_k=1}$. 
\end{proof}

For $\vecc \in \{-1,1\}^k$, we say that a function $A:\{-1,1\}^k \to \R^{\geq0}$ is {\em $\vecc$-subcube-respecting} ($\vecc$-respecting, for short) if for every $\vecc'$ such that $A(\vecc') > 0$, we have $\vecc' \geq \vecc$ or $\vecc' \leq \vecc$. We say that $A$ is {\em $\vecc$-downward-respecting} if $A$ is $\vecc$-respecting and the points in the support of $A$ above $\vecc$ form a partial chain, specifically, if $\vecu,\vecv > \vecc$ have $A(\vecu),A(\vecv) > 0$ then either $\vecu \geq \vecv$ or $\vecv \geq \vecu$. 

Note that if $A$ is supported on a chain then $A$ is $\vecc$-respecting for every point $\vecc$ in the chain. Conversely, if $A$ is supported on a chain and $A$ is $\vecc$-respecting, then $A$ is supported on a chain that includes $\vecc$. 

\begin{claim}[Polarization on subcubes]\label{clm:subcubes}
Let $A$ be a $\vecc$-respecting function and let $\tilde{A}$ be obtained from $A$ by a finite sequence of polarization updates, as in \autoref{def:polarization operator}.
Then $\tilde{A}$ is also $\vecc$-respecting. Furthermore if $A$ is $\vecc$-downward-respecting and $\vecw > \vecc$ then 
$\tilde{A}$ is also $\vecc$-downward-respecting and $A(\vecw) = \tilde{A}(\vecw)$. 
\end{claim}

\begin{proof}
Note that it suffices to prove the claim for a single update by a polarization operator since the rest follows by induction. So let $\tilde{A} = A_{\vecu,\vecv}$ for incomparable $\vecu,\vecv\in \supp(A)$.

Since $A$ is $\vecc$-respecting, and $\vecu,\vecv$ are incomparable, either $\vecu \leq \vecc,\vecv \leq \vecc$ or $\vecu \geq \vecc, \vecv \geq \vecc$. Suppose the former is true, then $\vecu \vee \vecv \leq \vecc$ and 
$\vecu \wedge \vecv \leq \vecc$, and hence, $\tilde{A}$ is $\vecc$-respecting. Similarly, in the case when $\vecu \geq \vecc, \vecv \geq \vecc$, we can show that $\tilde{A}$ is $\vecc$-respecting. 
The furthermore part follows by noticing that for $\vecu$ and $\vecv$ to be incomparable if $A$ is $\vecc$-downward-respecting and $A(\vecu),A(\vecv)>0$, then $\vecu,\vecv \leq \vecc$, and so the update changes $A$ only at points below $\vecc$. 
\end{proof}

The following claim asserts that in every iteration of the while loop, by the lexicographically minimal choice of $(i_t,j_t)$, there exists a coordinate $h\in[k-1]$ such that every vector $c<a_t(i_t)$ in the support of $A_t$, $B_t$, or $A_{t+1}$ has $c_h=-1$, and every vector $c\neq 1^k$ in the support of $(A_t)|_{x_k = 1}$ has $c_h=-1$.

\begin{claim}\label{claim:polarization At Bt support}
For every $t\ge 0$, $\exists h \in [k-1]$ such that $\forall \vecc \in \{-1,1\}^k$, if $\vecc\in \supp(A_t)\cup\supp(B_t)\cup\supp(A_{t+1})$, then the following hold:
\begin{itemize}
    \item If $\vecc<\veca_t(i_t)$, then $c_h=-1$.
    \item If $c_k=1$ and $\vecc\neq1^k$, then $c_h=-1$.
\end{itemize}
\end{claim}
\begin{proof}
Since $(i_t,j_t)$ is lexicographically the smallest incomparable pair in the support of $A_t$, for $i < i_t$, $j < k-1$, and 
$A_t(\veca(i)), A_t(\vecb(j)) > 0$, we have $\veca(i) \vee \vecb(j) \ne 1^k$. Let $m$ be the largest index smaller than $i_t$ such that
$A_t(\veca_t(m)) > 0$. Similarly, let $n < k-1$ be the largest index such that $A_t(\vecb_t(n)) > 0$. Then the fact that $\veca_t(m) \vee \vecb_t(n) \ne 1^k$ implies that there exists $h\in[k-1]$ such that $\veca_t(m)_h = \vecb_t(n)_h = -1$. Now, using the fact (from \autoref{clm:chain-supp}) that $(A_t)|_{x_k=-1}$ is supported on a chain, we conclude that for every $\vecc < \veca_t(i_t)$, $A_t(\vecc)> 0$ implies that $\vecc \leq \veca_t(m)$ and hence, $c_h=-1$. Similarly, for every vector $\vecc \ne 1^k$ in the support of $(A_t)|_{x_k = 1}$, by the maximality of $n$, we have $c_h = -1$.

We now assert that the same holds for $B_t$. First, recall that $\supp(B_t)\subset\supp(A_t)\cup\{1^k,\veca_t(i_t)\wedge\vecb_t(j_t)\}$ since $B_t=(A_t)_{\veca_t(i_t),\vecb_t(j_t)}$.
Next, note that the only point (other than $1^k$) where $B_t$ is larger than $A_t$ is
$\veca_t(i_t) \wedge \vecb_t(j_t)$. It suffices to show that $(\veca_t(i_t) \wedge \vecb_t(j_t))_h = -1$. We have $\veca_t(i_t) \wedge \vecb_t(j_t) \leq \vecb_t(j_t) \leq \vecb_t(n)$ and hence $(\veca_t(i_t) \wedge \vecb_t(j_t))_h = -1$.

Finally, we assert that same holds also for $A_{t+1}$. Since $A_{t+1}|_{x_k=1} = B_t|_{x_k=1}$, the second item in the claim follows trivially. To prove the first item, let us consider $\veca' \in \{-1,1\}^k$ defined as follows: $\veca'_h = -1$ and $\veca'_{r} = \veca_t(i_t)_{r}$ for $r \ne h$. Note that $B_t|_{x_k = -1}$ is $\veca_t(i_t)$-respecting since potentially the only new point in its support (compared to $A_t|_{x_k = -1}$) is $\veca_t(i_t) \wedge \vecb_t(j_t) \leq \veca_t(i_t)$. From the previous paragraph we also have that if $B_t(\vecc) > 0$ and $\vecc < \veca_t(i_t)$, then $c_h = -1$ and hence, $\vecc \leq \veca'$. On the other hand, if $B_t(\vecc) > 0$ and $\vecc \ge \veca_t(i_t)$, then $\vecc\ge a'$. Therefore, $B_t|_{x_k=-1}$ is $\veca'$-respecting. By applying \autoref{clm:subcubes}, we conclude that $(A_{t+1})|_{x_k=-1}$ is also $\veca'$-respecting. It follows that if $\vecc < \veca(i_t)$ and $A_{t+1}(\vecc) > 0$, then $\vecc \leq \veca'$ and so $c_h = -1$.
\end{proof}

\paragraph{Step 4: Proof of~\autoref{lem:polarization finite}.}

The following claim establishes that the while loop in the $\textsc{Polarize}$ algorithm terminates after a finite number of iterations.

\begin{claim}\label{clm:key-terminate}
For every $t\ge 0$, $(i_t,j_t)<(i_{t+1},j_{t+1})$ in lexicographic ordering.
\end{claim}
\begin{proof}

Consider the chain $\veca_{t+1}(0) < \cdots < \veca_{t+1}(k-1)$ supporting $A_{t+1}|_{x_k = -1}$. Note that for  $i \geq i_t$, $A_{t+1}|_{x_k = -1}$ is $\veca_t(i)$-respecting (since $A_t|_{x_k=-1}$ and $B_{t}|_{x_k=-1}$ were also so). In particular, $A_t|_{x_k=-1}$ is $\veca_t(i)$-respecting because it is supported on a chain containing $a_t(i)$. Next $B_t|_{x_k = -1}$ is $\veca_t(i)$-respecting since potentially the only new point in its support is $\veca_t(i_t) \wedge \vecb_t(j_t) \leq \veca_t(i)$. Finally, $A_{t+1}|_{x_k = -1}$ is also $\veca_t(i)$-respecting using \autoref{clm:subcubes}. Thus we can build a chain containing  $\veca_t(i)$ that supports $A_{t+1}|_{x_k=-1}$. It follows that we can use $\veca_{t+1}(i) = \veca_t(i)$ for $i \geq i_t$. Now consider $i < i_t$. We must have $\veca_{t+1}(i) < \veca_{t+1}(i_t) = \veca_t(i_t)$.
By~\autoref{claim:polarization At Bt support}, there exists $h\in [k-1]$ such that for $i < i_t$, $\veca_{t+1}(i)_h = -1$.

We now turn to analyzing $(i_{t+1},j_{t+1})$. Note that by definition, $A_{t+1}(\veca_{t+1}(i_{t+1}))>0$ and 
$A_{t+1}(\vecb_{t+1}(b_{t+1}))>0$. First, let us show that $i_t \le i_{t+1}$. On the contrary, let us assume that $i_{t+1} < i_t$. It follows from the above paragraph that $\veca_{t+1}(i_{t+1})_h = -1$. Also, for every $\vecb_{t+1}(j)$ with $j <k-1$ and $A_{t+1}(\vecb_{t+1}(j))>0$, we have $\vecb_{t+1}(j)_h = -1$. Therefore, $\veca(i_{t+1}) \vee \vecb(j_{t+1}) \ne 1^k$ (in particular $(\veca(i_{t+1}) \vee \vecb(j_{t+1}))_h=-1$), which is a contradiction.

Next, we show that if $i_{t+1} = i_t$, then $j_{t+1}\ge j_t$. By the minimality of $(i_t,j_t)$ in the $t$-th round, for $j<j_t$ such that $A_t(b_t(j))>0$, we have $a_{t}(i_{t})\vee b_t(j) \ne 1^k$. Since  $i_{t+1} = i_t$, $a_{t+1}(i_{t+1})=a_{t+1}(i_{t})=a_{t}(i_{t})$. We already noted in the proof of \autoref{clm:chain-supp} that $\vecb_t(0) < \cdots < \vecb_t(k-1)$ is also a supporting chain for $(A_{t+1})|_{x_k=1}$. The only point where the function $A_{t+1}|_{x_k=1}$ has greater value than $A_t|_{x_k=1}$ is $1^k$. Therefore, for $j<j_t$ such that $A_{t+1}(b_{t+1}(j))>0$, we have $a_{t+1}(i_{t+1})\vee b_{t+1}(j) \ne 1^k$ and hence, $j_{t+1}\ge j_t$.

So far, we have established that $(i_{t+1},j_{t+1})\ge (i_t,j_t)$ in lexicographic ordering. Finally, we will show that $(i_{t+1},j_{t+1}) \neq (i_t,j_t)$ by proving that at least one of $A_{t+1}(\veca_{t+1}(i_{t}))$ and $A_{t+1}(\vecb_{t+1}(j_t))$ is zero. The polarization update ensures that at least one of $B_{t}(\veca_{t}(i_{t}))$ and $B_{t}(\vecb_t(j_t))$ is zero. If $B_t(\vecb_t(j_t)) = 0$, then by definition, we have $A_{t+1}(\vecb_{t+1}(j_t)) = A_{t+1}(\vecb_t(j_t)) = 0$. Finally to handle the case $B_{t}(\veca_{t}(i_{t}))=0$, let us again define $\veca'$ as: $\veca'_h = -1$ and $\veca'_{r} = \veca_t(i_t)_{r}$ for $r \ne h$, where $h$ is as given by \autoref{claim:polarization At Bt support}. We assert that $B_t|_{x_k=-1}$ is $\veca'$-downward-respecting. As shown in the proof of \autoref{claim:polarization At Bt support}, we have $B_t|_{x_k=-1}$ is $\veca'$-respecting. The support of $B_t|_{x_k=-1}$ is contained in $\{\veca_t(0),\cdots,\veca_t(k-1)\} \cup \{\veca_t(i_t) \wedge \vecb_t(j_t)\}$ and $\veca_t(i_t) \wedge \vecb_t(j_t) < \veca_t(i_t)$, and by \autoref{claim:polarization At Bt support}, $\veca_t(i_t) \wedge \vecb_t(j_t)\leq \veca'$. It follows that $B_t|_{x_k=-1}$ is $\veca'$-downward-respecting. Finally, by the furthermore part of \autoref{clm:subcubes} applied to $B_{t}|_{x_k=-1}$ and $\vecw = \veca_t(i_t)$, we get that $A_{t+1}(\veca_{t+1}(i_{t}))=A_{t+1}(\veca_t(i_t)) = B_{t}(\veca_t(i_t)) = 0$. It follows that $(i_{t+1},j_{t+1}) \ne (i_t,j_t)$. 
\end{proof}

\begin{proof}[Proof of~\autoref{lem:polarization finite}]
By~\autoref{claim:polarization correctness}, we know that if~\autoref{alg:polarization} terminates, we have $\textsc{Polarize}(A)=A_{\vecmu(A)}$. Hence, the maximum number of polarization updates used in {\sc Polarize} (on input from $\cF(\{-1,1\}^k)$) serves as an upper bound for $N(k)$.
By~\autoref{clm:key-terminate}, we know that there are at most $k^2$ iterations of the while loop and so $N(k) \leq (k^2 + 3)\cdot (1 + N(k-1))$ as desired.
\end{proof}

\subsection{Putting it together}

We now have the ingredients in place to prove \autoref{thm:communication lb matching moments}.

\begin{proof}[Proof of \autoref{thm:communication lb matching moments}]
Given distribution $\cD_Y,\cD_N$ with $\vecmu = \vecmu(\cD_Y) = \vecmu(\cD_N)$, first we apply \autoref{lem:polarization finite} to $\cD_Y$ to get $\cD_0 = \cD_Y,\cD_1,\ldots,\cD_t = \cD_{\vecmu}$ such that $\cD_{i+1} = (\cD_i)_{\vecu(i),\vecv(i)}$, i.e., $\cD_i$ is an update of $\cD_i$, with $t \leq N(k) < \infty$. 
Similarly, we apply \autoref{lem:polarization finite} to $\cD_N$ to get $\cD'_0 = \cD_N,\cD'_1,\ldots,\cD'_{t'} = \cD_{\vecmu}$ such that $\cD'_{i+1} = (\cD'_i)_{\vecu'(i),\vecv'(i)}$ with $t' \leq N(k) < \infty$. 

Now, \autoref{lem:polarization indis}, applied to the pairs $\cD_i$ and $\cD_{i+1}$ with $\delta' = \delta/(2N(k))$, gives is $\tau_i$ such that every protocol for 
$(\cD_i,\cD_{i+1})$-\textsf{RMD} requires $\tau_i \sqrt{n}$ bits of communication to achieve advantage $\delta'$. 
Similarly applying \autoref{lem:polarization indis} again with $\delta' = \delta/(2N(k))$ to the pairs $\cD'_i$ and $\cD'_{i+1}$, we get $\tau'_i$ such that every protocol for 
$(\cD'_i,\cD'_{i+1})$-\textsf{RMD} requires $\tau'_i \sqrt{n}$ bits of communication to achieve advantage $\delta'$. 

Letting $\tau = \min\left\{\min_{i \in [t]}\{\tau_i\}, \min_{i \in [t']}\{\tau'_i\}\right\}$, we get, using the triangle inequality for indistinguishability, that every protocol $\Pi$ for $(\cD_Y,\cD_N)$-\textsf{RMD} achieving advantage $\delta \geq (t+t')\delta'$ requires $\tau\sqrt{n}$ communication.
\end{proof}

\section*{Acknowledgments}

 We are grateful to Lijie Chen, Gillat Kol, Dmitry Paramonov, Raghuvansh Saxena, Zhao Song, and Huacheng Yu, for detecting a fatal error in an earlier version of this paper~\cite{CGSV20} and then for pinpointing the location of the error. As a result the main theorem of the current paper is significantly different than the theorem claimed in the previous version.

Thanks to Johan H\aa stad for many pointers to the work on approximation resistance and answers to many queries.
Thanks to Dmitry Gavinsky, Julia Kempe and Ronald de Wolf for prompt and detailed answers to our queries on outdated versions of their work~\cite{GKKRW}. Thanks to Prasad Raghavendra for answering our questions about the approximation resistance dichotomy from his work~\cite{raghavendra2008optimal}. Thanks to Saugata Basu for the pointers to the algorithms for quantified theory of the reals. Thanks to Jelani Nelson for pointers to $\ell_1$ norm estimation algorithms. 

Thanks to Michael Hwang and Tarun Prasad for pointing out some errors in Example~1 in a previous version of this paper. Thanks to Noah Singer for pointing out some typos in the paper.
\bibliographystyle{alpha}
\bibliography{mybib}

\newcommand{\etalchar}[1]{$^{#1}$}
\begin{thebibliography}{GKK{\etalchar{+}}09}

\bibitem[AKSY20]{assadi2020multi}
Sepehr Assadi, Gillat Kol, Raghuvansh~R Saxena, and Huacheng Yu.
\newblock {Multi-Pass Graph Streaming Lower Bounds for Cycle Counting, MAX-CUT,
  Matching Size, and Other Problems}.
\newblock In {\em FOCS 2020}, 2020.

\bibitem[AM09]{AustrinMossel}
Per Austrin and Elchanan Mossel.
\newblock Approximation resistant predicates from pairwise independence.
\newblock {\em Comput. Complex.}, 18(2):249--271, 2009.

\bibitem[BPR06]{BasuPR}
Saugata Basu, Richard Pollack, and Marie-Fran\c{c}oise Roy.
\newblock {\em Algorithms in Real Algebraic Geometry}.
\newblock Springer, 2006.

\bibitem[Bul17]{Bulatov}
Andrei~A. Bulatov.
\newblock A dichotomy theorem for nonuniform {CSP}s.
\newblock In Chris Umans, editor, {\em FOCS 2017}, pages 319--330. {IEEE},
  2017.

\bibitem[BV04]{boyd2004convex}
Stephen~P. Boyd and Lieven Vandenberghe.
\newblock {\em Convex optimization}.
\newblock Cambridge University Press, 2004.

\bibitem[CGSV21]{CGSV20}
Chi{-}Ning Chou, Alexander Golovnev, Madhu Sudan, and Santhoshini Velusamy.
\newblock Classification of the streaming approximability of {Boolean CSPs}.
\newblock {\em CoRR}, abs/2102.12351v1, 2021.

\bibitem[CGV20]{CGV20}
Chi-Ning Chou, Alexander Golovnev, and Santhoshini Velusamy.
\newblock Optimal streaming approximations for all {B}oolean {Max-2CSPs} and
  {Max-$k$SAT}.
\newblock In {\em FOCS 2020}. IEEE, 2020.

\bibitem[Cha20]{chakrabarti2020data}
Amit Chakrabarti.
\newblock Data stream algorithms.
\newblock {\em Lecture notes}, page~94, 2020.

\bibitem[GKK{\etalchar{+}}09]{GKKRW}
Dmitry Gavinsky, Julia Kempe, Iordanis Kerenidis, Ran Raz, and Ronald de~Wolf.
\newblock Exponential separation for one-way quantum communication complexity,
  with applications to cryptography.
\newblock {\em {SIAM} J. Comput.}, 38(5):1695--1708, 2009.

\bibitem[GM08]{guha2008tight}
Sudipto Guha and Andrew McGregor.
\newblock Tight lower bounds for multi-pass stream computation via pass
  elimination.
\newblock In {\em ICALP 2008}, pages 760--772. Springer, 2008.

\bibitem[GT19]{GT19}
Venkatesan Guruswami and Runzhou Tao.
\newblock Streaming hardness of unique games.
\newblock In {\em APPROX 2019}, pages 5:1--5:12. LIPIcs, 2019.

\bibitem[GVV17]{GVV17}
Venkatesan Guruswami, Ameya Velingker, and Santhoshini Velusamy.
\newblock Streaming complexity of approximating {Max 2CSP} and {Max Acyclic
  Subgraph}.
\newblock In {\em APPROX 2017}. LIPIcs, 2017.

\bibitem[Ind00]{Indyk}
Piotr Indyk.
\newblock Stable distributions, pseudorandom generators, embeddings and data
  stream computation.
\newblock In {\em FOCS 2000}, pages 189--197. IEEE, 2000.

\bibitem[Kho02]{Kho02}
Subhash Khot.
\newblock On the power of unique 2-prover 1-round games.
\newblock In {\em STOC 2002}, pages 767--775. ACM, 2002.

\bibitem[KK19]{KK19}
Michael Kapralov and Dmitry Krachun.
\newblock An optimal space lower bound for approximating {MAX-CUT}.
\newblock In {\em STOC 2019}, pages 277--288. ACM, 2019.

\bibitem[KKL88]{KKL88}
Jeff Kahn, Gil Kalai, and Nathan Linial.
\newblock The influence of variables on {Boolean} functions.
\newblock In {\em FOCS 1988}, pages 68--80. IEEE, 1988.

\bibitem[KKS15]{KKS}
Michael Kapralov, Sanjeev Khanna, and Madhu Sudan.
\newblock Streaming lower bounds for approximating {MAX-CUT}.
\newblock In {\em SODA 2015}, pages 1263--1282. {SIAM}, 2015.

\bibitem[KKSV17]{KKSV17}
Michael Kapralov, Sanjeev Khanna, Madhu Sudan, and Ameya Velingker.
\newblock $(1+ \omega(1))$-approximation to {MAX-CUT} requires linear space.
\newblock In {\em SODA 2017}, pages 1703--1722. SIAM, 2017.

\bibitem[KNW10]{KNW10}
Daniel~M. Kane, Jelani Nelson, and David~P. Woodruff.
\newblock On the exact space complexity of sketching and streaming small norms.
\newblock In {\em SODA 2010}, pages 1161--1178. SIAM, 2010.

\bibitem[KTW14]{KhotTW}
Subhash Khot, Madhur Tulsiani, and Pratik Worah.
\newblock A characterization of strong approximation resistance.
\newblock In {\em STOC 2014}, pages 634--643, 2014.

\bibitem[McG14]{mcgregor2014graph}
Andrew McGregor.
\newblock Graph stream algorithms: a survey.
\newblock {\em SIGMOD Record}, 43(1):9--20, 2014.

\bibitem[O'D14]{o2014analysis}
Ryan O'Donnell.
\newblock {\em Analysis of {B}oolean functions}.
\newblock Cambridge University Press, 2014.

\bibitem[Pot19]{Potechin}
Aaron Potechin.
\newblock On the approximation resistance of balanced linear threshold
  functions.
\newblock In Moses Charikar and Edith Cohen, editors, {\em {STOC} 2019}, pages
  430--441. {ACM}, 2019.

\bibitem[Rag08]{raghavendra2008optimal}
Prasad Raghavendra.
\newblock Optimal algorithms and inapproximability results for every {CSP}?
\newblock In {\em STOC 2008}, pages 245--254, 2008.

\bibitem[Sch78]{Schaefer}
Thomas~J. Schaefer.
\newblock The complexity of satisfiability problems.
\newblock In {\em STOC 1978}, pages 216--226. ACM, 1978.

\bibitem[Yao77]{yao1977probabilistic}
Andrew Chi-Chin Yao.
\newblock Probabilistic computations: Toward a unified measure of complexity.
\newblock In {\em FOCS 1977}, pages 222--227. IEEE, 1977.

\bibitem[Zhu17]{Zhuk}
Dmitriy Zhuk.
\newblock A proof of {CSP} dichotomy conjecture.
\newblock In {\em FOCS 2017}, pages 331--342. {IEEE}, 2017.

\end{thebibliography}
\end{document}